\documentclass[11pt]{article}
\usepackage[margin=1in]{geometry} 
\usepackage{amsmath,amsthm,amssymb, color}
\usepackage{enumerate} 
\usepackage[colorlinks=true, linkcolor=blue]{hyperref}
\usepackage{color}
\usepackage{multirow}
\usepackage{enumitem}
\usepackage{mathtools}
\usepackage{braket}
\usepackage{tikz}
\usepackage{caption}
\usepackage{subcaption}
\usepackage{graphicx} %
\usepackage[noend]{algpseudocode}
\usepackage{algorithm}
\usepackage{dsfont}
\usepackage{authblk}
\usepackage{amsfonts}
\usepackage{booktabs}
\usepackage{thmtools}
\usepackage{mathpazo}
\usepackage{array}

\newtheorem{theorem}{Theorem}[section]
\newtheorem{lemma}[theorem]{Lemma}
\newtheorem{corollary}[theorem]{Corollary}
\newtheorem{remark}[theorem]{Remark}
\newtheorem{definition}[theorem]{Definition}
\newtheorem{proposition}[theorem]{Proposition}
\newtheorem{fact}[theorem]{Fact}
\newtheorem{assumption}{Assumption}

\newcommand{\brak}[1]{\left\{#1\right\}}
\newcommand{\pmone}{\brak{-1, 1}}
\newcommand{\R}{\mathbb{R}}
\DeclarePairedDelimiter{\norm}{\lVert}{\rVert}
\newcommand{\mathify}[1]{\ifmmode{#1}\else\mbox{$#1$}\fi}
\newcommand{\abs}[1]{\mathify{\left| #1 \right|}}
\newcommand{\E}{\mathbb{E}}
\newcommand{\Var}{\mathrm{Var}}
\newcommand{\diag}{\mathrm{diag}}
\newcommand{\GSE}{\textup{\textsf{GSE}}}
\newcommand{\Ecal}{\mathcal{E}}
\newcommand{\Ncal}{\mathcal{N}}
\newcommand{\Ocal}{\mathcal{O}}
\newcommand{\Mcal}{\mathcal{M}}
\newcommand{\Xcal}{\mathcal{X}}

\newcommand{\expP}{\mathbb{E}_{y \underset{P}{\sim} x}}

\newcommand{\Estar}{E^\ast}
\newcommand{\zstar}{z^\ast}
\newcommand{\tmix}{t_{\mathrm{mix}}}
\newcommand{\Ostar}{\Ocal^\ast}

\renewcommand{\star}{\ast}

\renewcommand{\>}{\rangle}

\newcommand{\er}{Erd\H{o}s-R\'enyi}

\DeclareMathOperator{\poly}{poly}
\DeclareMathOperator{\Ent}{Ent} 
\DeclareMathOperator{\KL}{KL} 
\DeclareMathOperator{\TV}{TV} 
\DeclareMathOperator{\SWAP}{\textup{\textsf{SWAP}}}

\DeclareMathOperator*{\argmin}{arg\,min}

\usepackage{authblk}
\title{Generalized Short Path Algorithms: Towards Super-Quadratic Speedup over Markov Chain Search for Combinatorial Optimization}
\author{Shouvanik Chakrabarti\thanks{These authors contributed equally to this work, and are listed alphabetically. Correspondence should be addressed to \texttt{shouvanik.chakrabarti@jpmchase.com}.}}
\author{Dylan Herman\protect\footnotemark[1]}
\author{Guneykan Ozgul\protect\footnotemark[1]}
\author{Shuchen Zhu\protect\footnotemark[1]}
\author{Brandon~Augustino}
\author{Tianyi Hao}
\author{Zichang He}
\author{Ruslan Shaydulin}
\author{Marco Pistoia}
\affil{Global Technology Applied Research, JPMorganChase, New York, NY 10017, USA}

\date{}

\begin{document}

\maketitle

\begin{abstract}
    We analyze generalizations of quantum algorithms based on the short path framework first proposed by Hastings~[\textit{Quantum} 2, 78 (2018)], which has been extended and shown by Dalzell~et~al.~[STOC~'23] to achieve super-Grover speedups for certain binary optimization problems. We demonstrate that, under some commonly satisfied technical conditions, an appropriate generalization can achieve super-quadratic speedups not only over unstructured search but also over a classical optimization algorithm that searches for the optimum by drawing samples from the stationary distribution of a Markov chain. We employ this framework to obtain algorithms for problems including variants of Max Bisection, Max Independent Set, and finding the ground states of the Antiferromagnetic Ising Model and the Sherrington-Kirkpatrick Model, whose runtimes are asymptotically faster than those obtainable with previous short path techniques. In certain cases, our algorithms achieve super-quadratic speedups compared to the  best known classical algorithms with rigorously established runtimes.
   We conclude the paper with a numerical analysis that guides the choice of parameters for short path algorithms and raises the possibility of super-quadratic speedups in settings that are currently beyond our theoretical analysis.
\end{abstract}
\newpage
\tableofcontents
\newpage

\section{Introduction}
\label{sec:intro}

\subsection{Motivation}
\label{sec:motivation}
The prospect of quantum algorithmic speedups for combinatorial optimization has been heavily studied for more than two decades \cite{farhi2000quantum, farhi2014quantumapproximateoptimizationalgorithm,  hastings2018shortPath, montanaro2018backtracking, montanaro2020branchAndBound}. This interest is partially motivated by practical considerations, since combinatorial optimization problems are ubiquitous in scientific and industrial applications, and are a major source of computational bottlenecks \cite{abbas2024quantumoptimization, herman2023quantum, dalzell2023quantalgsurv}. A second principled motivation is that there are reasons to expect such a speedup, arising from the existence of quantum algorithms such as Grover's search algorithm~\cite{grover1996fast}, which enjoys a quadratic quantum speedup over brute force search. Since the best classical algorithms with rigorously provable runtimes for combinatorial optimization often reduce to (possibly structured) search over a large space of possible solutions, one may expect algorithms building on Grover search to provide speedups for these algorithms. In recent years, this intuition has been largely confirmed with the development of quantum accelerated versions that offer quadratic speedups for backtracking \cite{montanaro2018backtracking, ambainis2017treeSizeEstimation} and branch-and-bound \cite{montanaro2020branchAndBound,chakrabarti2022universal}, two of the main classical meta-algorithms used to obtain provable runtime guarantees. There has been some success in obtaining sub-quadratic speedups for combinatorial optimization algorithms based on Markov chain Monte Carlo (MCMC) \cite{wocjan2008speedupQSampling}, and dynamic programming~\cite{ambainis2019dynamicProgramming} methods.

Despite this progress, there remain challenges towards leveraging quantum algorithms for combinatorial optimization. Firstly, there are many problems for which quadratic quantum speedups over the state-of-the-art classical approaches have not been demonstrated, including cases where the best algorithm is based on dynamic programming~\cite{fomin2006branching}, MCMC, and local search~\cite{schoning2002probabilistic, schoning2013satisfiability}. The second challenge is more fundamental, as recent research has identified challenges towards the realization of quadratic quantum speedups due to constant-factor slowdowns compared to classical hardware manifesting from slower clock speeds, the overhead of error correction, and the limited parallelizability of quantum algorithms. Realistic estimates for the resources required to execute a quantum algorithm on a scale where it can break-even with classical computing can result in quantum runtimes exceeding many days. The viability of practical realization of a polynomial speedup increases with the degree of the speedup. For example, the resource analysis in \cite{babbush2021focus}  indicates that the outlook for realizing a quartic speedup, when considering all overheads, can be much more realistic, requiring quantum runtimes on the scale of hours instead of days. It is thus of fundamental importance to investigate whether it is possible to obtain ``super-quadratic" speedups for combinatorial optimization.

The general quantum speedups for backtracking~\cite{montanaro2018backtracking}, branch-and-bound~\cite{montanaro2020branchAndBound}, dynamic programming~\cite{ambainis2019dynamicProgramming} and MCMC methods either rely directly on Grover search, or closely related methods like amplitude estimation, quantum minimum finding, or discrete time quantum walks. These frameworks therefore admit at most quadratic speedups by construction. Furthermore, in the case of unstructured search, backtracking and branch-and-bound, quadratic speedups can be shown to be the best one can hope for in the oracle setting. The study of super-quadratic speedups necessitates the investigation of mechanisms for quantum speedup beyond Grover search. It is also likely that these speedups must leverage problem specific structure to circumvent the aforementioned lower bounds in the oracle setting.

A simpler, but very nontrivial question, is whether there are quantum algorithms for combinatorial optimization that achieve \textit{super-Grover} speedups. That is, a super-quadratic speedup over unstructured search, but not necessarily the best classical algorithm. An important first step towards  rigorously obtaining positive results in this direction was the quantum \emph{short path} algorithm from Hastings~\cite{hastings2018shortPath}, which was demonstrated to solve combinatorial optimization problems with runtime $\Ostar\left(2^{(0.5 - c(n))n}\right)$\footnote{We use the notation $\Ocal^*(2^{h(n)})$ to indicate an upper bound on the runtime of the form $\Ocal(\poly(n) 2^{h(n)})$.} , where $c(n)$ is a positively valued function of $n$. A limitation of this result was that the rigorously established bounds on $c(n)$ asymptote to $0$ as $n$ increases, leading to sub-polynomial improvements over Grover search. In a followup work that built on the framework of Hastings (but using a significantly modified algorithm and analysis), Dalzell et al.~\cite{dalzell2022mind} gave an algorithm that obtains \textit{strictly} super-Grover speedups for several combinatorial optimization problems, i.e., the algorithm achieves a runtime $\Ocal^*\left(2^{(0.5 - c)n}\right)$, where $c$ is a positive parameter independent of $n$. Our work largely builds on the algorithm of Dalzell et al.~\cite{dalzell2022mind}, which we henceforth refer to as an improved short path algorithm due to its connection with Hastings' original work.

The successful demonstration of super-Grover speedups leads to natural optimism that similar techniques could be used in principle to show super-quadratic speedups over the best known classical algorithm. There remain several challenges towards such a demonstration. On one hand, the speedups shown in~\cite{dalzell2022mind} are only larger than quadratic by very small factors.  On the other hand, the best performing classical algorithms for well-studied combinatorial optimization problems (while still exponential-time) are significantly faster than unstructured search. For instance, even when restricting ourselves to algorithms with provable runtimes, the well-known 3-SAT problem can be solved in time $\Ocal^*\left(1.32065^{n} \right)$~\cite{hertli2010improving,paturi2005improved}, the maximum independent set of an $n$-vertex graph can be found in time $\Ocal^*(1.1996^{n})$~\cite{xiao2017exact}, and the exact ground state of the Sherrington-Kirkpatrick model can be determined in time $\Ocal^*\left(1.3670^{n}\right)$~\cite{montanaro2020branchAndBound}. As a consequence, the runtimes established in~\cite{dalzell2022mind} are in most cases slower than the best classical algorithm. An exception to this is the problem of minimizing the energy of $k$-spin models, for which there has been limited study of classical algorithms. It is important to note that the mathematical analysis of~\cite{dalzell2022mind} does not make much effort to optimize the parameters of the algorithm, and the actual runtime is predicted to be better than the theoretical predictions. However, we provide results from a numerical simulation of the algorithm in~\cite{dalzell2022mind} which indicate that, even when the parameters are optimized, the speedup is likely to be insufficient by itself. For maximum independent set (MIS) problem on graphs with up to 21 vertices, the best runtime scaling (using a penalty term to enforce the constraints) achieved in our experiments is around $\Ocal^*\left(1.3195^{n}\right)$. 
This is slower than the best classical algorithm despite optimizing the parameters, indicating that the frameworks in \cite{hastings2018shortPath, dalzell2022mind} must be further generalized if genuine super-quadratic speedups are the goal.

A second consideration is that many combinatorial optimization problems of industrial importance are constrained, including well-studied examples such as finding maximum independent set, maximum/minimum bisection, vertex and set cover, portfolio optimization, and Hamiltonian cycles. The current short path algorithms can only incorporate constraints by adding penalty terms to the cost function, an approach rarely used by the state-of-the-art algorithms for combinatorial optimization. Furthermore, the runtime of the short path algorithm scales with the number of bit-strings defined on the unconstrained solution space, which can often be much larger than the number of bit-strings that satisfy all constraints in the problem formulation. As an example, consider a combinatorial optimization problem with a constraint requiring solutions to have Hamming weight $\lfloor n^{\alpha} \rfloor$ for some $0< \alpha < 1$. The number of \textit{feasible} bit-strings is about $2^{(1 - \alpha) n^{\alpha} \log(n)}$, which is asymptotically smaller than $2^n$ by a significant margin. For such constrained problems, the short path algorithm cannot even offer a super-quadratic speedup over unstructured search (if the search is restricted to the feasible region).

This paper seeks to address these limitations by analyzing a generalized version of short path algorithms that, under some technical conditions, obtain super-quadratic speedups over classical algorithms that search the space of solutions using samples from the stationary distribution of a Markov chain. We refer to such algorithms as \textit{Markov chain search (MCS)}, and the framework can be simply described\footnote{We refer the reader to Section~\ref{sec:preliminaries} for background on Markov chains.} as follows. Suppose that our aim is to minimize a real-valued cost function  $H \colon \Xcal \rightarrow \R$ for a finite set $\Xcal \subset \{-1,1\}^{n}$, and let $P$ be the transition matrix of a Markov chain that mixes to a stationary distribution $\pi$ supported on $\Xcal$ in $\poly(n)$ steps. A Markov chain search algorithm using $P$ simply runs the chain to draw samples from $\pi$, and keeps track of the running minimum of the samples in terms of the cost function $H$. This minimum (and the corresponding sample) serves as an estimate of the global minimum (and minimizer) of $H$ on $\Xcal$. To bound the expected runtime of this algorithm, it is sufficient to bound the expected number of samples before a global minimum is encountered. Letting $\pi^\ast$ denote the total probability that a sample from $\pi$ is a global minimizer of $H$, it follows that drawing $\Ocal\left((\pi^\ast)^{-1} \log(\epsilon^{-1})\right)$ samples from $\pi$ suffices to ensure that we encounter the global minimizer with probability at least $1 - \epsilon$. We outline this framework in Algorithm~\ref{alg:MarkovChainSearch}.

Markov chain search is a natural extension of unstructured search, with several advantages. Firstly, it is often possible to classically sample from distributions over the feasible set that favor lower cost solutions. A common example is sampling from the Gibbs distribution corresponding to the cost function $H$, where a solution $z \in \Xcal$ is sampled with probability proportional to $\exp(-\beta H(z))$ for some parameter $\beta$ (usually called the inverse temperature). Gibbs distributions are the stationary distributions of well known chains such as the Glauber dynamics \cite{glauber1963time} (also referred to as Metropolis sampling), and the mixing times of these Markov chains have been extensively studied in physics and computer science for decades. It is evident that in a Gibbs distribution with $\beta > 0$, low cost solutions are more likely than in the uniform distribution. If the global minima has significantly lower cost than most of the ensemble, this can lead to algorithms that are polynomially faster than unstructured search. We note that Dalzell et al.~\cite{dalzell2022mind} make reference to precisely this framework when examining the possibility of faster classical algorithms for $k$-spin problems.  A second advantage of Markov chain search arises when the feasible set is asymptotically much smaller than $2^{n}$ as discussed previously. Clearly, if there is a chain $\mathcal{M}$ that mixes in polynomial time to the uniform distribution over $|\Xcal|$, the Markov chain search algorithm with $\mathcal{M}$ is faster than unstructured search over the unconstrained space. Even if we can only prepare a (not necessarily uniform) distribution whose support is restricted to $\Xcal$, this may result in a runtime that is substantially better than unstructured search. While the classical analysis of precise runtimes using Markov chain search is typically challenging, we later give explicit examples of settings where this separation from unstructured search is realized.

In both our generalized framework and \cite{dalzell2022mind}, the advantage over quadratic comes from producing an intermediate quantum state, commonly known as a \emph{guiding state}, that has an improved overlap with the optimal solutions than the initial state does. The initial state will have a measurement distribution, in the computational basis, that can be efficiently replicated classically. This intermediate state then acts as the starting point for Grover's algorithm \cite{grover1996fast} or more specifically quantum minimum finding \cite{durr1999quantumalgorithmfindingminimum}. Hence, if a classical algorithm can also replicate the improved overlap attained by the intermediate state, then the overall quantum speedup would only be quadratic. Hence, we refer  to this ``partial dequantization'' as a \emph{Groverization} of the short-path algorithm. It is currently unclear whether or not the applications of previous short-path frameworks can be Groverized.

\subsection{Contributions}
\label{sec:contribution}
Our primary contribution is the formulation of the generalized short path algorithm and the demonstration that, under certain technical conditions, it obtains a super-quadratic speedup over Markov chain search. In particular we obtain quantum runtimes of $\Ostar\left(T(n) ^{0.5 - c(n)}\right)$ where $T(n)$ is the (exponential in $n$) runtime of the classical Markov chain search, and $c(n) > 0$ is a positive parameter. When $c(n)$ is bounded below by a constant $c > 0$ we say we have a true super-quadratic speedup, and when $c(n)>0$ for $n<\infty$ but  $c(n) = o(1)$ we say we have an \emph{asymptotically decaying} advantage over quadratic speedup (in the vein of Hastings' results~\cite{hastings2018shortPath,hastings2018weaker,hastings2019shortPathToyModel}). We go on to discuss how these conditions may be established and demonstrate explicit results. We can further decompose the implications of our main result into two major contributions.

\paragraph{Contribution 1 -- Quantum Improvements over Non-trivial Classical Algorithms :}
The chosen applications highlight two major advantages of our framework. These properties enable speedups over \emph{non-trivial} classical algorithms and new quantum runtimes, both of which would not be possible with the existing short path frameworks \cite{hastings2018shortPath, dalzell2022mind}. 
\begin{enumerate}
    \item \textbf{Handling Hard Constraints}: As a first  example of constrained optimization, we focus on optimization over strings of fixed Hamming weight, a problem we call \emph{MaxCut Hamming}, using a Markov chain based on random transpositions. In this case, Markov chain search reduces simply to a brute force search over \textit{feasible} strings, i.e., those that satisfy the Hamming weight constraint. %
    We note that the new runtime by itself is not particularly interesting as the number of feasible solutions (i.e., the number of $n$-bit binary strings of Hamming weight $n/2$) is smaller than $2^n$ by a polynomial factor and thus, the runtime for Markov chain search is not notably faster than unstructured search for this problem. The second constrained problem we consider is finding  the ground state of the \emph{hardcore model}, i.e. \emph{Maximum Independent Set}, where we utilize a constraint-preserving Glauber dynamics chain. For  random regular graphs of sufficiently high degree, the Glauber chain asymptotically outperforms brute-force random search over all $2^n$ bit-strings\footnote{However, due to known computational barriers in Gibbs sampling, unlike the Maxcut Hamming case, we are unable to sample from the uniform distribution over feasible strings.}.
    \item %
    \textbf{Non-uniform Prior}:
    Our next application is to Markov chain search for unconstrained problems but with non-uniform stationary distributions, particularly to search with Gibbs distributions prepared using the Glauber dynamics\footnote{While the hardcore model is a constrained problem that also falls into this category, the inability to sample from positive inverse-temperatures due to computational barriers implies that the advantage MCS has over bruteforce search comes from restricting to the space of independent sets. Thus, it is less about the non-uniformity of the prior.}. The problems we consider in this setting are the \emph{antiferromagnetic Ising model} and the \emph{Sherrington-Kirkpatrick model}. In both cases, the non-uniform Gibbs sampler outperforms algorithms that start from a uniform distribution.
\end{enumerate}
 The following theorem  informally summarizes our results for the above applications.

\begin{theorem}[Applications of Generalized Short Path Framework (informal Theorem \ref{thm:formal_speedup})]
\label{thm:informal_speedup}
    The the following optimization problems exhibit a quantum runtime of $\Ostar\left(T(n)^{0.5 - c(n)}\right)$, where $T(n)$ is the runtime of a classical algorithm based on Markov chain search, and $c(n) > 0$ quantifies the degree of improvement over quadratic speedup:
    \begin{enumerate}
        \item solving a generalization of MaxBisection, which we call \ref{e:MaxCut}, over  Erd\H{o}s-R\'enyi random graphs with constant average degree,
        \item finding the ground state of the hardcore model, i.e. maximum independent set, on graphs with constant maximum degree, 
        \item finding the ground state of the antiferromagnetic Ising model on graphs of constant maximum degree, 
        \item and finding the ground state of the Sherrington-Kirkpatrick model.
    \end{enumerate}
\end{theorem}
Additionally, we provide some partial evidence that it may be possible to eventually show that some of the super-quadratic speedups attained in this work cannot be Groverized.

\paragraph{Contribution 2 -- Technical Improvements with Generalized Short-Path Framework :}  From a technical point of view, while the skeleton of the framework follows that of~\cite{dalzell2022mind} quite closely, the intermediate conditions must all be uplifted to incorporate the Markov chain used for the base classical algorithm, and significant care is needed to obtain the generalized results. These generalizations are crucial to leverage the two main advantages of Markov chain search over brute force search, the ability to sample from non-uniform distributions and distributions with support restricted to feasible solutions. As we will later show, these cannot be achieved by specializations of existing frameworks. The uplift of the conditions also clarifies some aspects of the role they play in the argument. We found that approaching the original results of~\cite{dalzell2022mind} from this new perspective led to some additional insights, which may be implicit, but are not explicitly documented in other works. We note also that the algorithm of~\cite{dalzell2022mind} follows directly as a special case of our generalized framework by considering the bit-flip walk (or the lazy, random walk on the vertices of the hypercube). Finally, the generalized analysis in this paper allows us to greatly simplify some of the statistical mechanics arguments made in~\cite{dalzell2022mind}, and instead rely directly on some standard results from the theory of Markov chains and this may be of separate interest.

\subsection{Related Works}
As discussed at length in the paragraph above, the primary inspiration for our work comes from the earlier short path algorithms developed in the series of works~\cite{hastings2018shortPath,hastings2019shortPathToyModel,hastings2018weaker,dalzell2022mind}. We now discuss other related quantum algorithms and their connections to our results. The most closely related family of algorithms are based on discrete time quantum walks. The framework of Magniez at al~\cite{magniez2011search} considers a framework that closely matches the one considered here. In~\cite{magniez2011search}, the authors seek to accelerate a classical algorithm, that first prepares in time $\mathrm{S}$ a sample from the stationary distribution of a Markov chain with spectral gap $\delta$, that has $\epsilon$ overlap with some marked state that can be checked in time $\mathrm{C}$, and then runs the Markov chain to repeatedly draw and check samples from the stationary distribution. The classical algorithm finds a marked element in time $\mathcal{O}\left(\mathrm{S} + \epsilon^{-1}(\delta^{-1}\mathrm{U} + \mathrm{C})\right)$ where $\mathrm{U}$ is the cost of simulating one step of the chain. The quantum algorithm obtains a runtime of $\mathcal{O}\left(\mathrm{S} + \epsilon^{-1/2}(\delta^{-1/2}\mathrm{U} + \mathrm{C})\right)$. 

In our setting $\epsilon$ corresponds to $\pi(\Estar)$ and falls exponentially, whereas $\mathrm{S, U, C},\delta^{-1}$ all grow polynomially. Applying the~\cite{magniez2011search} framework, we obtain a quantum runtime of $\Ostar(\pi(\Estar)^{-0.5})$. If the conditions for our framework are met, we obtain a runtime of  $\Ostar(\pi(\Estar)^{-0.5 + c(n)})$. In this setting, our algorithm accelerates the~\cite{magniez2011search} framework in a manner analogous to how~\cite{dalzell2022mind} accelerates Grover search. Another common framework for analyzing search via quantum walk is that of Szegedy~\cite{szegedy2004quantum} where the quadratic speedup is over the hitting time of a marked vertex. However, since the Hitting Time $\mathrm{HT}$ from stationarity satisfies $\epsilon^{-1} \le \mathrm{HT} \le \epsilon^{-1}\delta^{-1}$,  we have $\mathrm{HT} = \Ostar(\epsilon^{-1})$ in our setting where $\epsilon^{-1}$ is exponentially larger than $\delta^{-1}$, and the runtime of both quantum walk frameworks match up to polynomial factors. 

Aside from the works we have already mentioned, quantum algorithms have also been successfully leveraged to obtain speedups for solving linear systems of equations~\cite{harrow2009quantum, childs2017quantum, chakraborty2018power}, computing partition functions~\cite{harrow2020adaptive, cornelissen2023sublinear}, estimating the volume of convex bodies~\cite{chakrabarti2023quantum}, as well as sampling from both log-concave \cite{childs2022quantum} and non-logconcave \cite{ozgul2024stochastic} distributions. We note that our work is primarily concerned with exponential time search algorithms instead of the randomized approximation schemes considered in these papers. It would be interesting to understand whether our methods can be used to obtain super-quadratic speedups for exponential time counting algorithms~\cite{goldberg2021faster}.  Along the line of super-quadratic quantum speedups, there is also recent work providing positive results for Tensor Principle Component Analysis \cite{schmidhuber2024quartic}, and approximation algorithms for combinatorial optimization~\cite{jordan2024optimization}.

\subsection{Organization}
\label{sec:organization}
The rest of the paper is organized as follows. Section~\ref{sec:preliminaries} introduces some basic notation, background on Markov chains, and a refresher on the original short path algorithm. A reader familiar with Markov chains may simply skip to Section~\ref{sec:overview}, which provides a technical introduction to the generalized short path framework and our main theoretical results. This section also includes some discussion on the techniques and the consequences of these results. The following sections contain the technical proofs, with results about the general framework in Section~\ref{sec:framework} and results about the applications to specific problems in Section~\ref{sec:applications-complete}.

\section{Preliminaries}
\label{sec:preliminaries}
We write $\log$ and $\ln$ to indicate logarithms base 2 and base $e$, respectively. We denote the $i$-th element of a vector $x \in \mathbb{C}^{n}$ by $x_i$, and the $ij$-th element of a matrix $A \in \mathbb{C}^{m \times n}$ by $A_{ij}$. For a vector $x \in \mathbb{C}^n$, the matrix $\diag(x) \in \mathbb{C}^{n \times n}$ takes the values of $x$ on its diagonal and zero elsewhere. We write $A \succeq 0$ ($A \succ 0$) to indicate that a matrix $A \in \mathbb{C}^{n \times n}$ is positive semidefinite (positive definite), i.e., all of its eigenvalues are nonnegative (positive). For two $m \times n$ matrices $A$ and $B$, we write $A \circ B$ to indicate their Hadamard (or, element-wise) product.
Note that when we say the phrase \emph{with high probability} (w.h.p. for short), we imply that a result holds asymptotically in the problem size $n$ with probability $1$. 

We will use $\mathcal{G}_{\text{RG}}(n, d)$ to denote the uniform random ensemble of regular graphs on $n$ vertices with degree $d$. Also, $\mathcal{G}_{\text{ER}}(n, p)$ will be used to denote the Erd\H{o}s-R\'enyi ensemble, where there are $n$ vertices and an edge is created between two vertices with probability $p$.

\subsection*{Order Estimates}
We define $\Ocal (\cdot)$ as
$$f(x) = \Ocal(g(x)) \iff \exists \ell \in \R{}, \alpha \in \R{}_+,~\text{such that}~f(x) \leq \alpha g(x)\quad \forall x > \ell.$$
Similarly, 
$$f(x) = o (g(x)) \iff \exists \ell \in \R{},~\text{such that}~f(x) < \alpha g(x)\quad \forall x > \ell~\text{and}~\forall \alpha \in \R{}_+.$$
We write $f(x) = \Omega (g(x)) \iff g(x) = \Ocal(f(x))$. If there exists positive constants $\alpha_1$ and $\alpha_2$ such that  
$$\alpha_1 g(x) \leq f(x) \leq \alpha_2 g(x) \quad \forall x > 0,$$
then we write $f(x) = \Theta (g(x))$.

We also define $\widetilde{\Ocal} (f(x)) =\Ocal(f(x) \cdot \textup{polylog}(x))$ and $\Ocal^* (2^{f(x)}) =\Ocal(2^{f(x)} \cdot \poly (x))$.

\subsection{Probability Toolbox}
In what follows let $\mathcal{X}$ be a finite set satisfying $|\mathcal{X}| = V$. %
A \textit{Markov chain} $\mathcal{M}$ is a random process that defines movements between elements of $\mathcal{X}$. Transitions between states are determined according to a fixed probability distribution, and can be represented by an $V \times V$ (though not necessarily symmetric) \textit{transition matrix} $P$. The entry $P_{kj} := P (k, j)$ is the probability of making a transition from $k$ to $j$, and the rows of $P$ sum to 1 to preserve normalization $\sum_{j \in \mathcal{X}} P (k, j) = 1$; we say that such a matrix is \textit{stochastic}. Additionally, $\mathcal{M}$ is \textit{ergodic} if it is \textit{aperiodic} (it does not get trapped in cycles) and \textit{irreducible} (every state can reach every other). 

One step in the chain obeys 
$$ \mu^{(t)} P  = \mu^{(t + 1)} \quad \forall t \geq 0.$$
For any initial distribution $\mu^{(0)} \in \R^{n}$ over $\mathcal{X}$, the distribution after $t$ steps of the walk is
$$\mu^{(0)} P^t = \mu^{(t)} \quad \forall t \geq 0.$$
We say that a distribution $\pi$ over $\mathcal{X}$ is a \textit{stationary distribution} if 
$$ \pi P = \pi.$$

For a function $f : \Xcal \rightarrow \mathbb{R}$, we define
\begin{align}
    \expP[f(x)] := (Pf)(x) = \sum_{y \in \Xcal}P(x,y)f(y),
\end{align}
and for two such functions $f, g$ their $\pi$ inner product is
\begin{align}
    \langle f, g \rangle_{\pi} = \sum_{x\in \Xcal}\pi(x)f(x)g(x).
\end{align}

The \emph{reversed Markov chain} of $\Mcal$ is defined as the Markov chain $\mathcal{M}^{*}= (\Xcal, P^*, \pi)$, which shares the stationary distribution $\pi$ of $\Mcal$, and $P^*$ is defined by the equation:
\begin{align}
\pi(x)P(x,y)  = \pi(y) P^* (y, x).
\end{align}
The chain $\Mcal$ is called \emph{reversible} if $P^* = P$.

The \textit{total variation distance} between two probability distributions $\mu$ and $\nu$ on $\mathcal{X}$ is defined by 
$$\textup{TV}(\mu , \nu) := \frac{1}{2} \| \mu - \nu \|_{1}.$$ It also satisfies the folowing variational formula
\begin{align*}
\textup{TV}(\mu , \nu) = \textup{sup}_{f: \lVert f \rVert_{\infty} \leq 1}\frac{1}{2}(\mathbb{E}_{\mu}[f(x)]- \mathbb{E}_{\nu}[f(x)]).  
\end{align*}
Another important quantity for comparing probability distributions is the Kullback--Leibler (KL) divergence (or, relative entropy) of two probability distributions $\mu$ and $\nu$ over $\mathcal{X}$:
$$ \KL (\mu \| \nu ) := \sum_{x \in {\mathcal{X}}} \mu (x) \ln\left( \frac{\mu(x)}{ \nu (x)} \right).$$
For KL divergence we also have the Donsker and Varadhan's variational formula \cite{donsker1983asymptotic} for $\textup{KL} (\mu\|\nu)$:
\begin{equation*}%
    \textup{KL} (\mu\|\nu)=\sup_{f \in \mathcal{F}} \left\{ \mathbb{E}_{\mu}[f(x)]-\ln(\mathbb{E}_\nu [\exp(f(x))] \right\},
\end{equation*}
where $\mathcal{F}$ denotes the set of all measurable functions.

The \textit{mixing time} of a Markov chain is the amount of time it takes for the distance to stationarity to be small:
$$ t_{\text{mix}} (\varepsilon) := \min \left\{ t : d(t) \leq \varepsilon \right\},$$
where $d(t) := \sup_{\mu} \textup{TV}(\mu P^t, \pi)$.  The mixing time of a \emph{reversible} Markov chain is related spectral properties of $P$. In this case, the matrix $P$ is similar to a symmetric matrix, and thus diagonalizable. The eigenvalues of $P$ can be ordered as
$$1 \geq \lambda_1 \geq \cdots \geq \lambda_N \geq -1.$$
It is known that $\lambda_1 = 1$ and $\lambda_2 < 1$. If the chain is aperiodic, then we can take $\lambda_N > -1$, where any chain can be made aperiodic by introducing self-loops, i.e. making the chain lazy. We define the \textit{spectral gap} of a reversible Markov chain to be:
$$ \delta : = 1 - \max_{\{\lambda_2, \dots, \lambda_N\}} \{ \lvert \lambda \rvert :  \lambda \neq \pm 1\}.$$ 
The relationship between mixing time and the spectral gap can be expressed as 
$$ t_{\text{mix}} (\varepsilon) = \widetilde{\Ocal}_{\frac{1}{\varepsilon}} \left( \delta^{-1} \right).$$
A larger spectral gap therefore implies faster mixing, meaning that the Markov chain more rapidly converges to the stationary distribution.

Next we define the \textit{Discriminant matrix} of a  Markov chain, which is a useful tool for analyzing random walks. 
\begin{definition}[Discriminant Matrix]
\label{defn:discriminant}
For a Markov chain $\mathcal{M} = (\Xcal, P, \pi)$, the discriminant Matrix is the operator with elements
\begin{align}
    D(P)_{ji} := \sqrt{P_{ij} \circ P^{*}_{ji}} = \left( \diag(\pi)^{1/2} P \diag(\pi)^{-1/2} \right)_{ij}.
\end{align}
Furthermore,  if $P$ is reversible, then $D(P)$ is symmetric, and the following hold:
\begin{enumerate}
    \item The unique, maximum eigenvalue eigenstate of $D(P)$ is $|\sqrt{\pi}\rangle$ with eigenvalue $1$. 
    \item The spectral gap of $-D(P)$ (and equivalently $P$) is 
    \begin{align}
    \label{eqn:spec_gap}
        \delta : = 1 - \max \{ \lvert \lambda \rvert : \lambda \in \sigma(D(P)), \lambda \neq \pm 1\}.
    \end{align}
    \item $\lVert D(P) \rVert_2 = 1$.
    \item If $P$ is symmetric, then $D(P) = P$.
\end{enumerate}
\end{definition}

\begin{definition}[Markov Functionals]
Let $f: \Xcal \rightarrow \mathbb{R}$. The Dirichlet form, $\mathcal{D}(f ,f)$, generated by a Markov chain $\Mcal = (\Xcal, P,\pi)$ is defined by
\begin{align*}
    \mathcal{D}(f, f) := \langle f, (I - P) f\rangle_{\pi},
\end{align*}
and if $P$ is reversible, then $I-P$ is symmetric with respect to the $\pi$-inner product, and so $\mathcal{D}$ extends to an inner product for functions $f, g$:
$$
\mathcal{D}(f, g) = \langle f, (I - P) g\rangle_{\pi} = \frac{1}{2}\mathbb{E}_{x\sim\pi}\expP[(f(x) - f(y))(g(x)-g(y))].
$$
The $\pi$-Variance of $f$ is defined as 
$$
 \textup{Var}_{\pi}(f) := \mathbb{E}_{\pi}[f^2] - (\mathbb{E}_{\pi}[f])^2,
$$
and the $\pi$-Entropy of $f$ is defined as 
$$
 \textup{Ent}_{\pi}(f) := \mathbb{E}_{\pi}[f\ln(f)] - \mathbb{E}_{\pi}[f\ln(\mathbb{E}[f])].
$$
\end{definition}

One can relate the variance to the Dirichlet form and the spectral gap using a \textit{Poincar\'e inequality}:
\begin{definition}[Poincar\'e Inequality]
A Markov chain $\Mcal = (\Xcal, P,\pi)$ satisfies a Poincar\'e inequality with constant $\delta$ if 
$$
      \mathcal{D}(f, f) \geq \delta \textup{Var}_{\pi}(f).
$$
\end{definition}
For reversible Markov chains, the Poincar\'e constant is equal to the spectral gap. It is possible to obtain better bounds on the mixing time of a Markov chain using the so-called \textit{logarithmic Sobolev} inequalities.

\begin{definition}[Log-Sobolev Inequality]
A Markov chain $\Mcal = (\Xcal, P, \pi)$ satisfies a \textit{log-Sobolev inequality} with constant $\omega := \omega_{\textup{LS}}$ if 
$$
      \mathcal{D}(f, f) \geq \omega\textup{Ent}_{\pi}(f^2).
$$
\end{definition}

\begin{definition}[Modified Log-Sobolev Inequality]
A Markov chain $\Mcal = (\Xcal, P, \pi)$ satisfies a modified log-Sobolev inequality with constant $\omega_{\textup{MLS}}$ if 
$$
      \mathcal{D}(f, \ln f) \geq \omega_{\textup{MLS}}\textup{Ent}_{\pi}(f).
$$
\end{definition}
The following chain of inequalities is well-known:
\begin{align}
    \delta \geq \omega_{\textup{MLS}} \geq \omega_{\textup{LS}}.
\end{align}
The Poincar\'e and log-Sobolev inequalities belong to a group known as the functional inequalities. They are also well-defined for non-reversible chains, although Poincar\'e now bounds the singular-value gap, and enable one to bound the corresponding mixing time \cite{chatterjee2023spectralgapnonreversiblemarkov}.

Finally, we define the $P$-pseudo Lipschitz norm, which measures the smoothness of a function with respect to the transition probabilities of a Markov chain. 
\begin{definition}[$P$-pseudo Lipschitz Norm]
\label{defn:pseudo-Lipschitz}
Let $\mathcal{M} = (\Xcal, P, \pi)$ be a Markov chain. The $P$-pseudo Lipschitz norm of $f: \Xcal \rightarrow \mathbb{R}$ is defined to be
$$
   \lVert f \rVert_{P} := \max_{x\in \mathcal{X}}\expP[(f(x) - f(y))^2].
$$
\end{definition}
 
\subsection{Short Path Algorithms}
Let $H: \pmone^n \rightarrow \R$ be a classical cost function satisfying $\sum_z H(z) =0$ for $H$ with no constant term. Consider the combinatorial optimization problem 
\begin{equation*} 
    E^{\star}:= \min_{z \in \pmone^n} H(z),
\end{equation*}
where $E^{\star}$ is the optimal objective value. Let $\Pi^{\star}$ denote the orthogonal projector onto subspace spanned by optimal assignments $|z^{\star}\rangle$. %

Let $X = \sum_{i \in [n]} X_i$ be the transverse-field operator, where $X_i$ denotes the Pauli-$X$ operator acting on qubit $i \in [n]$. A well-known approach to determine some $|z^{\star}\rangle$ is the \textit{quantum adiabatic algorithm} (QAA) \cite{farhi2000quantum}. The QAA finds a $|z^{\star}\rangle$ by considering the adiabatic time evolution of 
$$ H_b^{(\textup{QAA})} = - (1-b)X + b H,$$
as $b$ is tuned from $b = 0$ to $b = 1$.
However, QAA is known to suffer from certain ``localization'' issues, which can be viewed as a quantum analogue of getting trapped in local minima, and can result in run times that are exponentially worse than classical brute-force search \cite{altshuler2009adiabatic}. This manifests as a result of the ``avoided crossing'' phenomenon, or first-order phase transition that can lead to exponentially (or even super-exponentially) small spectral gaps. 

Recently, Hastings \cite{hastings2018shortPath} and Dalzell et al.~\cite{dalzell2022mind} proposed a new paradigm for avoiding the first-order phase transition problem with the QAA. Following the approach of \cite{dalzell2022mind}, prototypical adiabatic optimization is eschewed through two modifications. First, the term $H$ is replaced with $g_\eta \left( \frac{H}{| E^{\star} |} \right)$ for a piecewise-linear function $g_\eta: [-1, \infty) \rightarrow [-1,0]$ parameterized by $\eta \in [0,1)$:
$$ g_{\eta} (x) := \min \left( 0, \frac{x+1- \eta}{\eta} \right),$$
leading to the \emph{short-path Hamiltonian}:
\begin{align}\label{eqn: H_b}
    H_b = -\frac{X}{n}+b g_\eta \left( \frac{H}{| E^{\star} |} \right),
\end{align}
where $X$ has been normalized by its spectral norm, and $b \in [0, \infty)$.
Second, rather than slowly evolve from $-\frac{X}{n}$ to $\frac{H}{| E^{\star} |}$ as in the QAA, we jump from $- \frac{X}{n}$ to the ground state $\ket{\psi_b}$ of $H_b$ for some value of $b$ that is independent of $n$ (where the spectral gap is guaranteed to be large), and then jump from $H_b$ to the ground-state space of $\frac{H}{| E^{\star} |}$.  Note in \cite{dalzell2022mind} they also allow for scaling $H$ by an overestimate of $\lvert E^*\rvert$, for simplicity we just stick with scaling by $\lvert E^*\rvert$.

The jumps are accomplished using a unitary $U$, which combines phase estimation and amplitude amplification. For a high-level understanding, suppose we seek to enact a jump between two Hamiltonians $H_1$ and $H_2$, each acting on $n$ qubits. Denote the ground state of $H_1$ by $\ket{\psi_1}$, and let $\Pi_2$ be the projector onto the ground space of $H_2$. The unitary $U$ first employs phase estimation to implement reflection operators $R_1$ and $R_2$ that reflect about the state $\ket{\psi_1}$, and the groundspace of $H_2$, respectively. If $\delta_j$ is the spectral gap of Hamiltonian $H_j$, the operator $R_j$ can be implemented up to error $\varepsilon$ using $\delta_j^{-1} \log (1/\varepsilon)$ calls to a block-encoding of $H_j$, and often realizable using $\poly (n)$ gates. When $H_j$ is a classical Hamiltonian, $R_j$ can be implemented using $\poly (n)$ gates irrespective of $\delta_j$. From here, the unitary $U$ employs fixed-point amplitude amplification \cite{yoder2014fixed} to implement $\frac{\Pi_2 \ket{\psi_1}}{\left\| \Pi_2 \ket{\psi_1} \right\|}$, requiring at most $\Ocal \left( \left\| \Pi_2 \ket{\psi_1} \right\|^{-1} \log (1/\varepsilon) \right)$ applications of $R_1$ and $R_2$.

The algorithm is initialized to $\ket{\pmb{+}}:=\ket{+}^{\otimes n}$, the ground state of $-\frac{X}{n}$. Then, the ground state $\ket{\psi_b}$ of $H_b$ is prepared by jumping from $-X/n$ to $H_b$. Finally, we prepare $\frac{\Pi^{\star}\ket{\psi_b}}{\norm{\Pi^{\star}\ket{\psi_b}}}$ by jumping from $H_b$ to the classical Hamiltonian $\frac{H}{\abs{E^{\star}}}$. The state $\Pi^{\star} \ket{\psi_b}$ is a superposition of optimal solutions to $\min_{z \in \{ \pm 1\}^n} H(z)$, and thus measurement in the computational basis yields an optimal bit-string $z^{\star}$ with high probability. The first jump is small (in the sense that the success probability is nearly 1), whereas the second jump is large (the success probability is exponentially small). The resulting time complexity scales as $\mathcal{O}^* \left(2^{\left( \frac{1}{2} -c\right) n}\right)$, indicating a super-Grover speedup when $c > 0$. In \cite{hastings2018shortPath}, the order of short and long jump steps is reversed. We refer to both approaches \cite{hastings2018shortPath, dalzell2022mind} as \textit{short path algorithms}.

\section{Technical Overview}\label{sec:overview}

\subsection{Framework}
Define $\Xcal \subseteq \{-1,1\}^n$ and let $H \colon \Xcal \rightarrow \R$ be a cost function. We are interested in (exactly) determining $\zstar \in \Xcal$ such that 
$$\zstar \in \argmin_{z \in \Xcal}H(z).$$
We also use $H$ to refer to a diagonal Hamiltonian in a Hilbert space with a basis indexed by $z \in \Xcal$, with $\<z|H|z\> $ for the Hamiltonian identified with $H(z)$. Whether we are referring to the quantum Hamiltonian or the function will be clear from context. We define $\Estar := \min_{z \in \Xcal} H(z)$ and assume everywhere that the cost is scaled to ensure $\Estar < 0$. We will further assume that $|\Xcal|$ is super-polynomial, as our primary concern is with super-quadratic speedups over algorithms with at least super-polynomial runtime. Suppose $\pi$ is a distribution such that $\pi(E^{\star})$ is the probability that a sample from $\pi$ is a global minimizer. Suppose there exists a Markov chain with transition matrix $P$ that mixes to stationary distribution $\pi$, then the mixing time is bounded by $\tmix(\varepsilon) = \poly(n,\log(\varepsilon^{-1}))$. %
\begin{algorithm}
  \caption{\textsc{MarkovChainSearch}}
  \label{alg:MarkovChainSearch}
  \hspace*{\algorithmicindent} \textbf{Prerequisites:} Solution space $\Xcal \subset \{-1,1\}^n$, Cost function $H \colon \Xcal \rightarrow \R$, distribution $\pi$ such that $\pi(E^{\star})$ is the probability that a sample from $\pi$ is a global minimizer, a Markov chain with transition matrix $P$ and mixing time $\tmix = \tmix(\pi(E^{\star})/2) = \mathcal{O}(\poly(n))$.  \\ 
  \hspace*{\algorithmicindent} \textbf{Input}: Description of $\poly(n)$ time procedures to evaluate $H$ and perform a step of the Markov chain described by $P$, failure probability $\epsilon$.\\   
 \hspace*{\algorithmicindent} \textbf{Output}: A global minimizer $\zstar$ of $H(z)$ over $\Xcal$.
  \begin{algorithmic}[1]
    \State Set $i=0$, $z^{(0)}$ to an arbitrary point in $\Xcal$.
    \While {$i \le \frac{2}{\pi(E^{\star})}\log\left(\frac{1}{\epsilon}\right)$}
        \State Simulate $\tmix$ steps of $P$ to obtain sample $\tilde{z}$
        \If {$H(\tilde{z}) \le H(z^{(i)})$}
        \State Set $z^{(i+1)} = \tilde{z}$  and $i 
        \gets i + 1$.
        \Else 
        \State Set $z^{(i+1)} = z^{(i)}$  and $i 
        \gets i + 1$.
        \EndIf
    \EndWhile
    \State Output $z^{(i)}$.
  \end{algorithmic}
\end{algorithm}

Under these conditions, it follows that Algorithm~\ref{alg:MarkovChainSearch} finds the global minimizer of $H$ in time $\Ostar\left(\pi(E^{\star})^{-1}\right)$. Our framework seeks to accelerate this runtime, and so we assume the same setting for the quantum algorithm. In order to define the quantum framework we must access $H, P,$ and $\pi$. We assume the existence of the following subroutines:

\begin{assumption}[Quantum Input Assumptions]
\label{asm:input-assumptions}
The following subroutines are used as primitives in the Generalized Quantum Short Path framework.
\label{asm:quantum-input}
    \begin{enumerate}
        \item \textbf{Initial State Preparation:} We assume the existence of a unitary $U_\pi$ implementable using $\poly(n)$ gates, such that $U_\pi|0\rangle = |\sqrt{\pi}\rangle \coloneqq \sum_{z \in \Xcal} \sqrt{\pi(z)} |z\>$.
        \item \textbf{Block-encoding of Markov Chain:} Suppose that $P$ is the transition matrix of a reversible Markov chain, then the discriminant operator $D(P)$ (see Definition \ref{defn:discriminant}) is Hermitian. We assume the existence of a unitary $U_{D(P)}$, implementable with $\textup{poly}(n)$ gates, that is a $(\kappa_1,a)$ block-encoding of $D(P)$ for $\kappa_1 =\mathcal{O}(\poly(n))$. 
        \item \textbf{Block-encoding of Cost Function:} We assume the existence of a unitary $U_{H}$, implementable with $\textup{poly}(n)$ gates, that is a $(\kappa_2,a)$ block-encoding of $H$ for $\kappa_2~=~\mathcal{O}(\poly(n))$.
    \end{enumerate}
\end{assumption}

We will justify Assumption~\ref{asm:quantum-input} for each application of the framework. Note that preparing a block-encoding of the cost function is straightforward given a $\poly(n)$ size circuit to evaluate it at a single point, and we do not make this analysis in every case. We also do not explicitly mention these input assumptions in each of our results to avoid cluttering the presentation, but they are prerequisites for the input model in each case.
With this setup, we can define the generalized short path framework. We define a generalized ``short path Hamiltonian" $H_b$ as 
\begin{equation}
    \label{eqn:generalized-hamiltonian}
    H_b = - D(P) + b g_\eta \left(\frac{H}{|E^\ast|}\right),
\end{equation}
where $D(P)$ is the Discriminant matrix corresponding to $P$ and $g_\eta$ is defined similarly to \cite{dalzell2022mind}, by
$$ g_{\eta} (x) := \min \left( 0, \frac{x+1- \eta}{\eta} \right).$$ The block-encoding for $H_b$ can be constructed using the \textit{linear combination of block-encodings} technique \cite{gilyen2019QSVT} using $U_{D(P)}$ and $U_H$.

The framework is specified in Algorithm~\ref{alg:generalized-short-path}. The implementation of the jumps uses the framework from~\cite[Proposition 21]{dalzell2022mind} that performs fixed point amplitude amplification using reflections constructed from the block-encodings of the Hamiltonians and the quantum singular value transformation (QSVT). 
The overall runtime of Algorithm \ref{alg:generalized-short-path}, in terms of the number of queries to $U_{\pi}$ and block-encodings of $H$ and $D(P)$, is
\begin{align}
   [\min(\text{Gap}(-D(P)), \text{Gap}(H_b))]^{-1}\left(\lvert \langle \sqrt{\pi} |\psi_b\rangle\rvert^{-1} +  \lVert \Pi^{\star}|\psi_b\rangle\rVert_2^{-1}\right)\label{eqn: runtime},
\end{align}
where %
$\Pi^{\star}$ is the projector onto the ground subspace of $H$. As in previous papers, we refer to one of the steps in the algorithm as the short jump and another as the long jump. The reason for this is that the choices of $b,\eta$ made for applications of the framework will always ensure that the short jump can be carried out with a polynomial number of queries to $U_{\pi}, U_{D(P)},$ and a block-encoding of $H_b$. Thus when including Assumption \ref{asm:input-assumptions}, the runtime of the algorithm is therefore primarily determined by the long jump, and under the appropriate conditions is bounded by $\Ostar(\pi(E^{\star})^{-(0.5-c)})$, leading to a super-quadratic speedup over Markov chain search.

\begin{algorithm}
  \caption{\textsc{GeneralizedShortPathAlgorithm}}
  \label{alg:generalized-short-path}
  \hspace*{\algorithmicindent} \textbf{Input}: Algorithmic parameters $b,\eta$, Problem parameters $H, P, \pi, E^{\star}$, which define $H_b$ in Equation~\eqref{eqn:generalized-hamiltonian}. \\
 \hspace*{\algorithmicindent} \textbf{Output}: an optimal assignment $z^{\star}$ for $H$.
  \begin{algorithmic}[1]
    \State Prepare $\ket{\sqrt{\pi}}$, the ground state of $-D(P)$.
    \State \textbf{Short Jump:} Prepare $\ket{\psi_b}$ up to exponentially small error with jump $-D(P)\rightarrow H_b$.
    \State \textbf{Long Jump:} Prepare $\frac{\Pi^{\star}\ket{\psi_b}}{\norm{\Pi^{\star}\ket{\psi_b}}}$ up to exponentially small error with jump $H_b\rightarrow \frac{H}{\abs{E^{\star}}}$.
  \end{algorithmic}
\end{algorithm}

Our bounds on the runtime rely on two conditions, that we view as uplifted versions of corresponding notions in~\cite{dalzell2022mind} to the case of general Markov chains. Our first conditon captures the \emph{smoothness} of the cost function under applications of the transition matrix.

\begin{definition}[$\Delta_{P}$ stability]\label{defn:delta-p}
    Let $\Mcal = (\Xcal, P, \pi)$ be a Markov chain. We say that the cost Hamltonian $H$ is $\Delta_P(\eta)$ stable under $\Mcal$ if 
\begin{equation}\label{e:DeltaCondition}
\expP[h_{\eta}\left(H(y)\right)]\leq  h_{\eta}\left(H(x)+\Delta_{P}(\eta)\right), \quad \forall x \in \Xcal
\end{equation}
where $h_{\eta}(x) := g_{\eta}\left(\frac{x}{\lvert E^*\rvert}\right)$. If the conditions holds for all $0< \eta < 1$ we omit $\eta$ and simply say $H$ is $\Delta_P$ stable under $\Mcal$.
\end{definition}

The analysis in the following sections will clarify that $\Delta_P$-stability is a generalization of the $\alpha$-subdepolarizing condition introduced in~\cite{dalzell2022mind}. In fact, it is equivalent to a more syntactically obvious generalization of the $\alpha$-subdepolarizing condition, we state Definition~\ref{defn:delta-p} as the primary condition as it is easier to demonstrate and interpret in most cases.

The next conditon is a generalization of the spectral density condition of~\cite{dalzell2022mind}. We capture the idea that the measure (according to $\pi$) of the set of solutions $z$ for which $g_\eta\left(\frac{H(z)}{|\Estar|}\right) \neq 0$ is polynomially related to the measure of the global minimizer, for some value of $\eta$. In other words, sampling from $\pi$ does not allow one to approximately minimize $H$ to arbitrary constant relative error, super-polynomially faster than finding the exact minimum. If this condition is violated, the problem admits a simple classical sub-exponential time approximation scheme. We define this condition as follows:
\begin{definition}[$\gamma$ Spectral Density]\label{defn:gamma-spectral}
The cost Hamiltonian $H$ is said to satisfy the $\gamma$ spectral density condition with respect to the stationary distribution $\pi$ if:
$$
    \pi(E \leq (1-\eta)E^{\star}) \leq \pi(E^{\star})^{\gamma},
$$  
for some $\eta \in (0, 1)$.
\end{definition}
In section \ref{subsec:constructing_short_path_markov_chains}, we will present conditions on the underlying Markov chain that enable bounding $\gamma$. However, in some cases, we will simply assume $\gamma = \Theta(1)$, which as mentioned earlier is a mild condition for hard problems.

We are now ready to state the main results concerning our framework. These results are subject to further technical conditions on the Markov chain used for the search, and are formulated in terms of conditions that lead to efficient mixing time bounds. We have two variants of our result, the first relies on a logarithmic-Sobolev inequality for $P$.
\begin{theorem}[Informal Theorem \ref{thrm:MTG_RT}]\label{thm:main-general-logsobolev}
Let $\Mcal = (\Xcal, P, \pi)$ be a reversible, aperiodic Markov chain, and let $H: \Xcal \rightarrow \R$ be a diagonal, $\Delta_P$-stable Hamiltonian with ground state energy $E^{\star}$, that satisfies the $\gamma$ spectral density condition for some parameter $\eta$. In addition, suppose $\Mcal$ satisfies an $\omega$ log-Sobolev inequality.%
Then, under Assumption \ref{asm:input-assumptions}, Algorithm~\ref{alg:generalized-short-path} determines the ground state of $H$ over $\Xcal$ with running time
$$
\mathcal{O} \left(\textup{poly}(n)\omega^{-1}[\pi(E^{\star})^{-1}]^{\left(\frac{1}{2}-\frac{\eta(1-\eta)\lvert E^{\star}\rvert b}{2\ln(1/\pi(E^{\star}))\Delta_P}\right)}\right),
$$
where $b = \mathcal{O}\left(\gamma\omega\ln\left(\frac{1}{\pi(E^{\star})}\right)\right)$.
\end{theorem}
The spectral density condition with respect to non-uniform starting states presents a technical challenge, since unlike the uniform distribution over all
bitstrings of length $n$, they are no longer product measures. Fortunately, the condition that $\pi$ is the ground state of a fast mixing Markov chain allows for simplification via concentration inequalities for Markov chains, e.g. the so-called Herbst argument \cite{lalley2013concentration}, that allows the spectral density conditions to be established as long as the cost function has an appropriately bounded pseudo-Lipschitz norm $\norm{H}_P$.

We also present a variant of the above result that only relies on the weaker Poincar\'e inequality and does not require the spectral density condition. Still, Corollary \ref{cor:tail-poinc} shows that a Poincar\'e inequality does imply the spectral density condition, anyways, for some potentially falling $\gamma$.

\begin{theorem}[informal Theorem \ref{thrm:MTG_RT_poincare}]\label{thm:main-general-poincare}
Let $\Mcal = (\Xcal, P, \pi)$ be a reversible, aperiodic Markov chain, and let $H: \Xcal \rightarrow \R$ be a diagonal, $\Delta_P$-stable Hamiltonian with ground state energy $E^{\star}$ that satisfies the $\gamma$ spectral density condition for some parameter $\eta$. In addition, suppose $\Mcal$ satisfies a $\delta$ Poincar\'e inequality. %
Then, under Assumption \ref{asm:input-assumptions}, Algorithm~\ref{alg:generalized-short-path} determines the ground state of $H$ over $\Xcal$ with running time
$$
\mathcal{O} \left(\textup{poly}(n)\delta^{-1}[\pi(E^{\star})^{-1}]^{\left(\frac{1}{2}-\frac{\eta(1-\eta)\lvert E^{\star}\rvert b}{2\ln(1/\pi(E^{\star}))\Delta_P}\right)}\right),
$$
where $b = \mathcal{O}(\delta)$.
\end{theorem}
Theorems~\ref{thm:main-general-logsobolev} and \ref{thm:main-general-poincare} are proven in Section~\ref{sec:framework}, in order to establish specializations that rely on a bounded pseudo-Lipschitz norm. Additionally, in Section \ref{sec:framework} we present conditions under which the super-quadratic speedup is asymptotic, i.e. does not fall with problem size $n$.

\subsection{Applications}
The fact that our main results in Theorems \ref{thm:main-general-logsobolev} and \ref{thm:main-general-poincare} rely on functional inequalities implies many possibilities for interesting algorithmic speedups. Once a Markov chain with the right properties (log-Sobolev or Poincar\'e with the proper parameters) is identified, we can derive conditions on cost functions for which we have super-quadratic speedup over Markov chains. As a simple toy example, consider an expander graph of size $2^{n}$, where we are given access to a polynomial time oracle that outputs the edges incident on any vertex. Since the graph is an expander, the graph random walk satisfies a Poincar\'e inequality with constant $\delta$. Consider any assignment of costs to the nodes such that the  difference in cost between any two endpoints of an edge is bounded above by a constant $\Delta$. It follows from our results (Theorem \ref{thm:main-general-poincare}), that the generalized short path framework can find the node with minimum cost with super-quadratically fewer queries than searching with the random walk on the graph. A systematic study of cost functions that yield a speedup for various Markov chains may lead to some interesting insights. For this paper, however, we focus on identifying connections to problems of general interest.

Theorem \ref{thm:formal_speedup} presents a summary of the quantum runtimes obtained for all applications considered in this paper.
\begin{theorem}[Applications of Generalized Short Path Framework]
\label{thm:formal_speedup}
    The following optimization problems exhibit a quantum runtime of $\Ostar\left(T(n)^{0.5 - c(n)}\right)$, where $T(n)$ is the runtime of a classical algorithm based on Markov chain search, and $c(n) > 0$, specified in Tables \ref{table:maxcut} and \ref{table:gibbs}, quantifies the degree of improvement over quadratic speedup: 
    \begin{enumerate}
        \item Finding a generalization of maximum bisection (called MaxCut Hamming) over  Erd\H{o}s-R\'enyi random graphs with constant average degree, where the smaller partition is constrained to have $k$ nodes (where $k \le n/2$). In both cases, the Markov chain is the transposition walk.
        
        \item Finding a maximum independent set on graphs with constant maximum degree, where the Markov chain is Glauber dynamics, 
        \item Optimizing the antiferromagnetic Ising model on graphs of constant maximum degree, where the Markov chain is Glauber dynamics
        \item Optimizing the Sherrington-Kirkpatrick model where the Markov chain is the Glauber dynamics
    \end{enumerate}
\end{theorem}

\begin{table}[h!]
\centering
\begin{tabular}{|m{4cm}|m{5cm}|m{4cm}|m{2cm}|}
\hline
\textbf{Problem} & \textbf{Graph Model} & \textbf{Advantage $(c(n))$} & \textbf{Theorem} \\
\hline
Maximum Bisection & $\mathcal{G}_{\textup{ER}}(n, \frac{d}{n}), d = \mathcal{O}_n(1)$ & $\Theta(1)$ & Theorem~\ref{thm:maxcut_runtime} \\
\hline
MaxCut Hamming & $\mathcal{G}_{\textup{ER}}(n, \frac{d}{n}), d = \mathcal{O}_n(1)$ & $\Theta(1/\log(n/k))$ & Theorem~\ref{thm:maxcut_runtime} \\
\hline
\end{tabular}
\caption{\label{table:maxcut}Degree of quantum improvement for problems using MCS with transposition walk.}
\end{table}

\begin{table}[h!]
\centering
\begin{tabular}
{|m{2.3cm}|m{3.6cm}|m{4.5cm}|m{2.7cm}|m{2.2cm}|}
\hline
\textbf{Model} & \textbf{Graph Model} & \textbf{Inverse Temperature} & \textbf{Advantage $(c(n))$} & \textbf{Theorem} \\
\hline
\multirow{3}{=}{Hardcore Model 
} 
    & $d = \mathcal{O}_n(1)$ 
        & $\beta  \geq -(d-1)\ln(d-1)+d\ln(d-2)$ & $\Theta(1)$ & Theorem~\ref{thm:mis_constant_frac_crit} \\
\cline{2-5}
    & $d = \mathcal{O}_n(1)$ 
        & $\beta = -(d-1)\ln(d-1)+d\ln(d-2)$& $\Omega\left(\frac{1}{n^{1+4e(1+\frac{1}{d-2})}}\right)$ & Theorem~\ref{thm:mis_at_criticality_runtime} \\
\cline{2-5}
    & $\mathcal{G}_{\textup{RG}}(n, d), d = \mathcal{O}_n(1)$ 
        &$\beta > \ln\left({2\sqrt{d-1}-1}\right)$ & $\Omega(1/n)$ & Theorem~\ref{thm:mis_beyond-criticality_runtime} \\
\hline
\multirow{3}{=}{Anti-\\ferromagnetic Ising Model
} 
    & $d = \mathcal{O}_n(1)$
        &$(d-1)\tanh(\beta) < 1 $& $\Theta(1)$ & Theorem~\ref{thm:ising_constant_frac_crit} \\
\cline{2-5}
    & $d = \mathcal{O}_n(1)$
        & $(d-1)\tanh(\beta) \leq 1 + o_{n}(1)$ & $\Omega(1/n^{1+1/d})$ & Theorem~\ref{thm:ising_at_critical_runtime} \\
\cline{2-5}
    & $\mathcal{G}_{\textup{RG}}(n, d), d = \mathcal{O}_n(1)$ 
        &  $(d-1)\tanh(\beta) \leq \frac{d-1}{8\sqrt{d-1}-1}$ & $\Theta(1)$ & Theorem~\ref{thm:ising_beyond_crit_runtime} \\
\hline
Sherrington-Kirkpatrick Model& -- & $\beta < \frac{1}{4}$ & $\Theta(1/\log(n))$ & Theorem~\ref{thm:sk_runtime} \\
\hline
\end{tabular}

\caption{\label{table:gibbs}Degree of quantum improvement for problems using MCS with Glauber Dynamics. Note in all cases $d\geq 3$. $\beta$ corresponds to inverse temperature for the Gibbs sample $\propto \exp(-\beta H)$.}
\end{table}

\paragraph{Quantum Improvements that Lead to Speedups over Classical State-of-the-art}
\label{subsec:contr_quantum_improve}
We would like to highlight two particular cases where we obtain quantum algorithms that are faster than existing provable, classical state-of-the-art generic algorithms. The first case is for the antiferromagnetic Ising model over random regular graphs. Here in Table \ref{table:gibbs} we obtain $c(n) = \Theta(1)$, leading to an asymptotic super-quadratic speedup over Markoc chain search. The underlying Markov chain is the state-of-the-art classical approach of \cite{chen2025rapidmixingrandomregular}. Secondly, while for the case of MIS over random-regular graphs the improvement over quadratic speedup falls with $n$, the underlying classical MCS algorithm appears to be faster than other classical generic approaches when the degree is sufficiently large.

In the following subsections, we expand more on the details of the applications that we consider and how the ability to speedup MCS enables faster quantum algorithms. In Section \ref{subsec:groverizing}, we also discuss about potential evidence that some of the speedups obtained cannot be what we call ``Groverized.'' We say that a speedup with (generalized) short-path can be \emph{Groverized} if the edge-over-quadratic can be replicated with a classical subroutine combined with quantum minimum finding.

\subsubsection{Optimization with Fixed Hamming Weight}
\label{sec:ham-weight-overview}
We first consider optimization problems for which feasible solutions are bitstrings of fixed Hamming weight. A well-studied example of this setting is Max-Bisection~\cite{frieze1997improved, dembo2017extremal}, defined as follows
\begin{equation}\tag{MaxBisection}
\mathcal{C}^{*}_{\frac{n}{2}} := \min_{ x \in \{-1, 1\}^{n}}\left\{ -\frac{1}{2}\sum_{ i < j} e_{ij} (1-x_ix_j) : \lvert x\rvert = \frac{n}{2} \right\},
\end{equation}
The algorithm of~\cite{dalzell2022mind} does not directly yield useful results for such a problem. Firstly, the framework does not naturally incorporate constraints. More importantly, although one could attempt to enforce the constraints by means of penalty terms, this prohibits the possibility of super-Grover speedups. To see why this is the case, recall that the algorithm of~\cite{dalzell2022mind} is simply our Algorithm~\ref{alg:generalized-short-path} with the Markov chain chosen to be the random walk on the edges of the hypercube. Since every transition of such a walk changes the hypercube, the best possible value of $\Delta_P$ for the stability condition to be satisfied is of the same order as the penalty terms. On the other hand, the penalty terms must be of the same order as the cost function in order to guarantee that constraints are satisfied. By inspection of Theorem~\ref{thm:main-general-logsobolev} we observe that no super-Grover speedup is possible via the penalty-based approach (more details on the penalty approach are in Appendix \ref{sec:penalized_obj}). We overcome this challenge by employing a generalized framework that uses a Markov chain, specifically the \textit{Bernoulli-Laplace diffusion} or \textit{transposition walk}, which preserves Hamming weights and transitions from a starting string of weight $k$ to the equal superposition over all such strings. In Section~\ref{sec:fixed-hamming-weight}, we present a condition on cost functions over Hamming weight slices $k$ for which we obtain runtimes of the form $\poly(n) \binom{n}{k}^{0.5 - c}$ for a constant~$c$. 

We study a generalization of \ref{e:MaxBisection} which we term \ref{e:MaxCut}, defined as
\begin{equation}\tag{MaxCut-Hamming}
\mathcal{C}^{*}_{k} := \min_{ x \in \{-1, 1\}^{n}} \left\{ -\frac{1}{2}\sum_{i < j} e_{ij} (1-x_ix_j) : \lvert x \rvert = k \right\}.
\end{equation}
The runtimes we obtain are highlighted in Table \ref{table:maxcut}. In Section~\ref{sec:fixed-hamming-weight} we prove that the Generalized Short Path Framework achieves an overall runtime of $\Ostar\left(\binom{n}{k}^{0.5 - c}\right)$ for \ref{e:MaxCut} on \er ~random graphs when $k = \Theta(n)$ (which includes \ref{e:MaxBisection} as a special case). We also demonstrate that under the assumption of spectral density, our approach achieves a runtime of $\Ostar\left(\binom{n}{k}^{0.5 - c(n)}\right)$ with $c(n) = \Theta\left([\log(n/k)]^{-1}\right)$, when $k = o(n)$, and thus the super-quadratic advantage over Markov chain search decays as $n \to \infty$ (similar to the results of ~\cite{hastings2018shortPath}).

\subsubsection{Glauber Dynamics} 
Glauber dynamics~\cite{glauber1963time} is a well known sampling algorithm designed to sample from the Gibbs measures corresponding to Hamiltonians such as the Ising or Hardcore models. Since sampling from a Gibbs measure at arbitrarily high inverse temperatures solves the exact optimization, for most hard problems there exists a critical threshold beyond which the Glauber dynamics no longer mixes efficiently, which is discussed in Section \ref{subsec:groverizing}. Performing Markov chain search with Glauber dynamics has two advantages, if the problem is constrained then it provides a natural way to search with a distribution whose support is restricted to feasible solutions only. On the other hand, if the Glauber dynamics mixes for positive inverse temperatures, then low-energy solutions are more favored compared to the uniform distribution and the result in Markov chain search is asymptotically faster than unstructured search. In each case, we will consider classical search algorithms that  use the Glauber dynamics at for inverse-temperatures in the ranges specified in Table \ref{table:gibbs}. We demonstrate in Section~\ref{sec:glauber_dynamics} that for three models of interest, we obtain super-quadratic, but potentially asymptotically falling, speedups over polynomially-mixing Glauber dynamics. These models are:
\begin{enumerate}
    \item \textbf{The Hardcore model on graphs of constant maximum degree:} For this problem the Glauber dynamics is shown to mix only up to critical temperatures that are negative. This means that our starting distribution favors small sets compared to large sets. However, there is the advantage that the Gibbs distribution has support only on independent sets (which are usually much fewer in number than $2^n$, which is the total number of subsets). In the case of random regular graphs of sufficiently high degree, we show that this Markov chain search algorithm is faster than unstructured search (Proposition \ref{prop:comparison-to-brute-force}) as well as the best known combinatorial algorithms for finding a maximum independent set. 
    \item \textbf{The Antiferromagnetic Ising model on graphs of constant maximum degree:} The Antiferromagnetic Ising model is an unconstrained optimization problem, encompassing problems like MaxCut, and there are $2^n$ feasible solutions. In this setting, the Glauber dynamics mixes up to a positive critical inverse-temperature, thus the starting stationary distribution favors low energy solutions and the Markov chain search is faster than unstructured search (Proposition \ref{prop:ising_improve_over_brute_force}).
    \item \textbf{The Sherrington-Kirkpatrick model:} Like the Ising model, the Sherrington-Kirkpatrick model is also unconstrained. The Glauber dynamics has been shown to have fast mixing in total variation distance up to inverse-temperature $\beta < \frac{1}{4}$ \cite{chen2022localizationschemesframeworkproving}. We show that the quantum short-path algorithm applies to Glauber dynamics for SK within this $\beta$ range. Additionally MCS for this $\beta$ range outperforms bruteforce search (Proposition \ref{prop:sk_positive-beta-overlap}). However, mixing in Wasserstein-2 is possible all the way up to a known computational threshold of $\beta < 1$ \cite{alaoui2024samplingsherringtonkirkpatrickgibbsmeasure, celentano2022sudakovferniquepostampnewproof}. %
\end{enumerate}

\subsubsection{Can the Generalized Short-Path Algorithm be ``Groverized"?}
\label{subsec:groverizing}
Our techniques also allow us to shed more light on a fundamental question about the viability of true super-quadratic speedups with the short path framework. The algorithms in this paper as well as the earlier frameworks, rely on preparing a quantum state whose overlap with the global minimizers is larger than that of some easily prepared starting state, referred to as the \emph{short-jump}, and then applying quantum minimum finding to that state. It is apparent that if there existed a classical algorithm to sample in polynomial time from a distribution with overlap that matches that of this intermediate state, then there is a classical algorithm that finds the global minimum only quadratically slower than the short path algorithm. The advantage over search is then essentially ``Groverized"   and there is no hope for true super-quadratic speedups. It is therefore important to understand to what degree classical sampling techniques can approach the overlap of the intermediate state, as discussed in~\cite{dalzell2022mind}. Most natural sampling algorithms are based on the analysis of Markov chains and so our framework provides a useful tool to probe this question.

The quantum speedups over MCS with Gibbs samplers, e.g., Glauber dynamics,   are potentially interesting candidates for providing evidence against Groverization. This is because there are known computational phase transitions in mixing time. Particularly there are two known thresholds for two different graph models. For constant-degree graphs, this is known as the \emph{tree-uniqueness threshold} \cite{sly2010computationaltransitionuniquenessthreshold, galanis2016} and for sparse, random-regular graphs (and the SK model) this known as the \emph{onset of disorder chaos} \cite{chatterjee2009disorderchaosmultiplevalleys, alaoui2024samplingsherringtonkirkpatrickgibbsmeasure, huang2024hardnesssamplingantiferromagneticising}. The first threshold holds unless $\mathsf{NP} = \mathsf{RP},$ and the latter holds against any ``stable'' Gibbs sampling  algorithm \cite{alaoui2024samplingsherringtonkirkpatrickgibbsmeasure}. If the quantum algorithm can produce an overlap with the optimal solution that is larger than the Gibbs distribution at either of these thresholds, then this is evidence against dequantization. This is because, in most cases, Gibbs sampling is the only ``natural'' classical thing to do.

For the case of bounded-degree graphs, we find that the improvement over quadratic, i.e. benefit from the intermediate state, falls with $n$ at the uniqueness theshold. However, this may only be a limitation of current proof techniques. Our numerical results show that the conditions for asymptotic super-quadratic speedup with the generalized short path framework tend to be overly pessimistic.

For the random $d$-regular case, it appears unclear if Glauber dynamics for the antiferromagnetic Ising model can mix all-the-way to the point of disorder chaos. However, the current state-of-the-art inverse temperature $\beta_T$ \cite{chen2025rapidmixingrandomregular} is close, i.e. both $\beta_T$ and the point of disorder chaos $\beta_{\textup{DC}}$ scale inversely with $\sqrt{d}$. Additionally, we prove that the generalized short-path algorithm attains an asymptotic super-quadratic speedup at $\beta_T$. 

Given the above two points, there is still hope that such evidence against Groverization can  be shown.

\begin{remark}
    An earlier version of the paper only included quantum speedups over Glaubers dynamics for inverse temperatures that are a constant fraction away from the tree-uniqueness threshold. Recent results in the classical literature have shown that Glauber dynamics can mix up to the uniqueness threshold and beyond the threshold when the ensemble is random-regular graphs. As highlighted above, this has led to some new quantum runtimes along with some reassessments of the original claims against Groverization. Specifically, Glauber dynamics has been shown to mix for larger inverse-temperatures than previously thought.
\end{remark}

\section{Generalized Short Path Framework}
\label{sec:framework}
This section provides a simplified and generalized analysis of the short path algorithm presented in \cite{dalzell2022mind}. It also highlights limitations of the current method of analysis, and describes a general recipe for the determining a super-quadratic speedup.

\subsection{Summary of main results}

In what follows, let $\Mcal = (\Xcal, P, \pi)$ be a reversible, aperiodic\footnote{Note that any Markov chain can be made aperiodic by making it lazy: $P \mapsto \frac{I+P}{2}, D(P) \mapsto \frac{D(P) + I}{2}$. We will generally assume that this transformation has been implicitly done if the chain is not aperiodic, as the assumption only acts to simplify some results and is not critical. } Markov chain over a finite set $\Xcal$. Here $\pi$ denotes the stationary distribution of $P$. It is assumed that the spectral gap of $P$ is lower-bounded by $\delta$. Using the discriminant matrix of $P$, we can define a more general short path Hamiltonian $H_b$ that allows one to work with mixing operators other than $-\frac{X}{n}$.
\begin{definition}[Generalized Short Path Hamiltonian $H_b$]
 Consider a reversible Markov chain $\Mcal = (\Xcal, P, \pi)$. Let $H: \Xcal \rightarrow \mathbb{R}$ be a cost Hamiltonian with $E^{\star} < 0$. The short path Hamiltonian $H_b$ is given by
$$
    H_b := -D(P) + bg_{\eta}\left( \frac{H}{| E^{\star} |} \right),
$$
where $D(P)$ is the discriminant matrix of $P$,
and
$$ g_{\eta} (x) := \min \left( 0, \frac{x+1- \eta}{\eta} \right).$$
\end{definition}
More generally $g_{\eta} :[-1, \infty) \rightarrow [-1, 0]$ can be a non-decreasing, concave function that is differentiable at every point where it is non-zero. However, the specific choice we make is sufficient for our purposes. We will also sometimes use the notation $G_{\eta} := g_{\eta}\left( \frac{H}{| E^{\star} |} \right)$.

One major component of the analysis of Algorithm~\ref{alg:generalized-short-path} is determining an upper bound on $b$ for which the spectral gap of the short   path Hamiltonian $H_b$ is still large, i.e., $\Omega \left( \frac{1}{\poly (n)} \right)$. This upper bound serves as a proxy for how large the ``short jump'' is. A second major component is determining the increased overlap with the optimal solution provided by the short jump. To do so, we rely on the definition of $\Delta_{P}$ stability (Definition \ref{defn:delta-p}). %
If the short jump can be accomplished efficiently for some constant $b$, then $\Delta_P$ stability captures whether the short path approach provides a super-Grover runtime. This condition also has an intuitive interpretation. If we consider the optimization landscape defined by $\Mcal$, $H$ and a well (controlled by $\eta$) around the global minimum with energy $E^*$, then we do not want the energy to increase too much when moving within and around the well. Specifically, for a super-Grover runtime it should hold that $\Delta_{P}(\eta) = \Theta\left(\frac{\lvert E^*\rvert}{\ln(1/\pi(E^*))}\right)$ for some $\eta$, where $\pi$ is the stationary distribution of $\Mcal$. It is worth remarking that if $\eta = 0$, we recover the quantum unstructured search algorithm. When the explicit value of $\eta$ such that Definition \ref{defn:delta-p} holds is not important, we will simply denote $\Delta_P(\eta)$ as $\Delta_{P}$.

It turns out that any upper bound on $\Delta_P$ suffices when bounding the runtime. For example, it is a simple consequence of Jensen's inequality that we can take $\Delta_P$ to be $\sqrt{\lVert \psi \rVert_{P}}$ with $\psi = H$. A key technical contribution of this work is to reduce the conditions for determining whether a super-Grover runtime is possible to determining the log-Sobolev constant $\omega$, or Poincar\'e constant $\delta$, and the $P$-pseudo Lipschitz norm $\lVert H \rVert_{P}$ for cost Hamiltonian $H$. 

We summarize our main results below, which are the formal versions of Theorems \ref{thm:main-general-logsobolev} and \ref{thm:main-general-poincare}:
\begin{theorem}[Short-Path Runtime Under log-Sobolev]\label{thrm:MTG_RT}
Let $\Mcal = (\Xcal, P, \pi)$ be a reversible, aperiodic Markov chain, and let $H:\Xcal \rightarrow \R$ be a diagonal, $\Delta_P$-stable Hamiltonian (Definition \ref{defn:delta-p}) with ground state energy $E^{\star} < 0$, $P$-pseudo Lipschitz norm $\lVert H \rVert_{P}$. In addition, suppose $\Mcal$ has a log-Sobolev constant lower bounded by $\omega$. Then
 $\gamma$-spectral density (Definition \ref{defn:gamma-spectral}) is satisfied for
$$
    \gamma = \frac{\omega((1-\eta)E^{\star} - \mathbb{E}_{\pi}[H])^2}{\lVert H \rVert_{P}\ln(1/\pi(E^{\star}))}.
$$ 
Also if $b$ satisfies
\begin{align*}
&b< b^{\star} := \frac{2}{3}\gamma \omega\ln\left(\frac{1}{\pi(E^{\star})}\right)
\end{align*}
and \begin{align}
\label{eqn:ell_condition_ls}
\frac{1}{3\gamma \ln(1/\pi(E^{\star}))\sqrt{\mathbb{P}_{\pi}( E \leq (1-\eta)E^*)}} - \max\left(\frac{4\ln(\lvert \Xcal\rvert)}{\omega}, \frac{3\lvert E^{\star}\rvert^2(1-\eta)^2}{\Delta_P^2}\right)  \geq 2,
\end{align}
then there exists a short path algorithm that determines the ground state of $H$ over $\Xcal$ with running time
$$
\mathcal{O} \left(\textup{poly}(n)\omega^{-1}[\pi(E^{\star})^{-1}]^{\left(\frac{1}{2}-\frac{\eta(1-\eta)\lvert E^{\star}\rvert b}{2\ln(1/\pi(E^{\star}))\Delta_P}\right)}\right).
$$
Note that any upper bound on $\Delta_P$ suffices, for example one may use $\sqrt{\lVert H\rVert_{P}}$.
\end{theorem}
\begin{proof}
The proof 
is evident after combining the  statements of Corollary \ref{cor:total_runtime_bound}, Theorem \ref{thrm:b_log_sob}, Lemma \ref{lem:delta_alpha_eqv}, Lemma \ref{lem:short-path-overlap}, and Corollary \ref{cor:tail-herbst}
\end{proof}

We also present a variant of the above result that only relies on a Poincar\'e inequality.

\begin{theorem}[Short-Path Runtime Under Poincar\'e]\label{thrm:MTG_RT_poincare}
Let $\Mcal = (\Xcal, P, \pi)$ be a reversible, aperiodic Markov chain, and let $H: \Xcal \rightarrow \R$ be a diagonal, $\Delta_P$-stable Hamiltonian (Definition \ref{defn:delta-p}) with ground state energy $E^{\star} < 0$, $P$-pseudo Lipschitz norm $\lVert H \rVert_{P}$. In addition, suppose $\Mcal$ has a Poincar\'e constant $\delta$. Then
 $\gamma$-spectral density (Definition \ref{defn:gamma-spectral}) is satisfied for
\begin{align*}
    \gamma = \frac{\sqrt{\delta}((1-\eta)E^{\star} - \mathbb{E}_{\pi}[H])}{\sqrt{\lVert H \rVert_{P}}\ln(1/\pi(E^{\star}))},
\end{align*}
and suppose that $\frac{\delta}{2} > \pi(E^{\star})^{\gamma}$.
Also if $b$ satisfies
\begin{align*}
    b < b^{\star} :=\frac{\delta}{4}
\end{align*}
and
 \begin{align}
\label{eqn:ell_condition_pc}
\frac{1}{2\sqrt{\mathbb{P}_{\pi}( E \leq (1-\eta)E^*)}} - \max\left(\frac{4\ln(\lvert \Xcal\rvert)}{\delta}, \frac{3\lvert E^{\star}\rvert^2(1-\eta)^2}{\Delta_P^2}\right)  \geq 2,
\end{align}
then there exists a short path algorithm that determines the ground state of $H$ over $\Xcal$ with running time
$$
\mathcal{O} \left(\textup{poly}(n)\delta^{-1}[\pi(E^{\star})^{-1}]^{\left(\frac{1}{2}-\frac{\eta(1-\eta)\lvert E^{\star}\rvert b}{2\ln(1/\pi(E^{\star}))\Delta_P}\right)}\right).
$$
Note that any upper bound on $\Delta_P$ suffices, for example one may use $\sqrt{\lVert H\rVert_{P}}$.
\end{theorem}
\begin{proof}
The proof  is  evident after combining the statements of Corollary \ref{cor:total_runtime_bound}, Theorem \ref{thrm:b_poincare}, Lemma \ref{lem:short-path-overlap}, Lemma \ref{lem:delta_alpha_eqv}, and Corollary \ref{cor:tail-poinc}.
\end{proof}
\begin{remark}
\label{rem:conditions}
The conditions in \eqref{eqn:ell_condition_ls} and \eqref{eqn:ell_condition_pc} are more of a technical nature and actually very weak. Specifically, the two quantities that we are subtracting will likely be separated by exponential factors in $n$, e.g. $\mathbb{P}_{\pi}( E \leq (1-\eta)E^*)$ will be exponentially small for hard problems. This is  discussed in Section \ref{subsec:long_jump}. For similar reasons the additional condition on $\delta, \pi(E^{\star}), \gamma$ in Theorem \ref{thrm:MTG_RT_poincare} is also technical and very weak. This is discussed in Section \ref{sec:short_jump_subsec}.  
\end{remark}

In general, the log-Sobolev constant $\omega$ can be significantly smaller than the spectral gap of the chain $\delta$, however, we argue that this is not the case when Theorem \ref{thrm:MTG_RT} provides a super-Grover runtime. Specifically, Theorem \ref{thrm:MTG_RT} requires that $b$ is a constant, and by extension, implies we need $\omega^{-1} = \Theta(\ln(1/\pi(E^*)))$. For example, for  a very hard problem, where Markov chain search finds an optimal assignment with exponentially-small probability, i.e., $\Theta(\ln(1/\pi(E^*))) = \Theta(n)$, the condition on $b$ will imply that $\omega$ will be large, i.e.,  $\Omega\left(\frac{1}{\poly(n)}\right)$. Thus, in cases where it provides a super-Grover runtime, Theorem \ref{thrm:MTG_RT} asserts that we do not get a slower runtime by using $\omega$ instead of $\delta$.

We have the following evident corollary of the above results. This provides conditions under which the advantage over Grover with short-path does not asymptotically fall with the problem size.
\begin{corollary}[Conditions for Asymptotic Super-Grover Runtime]
\label{cor:super_grover_run}
Let $\Mcal = (\Xcal, P, \pi)$ be a reversible, aperiodic Markov chain, and let $H: \Xcal \rightarrow \R$ be a diagonal Hamiltonian with ground state energy $E^{\star} < 0$.
Suppose that the chain $\Mcal$ and cost-function $H$ are such that either Theorem \ref{thrm:MTG_RT} or \ref{thrm:MTG_RT_poincare} is satisfied for some $b = \Theta(1)$ and
\begin{align*}
    \frac{\lvert E^{\star}\rvert}{\Delta_{P}} = \Theta\left(\ln\left(1/\pi(E^{\star})\right)\right).
\end{align*} 
Then under Assumption \ref{asm:input-assumptions} the short path algorithm finds the minimizer of $H$ over $\Xcal$ with running-time bounded by
\begin{align*}
    \mathcal{O}^{\star}\left([\pi(E^{\star})^{-1}]^{\frac{1}{2} - c}\right),
\end{align*}
where
$$
    c = \frac{\eta(1-\eta)\lvert E^{\star}\rvert b}{2\Delta_{P}\ln(1/\pi(E^{\star}))} = \Theta(1).
$$
\end{corollary}

\subsection{Constructing Short Path Algorithms from Markov Chains}
\label{subsec:constructing_short_path_markov_chains}

This subsection details how the results from \cite{dalzell2022mind} can be generalized to the setting of reversible Markov chais. One of the main conditions from the aformentioned paper is that there should be a small number of low-energy states, effectively capturing that the underlying problem is hard. This is made precise through the spectral density condition generalized for an arbitrary chain $\Mcal$ (Definition \ref{defn:gamma-spectral}).
We show that the tail bound given by the spectral density condition is implied by pseudo Lipschitzness together with a functional inequality. This follows from existing results from Markov chain theory:
\begin{theorem}[Herbst's Argument, Adapted from Theorem 4.3 in \cite{lalley2013concentration} ]\label{Herbst} Suppose $\Mcal = (\Xcal, P, \pi)$ is a Markov chain with log-Sobolev constant $\omega$, and $f : \Xcal \rightarrow \mathbb{R}$ is $\lVert f \rVert_{P}$ pseudo-Lipschitz. Then,
$$
    \mathbb{P}_{\pi}[f \geq \mathbb{E}_{\pi}[f] + t] \leq e^{-\frac{\omega}{\lVert f \rVert_{P}}t^2}.
$$
\end{theorem}

\begin{theorem}[Adapted from Theorem 3.5 in \cite{lalley2013concentration}]\label{pc_tail} Suppose $\Mcal = (\Xcal, P, \pi)$ is a Markov chain with Poincar\'e constant $\delta$, and $f : \Xcal \rightarrow \mathbb{R}$  is $\lVert f \rVert_{P}$ pseudo-Lipschitz. Then,
$$
    \mathbb{P}_{\pi}[f \geq \mathbb{E}_{\pi}[f] + t] \leq e^{-\frac{\sqrt{\delta}}{\sqrt{\lVert f \rVert_{P}}}t}.
$$
\end{theorem}

The following corollaries are immediate, and reduce spectral density to pseudo Lipschitzness and a functional inequality.
\begin{corollary}[Log-Sobolev Inequality \& Pseudo-Lipschitz $\implies \gamma$-Spectral Density]
\label{cor:tail-herbst}
Suppose $\Mcal = (\Xcal, P, \pi)$ is a Markov chain with log-Sobolev constant $\omega$, and that the cost function $H$ is $\lVert H \rVert_{P}$ pseudo Lipschitz. Then,
$$
    \mathbb{P}_{\pi}[H \leq (1-\eta)E^{\star}] \leq \pi(E^{\star})^{\gamma},
$$
with
$$\gamma = \frac{\omega((1-\eta)E^{\star} - \mathbb{E}_{\pi}[H])^2}{\lVert H \rVert_{P}\ln(1/\pi(E^{\star}))}.$$
\end{corollary}

\begin{corollary}[Poincar\'e Inequality\& Pseudo-Lipschitz $\implies \gamma$-Spectral Density]
\label{cor:tail-poinc}
Suppose $\Mcal = (\Xcal, P, \pi)$ is a Markov chain with Poincar\'e constant $\delta$, and that the cost function $H$ is $\lVert H \rVert_{P}$ pseudo Lipschitz. Then,
$$
    \mathbb{P}_{\pi}[H \leq (1-\eta)E^{\star}] \leq \pi(E^{\star})^{\gamma},
$$
with
$$\gamma = \frac{\sqrt{\delta}((1-\eta)E^{\star} - \mathbb{E}_{\pi}[H])}{\sqrt{\lVert H \rVert_{P}}\ln(1/\pi(E^{\star}))}.$$
\end{corollary}

As noted in previous papers, the spectral density condition is relatively weak for hard problems. For example, suppose that for any constants $\gamma$ and $\eta$ it were not satisfied, but $\pi(E^*) = \Ocal(2^{-cn})$. Then Markov chain search can prepare an $\eta$ relative error, for $\eta$ arbitrarly close to $1$, approximate minimizer in time subexponential in $n$.

  \subsubsection{The Short Jump}
\label{sec:short_jump_subsec}
The short jump is defined as the preparation of $|\psi_b\rangle$ from $|\sqrt{\pi}\rangle$. The ability to find a good point to short-jump to (i.e., constant $b$ where a jump takes $\poly(n)$ time to make) is where the inherent speedup over Grover comes from. If we just decided to do only the short-jump and sample until we found the ground state, the algorithm would only be quadratically slower, due to amplitude amplification on the long-jump (which assuming good gap, costs $|\langle \psi_b|z^*\rangle|^{-1}$), than the full short path algorithm.  

The goal of this section is to determine conditions, using an initial Hamiltonian that is the discriminant of a reversible Markov chain, under which a short jump can be done efficiently. However, it does not determine whether such a short jump provides super-Grover runtime, which is the goal of the next subsection (Section \ref{subsec:long_jump}). The runtime of a short jump  is captured by the short path condition from \cite{dalzell2022mind, hastings2018shortPath}, where we present a natural generalization now.

\begin{definition}[$\theta$-Short Path Condition]
\label{defn:short-path-conditon}
Let $\Pi_{\perp} := \mathbb{I} - |\sqrt{\pi}\rangle\!\langle\sqrt{\pi}|$. %
Then, the $\theta$-short path condition holds for some constant $b > 0$ if
$$
        \GSE (\Pi_{\perp}H_b\Pi_{\perp})\geq -1 + \theta
$$
where \GSE ~denotes ``ground-state energy''.
\end{definition}
The short path condition was used in previous papers to prove a variety of other sufficient conditions for super-Grover runtime. For example, the short path condition implies a lower bound on  the spectral gap of $H_b$. %
\begin{restatable}[$\theta$-Short Path Condition $\implies$ $\theta$ Spectral Gap Bound, Adapted from Proposition 5 of \cite{dalzell2022mind}]{lemma}{shortPathToGapLem}
\label{lem:shortpath-gap}
If $H_b$ satisfies the $\theta$ short path condition, then $H_b$ has a unique ground state and the spectral gap is at least $\theta$, i.e., all excited states have energy at least $-1 + \theta$.
\end{restatable}

We also show that the short-path condition implies an overlap lower bound.
\begin{restatable}[$\theta$-Short Path Condition $\implies \frac{\theta}{4}$ Overlap]{lemma}{shortPathImpliesOverlap}
\label{lem:short-path-overlap}
If $H_b$ satisfies the $\theta$ short path condition with $\theta \leq 1$, then $\lvert \langle \sqrt{\pi} | \psi_b\rangle\rvert \geq  \frac{\theta}{4}$.
\end{restatable}
Thus one can view it as a convenient way of combining the conditions on the gap and overlap.

We stress that the short path condition's main purpose is to show the existence of an efficient ``short jump'', i.e. transforming states $|\sqrt{\pi}\rangle \rightarrow |\psi_b\rangle$. %
The other consequence derived from the short path condition was of a more technical nature. This technical condition, called the ``small-ground-state-energy shift condition'', was for bounding the cost of the long jump. Alternative conditions for this consequence of the short-path condition  are presented in the next subsection. 

We summarize the implications of the short path condition on the short-jump runtime in the following result.

\begin{theorem}[Sufficient Conditions for Efficient Short Jump]
\label{thm:short-jump-eff}
Suppose $\Mcal = (\Xcal, P, \pi)$ is a reversible Markov chain with spectral gap that is at least inverse-polynomial in $n$, and $H$ is a cost Hamiltonian satisfying the $\theta$ short path condition at $b > 0$ independent of the problem size $n$. If $\theta = \Omega\left(\frac{1}{\poly(n)}\right)$, then there exists quantum algorithm for preparing an $\varepsilon$-approximation to $|\psi_b\rangle$ starting with $|\sqrt{\pi}\rangle$, which makes $\textup{poly}(n, \log(1/\varepsilon))$ queries to block-encodings of $D(P)$ and $H_b$.
\end{theorem}
\begin{proof}
The result will follow if we can show that 
\begin{align}
    \min(\text{Gap}(-D(P)), \text{Gap}(H_b))\lvert \langle \sqrt{\pi} |\psi_b\rangle\rvert
\end{align}
is inverse-polynomial in $n$. The assumption on the Markov chain gap implies that $\text{Gap}(-D(P))$ is inverse-polynomial in $n$. Lemma \ref{lem:shortpath-gap} implies that $\text{Gap}(H_b)= \Omega\left(\frac{1}{\poly(n)}\right)$, and Lemma \ref{lem:short-path-overlap} implies $\lvert \langle \sqrt{\pi}| \psi_b\rangle\rvert = \Omega\left(\frac{1}{\poly(n)}\right)$.

\end{proof}

The rest of the main results of this subsection determine a range of $b$'s under which the short-path condition can be shown to hold given the $\gamma$-spectral density condition. Additionally, we provide the value of $\theta$ that is implied in each case. In Theorems \ref{thrm:MTG_RT} and \ref{thrm:MTG_RT_poincare}, the $b^{\star}$'s are a result of theorems derived in this section.
Specifically, Theorems \ref{thrm:b_log_sob} and  \ref{thrm:b_poincare}, provide the bounds on $b^{\star}$ presented in Theorems \ref{thrm:MTG_RT} and \ref{thrm:MTG_RT_poincare}, respectively. Theorem \ref{thrm:b_log_sob} shows that a log-Sobolev inequality on $\Mcal$ and  $\gamma$-spectral density for $H$ suffice to derive a range of $b$'s over the short-path condition is satisfied. In the case of only a Poincar\'e inequality and $\gamma$-spectral density, Theorem \ref{thrm:b_poincare} shows that the range of such $b$'s is bounded by the Poincar\'e constant. In either case, $\theta$ is shown to be proportional to the LS constant or Poincar\'e constant, respectively. Hence, if either the LS constant or Poincar\'e constant is at least inverse-polynomial in $n$, then Theorem \ref{thrm:b_log_sob} or  \ref{thrm:b_poincare} plus Theorem \ref{thm:short-jump-eff} show that $|\psi_b\rangle$ can be efficiently prepared for some range of $b$'s.

It is possible that if certain conditions are not met, then $b^{\star}$ falls with $n$, leading to an asymptotically vanishing range of $b$'s. Our results imply sufficient conditions under which we can show this does not occur. Additionally, the numerical results in Section \ref{sec:numerical_results} will reveal that the provable values of $b^{\star}$ are typically overly pessimistic.

From a bird's-eye view, the role of the functional inequality in bounding $b^{\star}$ is to upper bound a metric or divergence between $\psi_b^2$, the $\ell_2$ distribution of $|\psi_b\rangle$ in the computational basis, and $\pi$. The variational definition of the corresponding metric or divergence is used to to establish the lower bound. Together, these bounds, via proof by contradiction, are sufficient to derive a range of $b$'s where short path condition holds.

\begin{lemma}
\label{lemma:KL-lemma}
Suppose $\Mcal = (\Xcal, P, \pi)$ is a Markov chain that satisfies a log-Sobolev inequality with constant $\omega$. Then, for all quantum states $|\psi\rangle$ one has
$$
        \frac{1 - \langle\psi|D(P)|\psi\rangle}{\omega} \geq \textup{KL}(\psi^2 \| \pi),
$$
where $\textup{KL}(\cdot \| \cdot)$ denotes the Kullback–Leibler divergence, and $\psi^2$ is the $\ell_2$ distribution of $|\psi\rangle$ in the computational basis.
\end{lemma}

\begin{proof}
From a straightforward calculation, it follows
\begin{align}
\label{eqn:dirichlet_form_reduction}
\mathcal{D}(\psi, \psi) &= 
\langle \psi, (I-P) \psi\rangle_{\pi}\\
&=\mathbb{E}_{\pi}(\psi^{\mathsf{T}}\psi) - \mathbb{E}_{\pi}(\psi^{\mathsf{T}}P\psi)\\
&=\sum_{x \in \Xcal}(\sqrt{\pi(x)}\psi(x))^2 - \sum_{x \in \Xcal} \sqrt{\pi(x)}\psi(x)D(P)_{xy}\sqrt{\pi(y)}\psi(y)\\
&=\lVert \sqrt{\pi}\psi\rVert_2^2 - \langle  \sqrt{\pi}\psi|D(P)| \sqrt{\pi}\psi\rangle.
\end{align}
Now,
$$
\lVert \sqrt{\pi}\psi\rVert_2^2 - \langle  \sqrt{\pi}\psi|D(P)| \sqrt{\pi}\psi\rangle  \geq \omega (\mathbb{E}_{\pi}(\psi^2\ln(\psi^2)) - \mathbb{E}_{\pi}(\psi^2\ln(\mathbb{E}\psi^2))),
$$
and  
$$
\mathbb{E}_{\pi}(\psi^2\ln(\psi^2)) - \mathbb{E}_{\pi}(\psi^2\ln(\mathbb{E}\psi^2)) = \sum_{x \in \Xcal} \pi(x)\psi^2(x)\ln(\psi^2(x)) - \lVert \sqrt{\pi}\psi\rVert_2^2\ln(\lVert \sqrt{\pi}\psi\rVert_2^2).
$$

Consider $\psi = \frac{\psi'}{\sqrt{\pi}}$ with $\lVert \psi'\rVert_2 = 1$. Then
$$
1 - \langle\psi|D(P)|\psi\rangle \geq \omega \sum_{x \in \Xcal} \psi^2(x)\ln\left(\frac{\psi(x)^2}{\pi(x)}\right) = \omega \text{KL}(\psi^2 || \pi).
$$
\end{proof}

\begin{lemma}
\label{lem:kl_spec_dens_lower_bound}
Suppose the $\gamma$ spectral density condition is satisfies, then for any quantum state $\psi$
\begin{align*}
    \textup{KL}(\psi^2 | \pi) \geq -\gamma \ln\left(\frac{1}{\pi^{\star}}\right) \langle \psi |g_{\eta}\left(\frac{H}{E^*}\right)|\psi \rangle - 1.
\end{align*}
\end{lemma}
\begin{proof}
Let $F(E)$ be the cumulative distribution function for the cost function $H$. Define $F_{\eta}(E)$ to be the cumulative distribution function for $g_{\eta}\left(\frac{H}{|E^{\star}|}\right)$. Then,
\[F_{\eta}(v)=
\begin{cases} 
       0 & v < -1 \\
     F(E^{\star}(1-\eta - \eta v))\leq (\pi(E^\star))^ {\gamma} & -1\leq v<0 \\
      1 & v\geq 0
   \end{cases}
\]
where we assume that the probability of low energy states under $\pi$ is low, i.e., 
$$
   F((1-\eta)E^{\star})=\mathbb{P}_{\pi}(E\leq (1-\eta)E^{\star})\leq \left(\pi(E^{\star})\right)^{\gamma} 
$$
for $0<\gamma\leq 1$ where $\pi(E^{\star})$ is given by
$$
    \pi(E^{\star}) :=  \mathbb{P}_{\pi}(E = E^{\star}) = \sum_{\substack{\{x\in \Xcal : E(x) = E^{\star}\}}} \pi(x).
$$
If this bound does not hold, then for a very small $\eta$, there is a high probability mass for the states with energy closer to $E^{\star}$. Therefore, we can find an approximate optimizer by randomly sampling from $\pi(x)$ in time sub-exponential in $\log(1/\pi(E^{\star}))$. Note that when $\pi$ is uniform and the ground state is non-degenerate, we recover the original condition of \cite[see, Lemma 5]{dalzell2022mind} since $$F((1-\eta) E^{\star})\leq 2^{-\gamma n}.$$ 

For a quantum state $|\psi\rangle$ let $\psi^2$ denote its $\ell_2$ distribution in the computational basis. Applying Donsker and Varadhan's variational formula \cite{donsker1983asymptotic} for $\textup{KL} (\psi^2\| \pi)$, we may write
$$
    \textup{KL} (\psi^2\|\pi)=\sup_{f \in \mathcal{F}{}{}}\{ \mathbb{E}_{\psi^2}[f(x)]-\ln(\mathbb{E}_\pi[\exp(f(x))]\},
$$
where $\mathcal{F}$ denotes the set of all measurable functions.
Choosing $f(x) = -\gamma\ln\left(\frac{1}{\pi^{\star}}\right)g_{\eta} \left(\frac{H(x)}{E^{\star}}\right)$ and defining  $U_{\psi} =\mathbb{E}_{\psi^2}[g_{\eta}\left(\frac{H}{E^{\star}}\right)]$, it follows
\begin{align*}
    \textup{KL} (\psi^2\|\pi)&\geq -\gamma\ln\left(\frac{1}{\pi^{\star}}\right)\mathbb{E}_{\psi^2} \left[g_{\eta} \left(\frac{H(x)}{E^{\star}} \right)\right]-\ln \left(\mathbb{E}_\pi\left[e^{-\gamma\ln\left(\frac{1}{\pi^{\star}}\right)g_{\eta} \left(\frac{H(x)}{E^{\star}} \right)} \right] \right)\\
      &\geq-\gamma \ln\left(\frac{1}{\pi^{\star}}\right) U_{\psi} - 1.
\end{align*}

The final inequality follows from 
\begin{align*}
    \mathbb{E}_\pi \left[e^{-\gamma\ln\left(\frac{1}{\pi^{\star}}\right)g_{\eta}(\frac{H(x)}{E^{\star}})} \right] &= \sum_{g_{\eta}(\frac{H(x)}{E^{\star}}))= 0}\pi(x)+\sum_{g_{\eta}(\frac{H(x)}{E^{\star}})\neq 0}\pi(x)e^{-\gamma \ln\left(\frac{1}{\pi^{\star}}\right)g_{\eta}(\frac{H(x)}{E^{\star}})}\\
    &\leq \pi(E(x)=0) + \pi(E(x)\neq 0)e^{\gamma \ln\left(\frac{1}{\pi^{\star}}\right)}\\
    &= \pi(E(x)=0) + \left(\pi^{\star}\right)^{\gamma} e^{\gamma \ln\left(\frac{1}{\pi^{\star}}\right)}\\
    &\leq 2,
\end{align*}
as we assume $\pi(g_{\eta}(\frac{H}{E^{\star}})\neq 0)\leq (\pi^{\star})^{\gamma}$ and $g_\eta(\frac{H}{E^{\star}})\geq -1$.

Thus we have the following lower bound on the KL divergence:
\begin{align}
\label{eqn:kl-lower-bound}
     \textup{KL} (\psi^2\|\pi) \geq -\gamma \ln\left(\frac{1}{\pi^{\star}}\right) U_{\psi} - 1.
\end{align}
\end{proof}

\begin{theorem}[Sufficient $b$ For Spectral Density To Imply Short Path Given Log-Sobolev]
\label{thrm:b_log_sob}
Suppose $\Mcal = (\Xcal, P, \pi)$ is a reversible Markov chain that satisfies an $\omega$ log-Sobolev inequality. If
\begin{align*}
 b< \frac{2}{3}\gamma \omega\ln\left(\frac{1}{\pi(E^{\star})}\right)
\end{align*}
    then $\gamma$ spectral density implies an $\frac{\omega}{2}$ short path condition.
\end{theorem}
\begin{proof}

Let $U_{\psi} = \mathbb{E}_{\psi^2}[g_{\eta}\left(\frac{H}{E^*}\right)]$. From Lemma \ref{lem:kl_spec_dens_lower_bound}, we have the following lower bound on the KL divergence:
\begin{align}
\label{eqn:kl-lower-bound}
     \textup{KL} (\psi^2\|\pi) \geq -\gamma \ln\left(\frac{1}{\pi^{\star}}\right) U_{\psi} - 1.
\end{align}
We also have the following upper bound the KL divergence from Lemma \ref{lemma:KL-lemma}:
\begin{align}
\label{eqn:kl-upper-bound}
        \frac{1 - \langle\psi|D(P)|\psi\rangle}{\omega} \geq \textup{KL}(\psi^2 \| \pi),
\end{align}
where $\omega$ is the log-Sobolev constant of $P$.

The following argument attempts to find an upper bound $b^*$ on the $b$, such that for all $b < b^*$ the two bounds above become contradicting if short path is not satisfied.

Suppose for contradiction that the $\frac{\omega}{2}$ short path condition is violated at $b$, i.e.
$$
    \GSE(\Pi_{\perp}H_b\Pi_{\perp}) < -1 + \frac{\omega}{2},
$$
where $\omega$ is the log-Sobolev constant of $P$. If $|\psi_{b}'\rangle$ is the ground state of $\Pi_{\perp}H_b\Pi_{\perp}$, then
\begin{align}
\label{eqn:short-path-viol}
-1 + \frac{\omega}{2}  >  \langle \psi_{b}' |H_b | \psi_{b}'\rangle  =- \langle \psi_{b}' |D(P)| \psi_{b}'\rangle + b\langle \psi_{b}' |G_{\eta}| \psi_{b}'\rangle =-\langle \psi_{b}' |D(P)| \psi_{b}'\rangle + bU_{\psi_b'},
\end{align}
which implies 
\begin{equation}\label{eqn:u_bound}
     1 - \langle \psi_{b}' | D(P) | \psi_{b}'\rangle  < \frac{\omega -2bU_{\psi_b'}}{2}.
\end{equation}
Thus if the short path condition is violated, then the $\KL$ upper bound above reduces to
$$
     \frac{\omega - 2bU_{\psi_b'}}{2\omega} \geq \textup{KL}(\psi_b'^2\| \pi).
$$

We also have generally that a Poincar\'e inequality (as log-Sobolev inequality implies a Poincar\'e inequality with the same constant) implies that for any $|\psi\rangle$
\begin{align}\label{e:ls_constant_ineq}
\langle\psi|-D(P) |\psi\rangle \geq -1  + \delta \geq -1  + \omega,
\end{align}
where $\delta$ denotes the Poincar\'e constant of  $P$, so if short path condition is violated combining the above with Equation \eqref{eqn:short-path-viol} gives
 \begin{align}
 \label{eqn:omega_u}
\frac{\omega}{2b} < -U_{\psi_b'},
 \end{align}
where $U_{\psi_b'} < 0$ by construction. Also, by construction $-U_{\psi_b'} \leq 1$.

Hence, when the short path condition is violated, Equations \eqref{eqn:kl-lower-bound} and  \eqref{eqn:kl-upper-bound} imply that
\begin{align}
\label{eqn:violated_short_path_KL}
    0 \geq-\gamma \ln\left(\frac{1}{\pi(E^{\star})}\right) U_{\psi_b'} -1- \frac{\omega - 2bU_{\psi_b'}}{2\omega}.
\end{align}
As mentioned earlier, we want to a range of $b$'s that contradicts this inequality, and so we solve
\begin{align*}
0 < -\gamma \ln\left(\frac{1}{\pi(E^{\star})}\right) U_{\psi_b} -1- \frac{\omega -2bU_{\psi_b'}}{2\omega},
\end{align*}
which yields
\begin{align*}
b < \frac{3\omega}{2U_{\psi_b'}} + 2\gamma\omega\ln(1/\pi(E^{\star}))
\end{align*}  
and using Equation \eqref{eqn:omega_u} since we want to consider the smallest right hand side:
$$b< \frac{2}{3}\gamma \omega\ln\left(\frac{1}{\pi(E^{\star})}\right).$$ 
Thus, the short path condition must hold for the values of $b$ given in theorem statement.
\end{proof}

It is reasonable to question if a simpler Poincar\'e inequality for $\Mcal = (\Xcal, P, \pi)$ which only depends on the spectral gap $\delta$ of $P$ would suffice to get a bound on $b$. Unfortunately, it appears that this does not provide a useful bound unless $\delta$ is constant. %

\begin{lemma}
\label{lemma:TV-lemma}
Suppose $\Mcal = (\Xcal, P, \pi)$  is a Markov chain that satisfies a Poincar\'e inequality with constant $\delta$. Then, for all quantum states $|\psi\rangle$ one has
$$
        \frac{1 - \langle\psi|D(P)|\psi\rangle}{\delta} \geq [\textup{TV}(\psi^2 , \pi)]^2,
$$
where $\textup{TV}(\cdot , \cdot)$ denotes the total variation distance,  and $\psi^2$ is the $\ell_2$ distribution of $|\psi\rangle$ in the computational basis.
\end{lemma}
\begin{proof}
The proof roughly follows that of Lemma \ref{lemma:KL-lemma}, i.e. using Equation \eqref{eqn:dirichlet_form_reduction}, but instead uses $\pi$-Variance. From the Poincar\'e inequality:
\begin{align*}
\frac{1 - \langle\psi|D(P)|\psi\rangle}{\delta} = \frac{\mathcal{D}(\psi/\sqrt{\pi},\psi/\sqrt{\pi})}{\delta} &\geq \text{Var}_{\pi}[\psi/\sqrt{\pi}]=\mathbb{E}_{\pi}[\psi^2/\pi] - (\mathbb{E}_{\pi}[\psi/\sqrt{\pi}])^2 \\&= 1 - |\langle\sqrt{\pi}| \psi\rangle|^2 \\&\geq [\text{TV}(\pi , \psi^2)]^2.
\end{align*} 
\end{proof}

The following theorem provides a sufficient bound on $b$ for $\gamma$-spectral density to imply the short-path condition. There is one addtional condition relating $\pi(E^{\star}), \gamma$ and $\delta$. However, this condition is satisfied whenever it makes sense to apply the generalized short-path algorithm. Specifically, we need the short jump to be efficient, i.e. $\delta = \Omega\left(\frac{1}{\textup{poly}(n)}\right)$ and the problem to be hard, i.e. $\pi(E^{\star})$ is falling super polynomially.
\begin{theorem}[Sufficient $b$ For Spectral Density To Imply Short Path Given Poincar\'e]
\label{thrm:b_poincare}
Suppose $\Mcal = (\Xcal, P, \pi)$ is a reversible Markov chain that satisfies a $\delta < 1$ Poincar\'e inequality. Suppose
\begin{align}
    b < \frac{\delta}{4}.
\end{align}
If the  $\gamma$-spectral density is satisfied with $\frac{\delta}{2} >\pi(E^{\star})^{\gamma}$, then the $\frac{\delta}{2}$ short-path condition is also satisfied.
\end{theorem}
\begin{proof}

The proof follows a similar structure as Theorem \ref{thrm:b_log_sob}. We have the variational definition of TV:
$$
\text{TV}(\pi, \psi) = \textup{sup}_{f: \lVert f \rVert_{\infty} \leq 1}\frac{1}{2}(\mathbb{E}_{\psi^2}[f(x)]- \mathbb{E}_{\pi}[f(x)]).
$$
We can choose $f(x) = -g_{\eta}(H(x)/E^*)$, which satisfies $\Vert f \rVert_{\infty} \leq 1$.  Let $$U_{\psi} = \mathbb{E}_{\psi^2}\left[g_{\eta}\left(\frac{H}{E^*}\right)\right].$$

The $\gamma$-spectral density condition implies
\begin{align*}
    & \mathbb{E}_{\pi}[f(x)] \leq \pi(E \leq (1-\eta)E^{\star}) \leq \pi(E^{\star})^{\gamma},
\end{align*}
so
\begin{align}
\label{eqn:tv-lower-bound}
\text{TV}(\pi, \psi^2) \geq  \frac{1}{2}( -U_{\psi} - \pi(E^{\star})^{\gamma}).
\end{align}

We also have the following upper bound on the TV from Lemma \ref{lemma:TV-lemma}:
\begin{align}
\label{eqn:tv-upper-bound}
        \frac{1 - \langle\psi|D(P)|\psi\rangle}{\delta} \geq [\textup{TV}(\psi^2 \| \pi)]^2,
\end{align}
where $\delta$ is the Poincar\'e constant of $P$.

Let $|\psi_{b}'\rangle$ be the ground state of $\Pi_{\perp}H_b\Pi_{\perp}$
Suppose for contradiction that the $\frac{\delta}{2}$ short path condition is violated at $b$, i.e.
$$
    \GSE(\Pi_{\perp}H_b\Pi_{\perp}) < -1 + \frac{\delta}{2},
$$
where $\delta$ is the spectral gap of $P$. Thus
\begin{align}
\label{eqn:del-short-path-viol}
-1 + \frac{\delta}{2}  >  \langle \psi_{b}' |H_b | \psi_{b}'\rangle  =- \langle \psi_{b}' |D(P)| \psi_{b}'\rangle + b\langle \psi_{b}' |G_{\eta}| \psi_{b}'\rangle =-\langle \psi_{b}' |D(P)| \psi_{b}'\rangle + bU_{\psi_b'},
\end{align}
which implies 
\begin{equation}\label{eqn:u_bound}
     1 - \langle \psi_{b}' | D(P) | \psi_{b}'\rangle  < \frac{\delta -2bU_{\psi_b'}}{2}.
\end{equation}
Thus if short path is violated, then the total variation distance upper bound above reduces to
$$
     \frac{\delta - 2bU_{\psi_b'}}{2\delta} \geq [\textup{TV}(\psi_b'^2, \pi)]^2.
$$

If $P$ satisfies a $\delta$ Poincar\'e inequality, then we have that for any $|\psi\rangle \perp |\sqrt{\pi}\rangle$
\begin{align}\label{e:ls_constant_ineq}
\langle\psi|-D(P) |\psi\rangle \geq -1  + \delta.
\end{align}
If short path is violated, then combining the above with Equation \eqref{eqn:del-short-path-viol} gives
 \begin{align}
 \label{eqn:omega_u_pon}
\frac{\delta}{2b} < -U_{\psi_b'},
 \end{align}
where $U_{\psi_b'} < 0$ by construction. Also, by construction $-U_{\psi_b'} \leq 1$.

The above, our assumption on $\frac{\delta}{2} > \pi(E^{\star})^{\gamma}$, assumption that $b < 1$, and Equation \eqref{eqn:tv-lower-bound} imply
\begin{align*}
[\text{TV}(\pi, \psi^2)]^2 \geq  \frac{1}{4}\left(-U_{\psi_b'} - \pi(E^{\star})^{\gamma}\right)^2.
\end{align*}

Hence, when the short path condition is violated, Equations \eqref{eqn:tv-lower-bound} and \eqref{eqn:tv-upper-bound} imply that
\begin{align*}
\label{eqn:violated_short_path_TV}
    0 \geq  \frac{1}{4}(-U_{\psi_b'}  - \pi(E^{\star})^{\gamma})^2 - \frac{\delta - 2bU_{\psi_b'}}{2\delta}.
\end{align*}
As mentioned earlier, we want to find a range of $b$'s that contradicts this, so we solve
\begin{align*}
0 <\frac{1}{4}(-U_{\psi_b'}  - \pi(E^{\star})^{\gamma})^2 - \frac{\delta -2bU_{\psi_b'}}{2\delta},
\end{align*}
which reduces to

\begin{align*}
&0 < \frac{1}{4}[ U_{\psi_b'}^2 
+2U_{\psi_b'}\pi(E^{\star})^{\gamma} + \pi(E^{\star})^{2\gamma}] - \frac{\delta -2bU_{\psi_b'}}{2\delta}\\
&    0 < \delta[ U_{\psi_b'}^2 + 2U_{\psi_b'}\pi(E^{\star})^{\gamma} + \pi(E^{\star})^{2\gamma}] - 2(\delta -2bU_{\psi_b'})\\
& 0 < \delta [-U_{\psi_b'}]^2 - (2\delta\pi(E^{\star})^{\gamma} + 4b)[-U_{\psi_b}'] + \delta(\pi(E^{\star})^{2\gamma} -2) 
\end{align*}

and using Equation \eqref{eqn:omega_u_pon} since we want to consider the smallest r.h.s.:
\begin{align*}
&0 < \frac{\delta^3}{4b^2} - (2\pi(E^{\star})^{\gamma} + 4b)\frac{\delta}{2b} + \delta(\pi(E^{\star})^{2\gamma} -2)\\
&0 < \delta^3 - (2\delta\pi(E^{\star})^{\gamma} + 4b)2\delta b + \delta(\pi(E^{\star})^{2\gamma} -2)4b^2\\
&0 < [-16 + 4 \pi(E^{\star})^{2\gamma}]b^2 + 4\delta \pi(E^{\star})^{\gamma}b + \delta^2
\end{align*}

so using the positive root $b^{\star}$:

\begin{align*}
&b^{\star}=\frac{\delta \pi(E^{\star})^{\gamma}}{2(4-\pi(E^{\star})^{2\gamma})} - \frac{\sqrt{16\delta^2\pi(E^{\star})^{2\gamma}- 4[-16+4\pi(E^{\star})]\delta^2}}{2[-16+4\pi(E^{\star})^{2\gamma}]}\\
&=\frac{\delta \pi(E^{\star})^{\gamma}}{2(4-\pi(E^{\star})^{2\gamma})} + 4\delta \frac{\sqrt{\pi(E^{\star})^{2\gamma} -\pi(E^{\star}) + 4}}{2[16-4\pi(E^{\star})^{2\gamma}]}\\
&\geq \frac{\delta}{4}. 
\end{align*}

Thus, the short path condition must hold for the values of $b$ given in theorem statement. 
\end{proof}

It is also natural to ask if a modified log-Sobolev inequality would work. The apparent issue is with relating the Dirichlet form to the energy with respect to $D(P)$. Although we could not prove a generic bound in terms of modified log-Sobolev constant, we can use it to derive a lower bound on the standard log-Sobolev constant by using the following thoerem.
\begin{theorem}[Theorem 1 in \cite{salez2022upgradingmlsilsireversible}]
\label{thm:modified-LSI}
    Let $\omega_{\textup{LS}}$ and $\omega_{\textup{MLS}}$ respectively denote the log-Sobolev constant and the modified log-Sobolev constant of a Markov chain $\Mcal = (\Xcal,P,\pi)$. If $\Mcal$ is reversible, then
    \[
    \omega_{\textup{LS}}\geq \frac{\omega_{\textup{MLS}}}{\log(1/p)}
    \]
    where $p$ is the smallest non-zero element in $P$.
\end{theorem}
Note that for general Markov chains $p$ can be exponentially small. However, in certain cases such as single-site Glauber dynamics on bounded degree graphs, $p$ is only polynomially small, e.g. the antiferomagnetic $2$-spins in \cite[Proof of Theorem 1.1]{chen2023optimal}. Therefore, in these special cases, our main theorem implies a super-quadratic speed up although the speedup term $c$ might be falling with $n$.

If the conditions of Theorem \ref{thm:short-jump-eff} are satisfied (potentially by using Theorems \ref{thrm:b_log_sob} or \ref{thrm:b_poincare}), then we only need to show that $\lVert \Pi^{\star}|\psi_b\rangle\rVert_{2}^{-1}$ is exponentially smaller (say by $\pi(E^*)^{c/2}$ for some constant $c$) than $\lVert \Pi^{\star}|\sqrt{\pi}\rangle\rVert_{2}^{-1}$ to achieve a super-quadratic speedup over Markov chain search. This is the task of bounding the runtime of the \emph{long jump}.

\subsubsection{The Long Jump}
\label{subsec:long_jump}
The long jump is the preparation of an optimal solution in $\Pi^{\star}$ given $|\psi_b\rangle$, where ideally  $\lVert \Pi^{\star}|\psi_b\rangle\rVert_{2}^2 \geq \pi(E^*)^{1-c}$ for some constant $c > 0$. As discussed in the previous subsection, the only quantum speedup from this step is quadratic and is due to amplitude amplification. That is, the quantum short-jump plus classical sampling (classical long-jump) would cost $\Ostar(\pi(E^*)^{-(1-c)})$, provided that the overlap condition just mentioned is satisfied. Accordingly, a quantum short-jump applied to the discriminant of the Markov chain could still provide a nontrivial speedup. The goal of this subsection is to determine conditions on the cost Hamiltonian in terms of the Markov chain $\Mcal$ under which $\lVert \Pi^{\star}|\psi_b\rangle\rVert_{2}^2 \geq \pi(E^*)^{1-c}$ holds for some constant $c$, or more generally derive a formula for a problem-size-dependent $c(n)$ in terms of problem parameters. If we combine the results of  this section with the known runtime of amplitude amplification, then we get the runtime stated for the generalized short path algorithm (Equation \eqref{eqn: runtime}).

In order to lower bound $\lVert \Pi^{\star}|\psi_b\rangle\rVert_{2}^2$, we follow \cite{dalzell2022mind} in approximating the ground-state projector $|\psi_b\rangle\!\langle\psi_b|$ through the use of a simple degree-$\ell$ polynomial. This also makes use of a technical condition that ensures the ground state energy of $H_b$,  $E_b$, does not decrease too much from the ground state energy of $-D(P)$. Unlike \cite{dalzell2022mind},  we derive this condition without explicitly using the short path condition (Definition \ref{defn:short-path-conditon}), simplifying some of the presentation and weakening some of the algorithm's assumptions. %

At a high-level, under various problem-specific conditions, we can bound $\lVert \Pi^{\star}|\psi_b\rangle\rVert_{2}^2$ by instead bounding the relative \emph{localization enhancement} on optimal solutions attained by $|\psi_b\rangle$ over $|\sqrt{\pi}\rangle$, i.e. lower bounding
\begin{align*}
     \frac{|\langle z^{\star} | \psi_b\rangle|}{|\langle z^{\star}|\sqrt{\pi}\rangle|},
\end{align*}
for any optimal solution $z^{\star}$. If the short jump from $|\sqrt{\pi}\rangle$ to $|\psi_b\rangle$ can be done in polynomial time, then the relative improvement in $\Ostar$ runtime will be determined by the above ratio. %

We start with the following lower bound that relates the localization of $|\psi_b\rangle$ on optimal solutions to $|\sqrt{\pi}\rangle$ and $H_b$.
\begin{restatable}[Overlap To Coupling Lower Bound]{lemma}{longJumpBound}
\label{lem:long_jump_overlap_lower_bound}
Let $\Mcal = (\Xcal, P, \pi)$ be a reversible, aperiodic Markov chain with $|\Xcal| = V$. Let $\nu$ be a lower bound on either the log-Sobolev or Poincar\'e constant of $P$. Suppose there exists an integer $\ell$ such that
\begin{align*}
&\frac{4\ln(V)}{\nu}  \leq \ell \leq \frac{\langle \psi_b|\sqrt{\pi}\rangle}{b\sqrt{\mathbb{P}_{\pi}( E \leq (1-\eta)E^*)}}.
\end{align*}
If $|\psi_b\rangle$ is the unique ground state of $H_b$ and the spectral gap of $H_b$ is lower bounded by $\frac{\nu}{2}$, then for any optimal assignment $z^*$:
\begin{align*}
    \langle \psi_b | z^*\rangle = \Omega\left(\langle \sqrt{\pi}|(-H_b)^{\ell}|z^*\rangle\right),
\end{align*}
asymptotically in $V$.
\end{restatable}

We now recall the definition of $\alpha$-subdepolarizing from \cite{dalzell2022mind} but generalized to arbitrary Markov chains. The subdepolarizing property helps us to explicitly remove $H_b$ from the above bound and express things only in terms of $\pi(E^{\star})$ and problem parameters.
\begin{definition}[$\alpha_{P}$-subdepolarizing]
Let $\Mcal = (\Xcal, P, \pi)$ be a Markov chain. The pair $(\Mcal, H)$ satisfies the $\alpha$-subdepolarizing property if for $f(x) = - g_{\eta}(-x)$ (as defined earlier), the following holds for any set of constants $0 < c_1, \dots, c_{T} < 1, \forall T\in \mathbb{N}$:
$$
    \expP\prod_{t=1}^{T}f\left(\frac{c_tH(y)}{E^{\star}}\right) \geq \prod_{t=1}^{T} f\left(\frac{c_t(1-\alpha_{P})H(x)}{E^{\star}}\right),
$$
where $E^{\star}$ is the ground state energy of $H$.
\end{definition}

As shown in the proof of Proposition 3 of \cite{dalzell2022mind}, if $\Delta_{P}$-stability holds, then $\alpha_{P}$-subdepolarizing is satisfied. However, the converse also holds.
\begin{restatable}
[$\Delta_P$ stable $\iff \alpha_P$ subdepolarizing ]{lemma}{equivDepol}
\label{lem:delta_alpha_eqv}
Let $\Mcal = (\Xcal, P, \pi)$ be a Markov chain. The pair $(\Mcal, H)$ is $\Delta_P(\eta)$-stable if and only if it satisfies  $\alpha_{P}$-subdepolarizing.  They are related by the equation $\alpha_{P} = \frac{\Delta_{P}(\eta)}{\lvert E^{\star} \rvert (1-\eta)}$.
\end{restatable}

We now briefly remark on some useful upper bounds on $\Delta_{P}(\eta)$ that we alluded to earlier. However, using too loose of an upper bound may result in the runtime analysis not indicating a speedup. Note that choice of the bound on $\Delta_{P}$  is not actually used by Algorithm~\ref{alg:generalized-short-path}.
\begin{restatable}[Upper bounds on $\Delta_{P}$]{lemma}{deltaUpper}
If $H$ has $P$ pseudo-Lipschitz norm $\lVert H \rVert_{P}$, then for $\eta \in (0, 1)$, 
\begin{align*}
    \sqrt{\lVert H \rVert_{P}} \geq \Delta_{P}(\eta).
\end{align*}
Furthermore, if 
\begin{align*}
    \expP[H(y)] \leq  H(x)+\tilde{\Delta}_{P},
\end{align*}
then
\begin{align*}
     \tilde{\Delta}_{P} \geq \Delta_{P}(\eta).
\end{align*}
\end{restatable}

Our next result generalizes \cite[Lemma 3]{dalzell2022mind}, which uses $\alpha_{P}$ subdepolarizing (or equivalently $\Delta_P$ stability)  to lower bound the coupling term $\langle \sqrt{\pi}|(-H_b)^{\ell}|z^*\rangle$ in Lemma \ref{lem:long_jump_overlap_lower_bound}, removing the explicit dependence on $H_b$. This leads to a quantity that effectively determines how much more localized (in a relative sense) $\psi_b$ is on a given $z^*$ than $\sqrt{\pi}$. Note that obviously $\langle \sqrt{\pi}|(-H_b)^{\ell}|z^*\rangle$ becomes the overlap between $|\sqrt{\pi}\rangle$ and $|z^*\rangle$ when $b = 0$. However, one can show that this quantity is at least $\langle \sqrt{\pi}| z^*\rangle$ for $b > 0$, but the stability condition is used to determine when it is strictly greater.
\begin{lemma}[Enhanced Relative Localization]
\label{lem:pi_to_z_overlap}
Let $\Mcal = (\Xcal, P, \pi)$ be a Markov chain. Given positive parameters $\eta<1$, $b<1$, $\alpha_P < (1 - b)/2$, and integer $\ell > 3/\alpha^2$, suppose that
$(H, g_\eta)$ has the $\alpha_P$-subdepolarizing property. Let $z^{\star}\in\{-1, +1\}^n$ be an optimal assignment, i.e. $H(z^{\star}) = E^{\star}$. %
Then,
$$
    \langle \sqrt{\pi}|(-H_b)^{\ell}|z^*\rangle = \Omega\left( \sqrt{\pi(z^{\star})}e^{\frac{b\eta}{2\alpha_P}}\right).  
$$
In other terms, when Lemma \ref{lem:long_jump_overlap_lower_bound} holds:
\begin{align*}
    \frac{\langle z^{\star} | \psi_b\rangle}{\langle z^{\star}|\sqrt{\pi}\rangle} = \Omega\left(e^{\frac{b\eta}{2\alpha_P}}\right).
\end{align*}
\end{lemma}
\begin{proof}
Define $A, B$, and $f$ by the following equations:
\begin{align*}
    A &= D(P)\\
    B &=  -bg_{\eta}\left(\frac{H}{E^{\star}}\right)=bf\left(\frac{H}{E^{\star}}\right),
\end{align*}
so
\begin{align*}
    \langle y |A|x\rangle = \pi^{1/2}(x)P(x,y)\pi^{-1/2}(y).
\end{align*}

To compute the numerator $(A + B)^{\ell}$, start with
    \begin{align*}
        &\braket{\sqrt{\pi}|A|z^{\star}} = \braket{\sqrt{\pi}|z^{\star}} = \pi(z^{\star})^{1/2}\\
        &\braket{\sqrt{\pi}|B|z^{\star}} = b\pi(z^{\star})^{1/2}. 
    \end{align*}
Next, we compute
\begin{align*}
    \braket{\sqrt{\pi}|BA^k|z^{\star}} &= b\braket{\sqrt{\pi}|BD^k|z} = b\sum_{x_1\dots x_k}\braket{\sqrt{\pi}|B|x_k}\braket{x_k|D|x_{k-1}}\cdots \braket{x_1|D|z} \\
    &= b\sum_{x_1\dots x_k}\braket{\sqrt{\pi}|B|x_k} \pi(x_k)^{-1/2}P(x_{k-1},x_k)\pi(x_{k-1})^{1/2}\cdots \pi(x_1)^{-1/2}P(z^{\star}, x_1)\pi(z^{\star})^{1/2}\\
    &= b\sum_{x_1\dots x_k}f\left(\frac{H(x_k)}{E^{\star}}  \right)P(z^{\star},x_1)\cdots P(x_{k-1}, x_k)\pi(z^{\star})^{1/2}\\
    &= b\pi(z^{\star})^{1/2}\mathbb{E}_{x_1\sim z}\cdots\mathbb{E}_{x_k\sim x_{k-1}}f\left(\frac{H(x_k)}{E^{\star}}  \right)\\
    &\geq b\pi(z^{\star})^{1/2}f\left((1-\alpha)^k\right).
\end{align*}  

In general, we can write any string of $A$'s and $B$'s as
$$
    \dots AB^{c_3}AB^{c_2}AB^{c_1}AB^{c_0},
$$
for $c \in \ell_1(\mathbb{N}^{\infty})$, i.e., finite sequences of natural numbers of unbounded length. Let $\tilde{f}(x) := f(H(x)/E^{\star}).$ Accordingly, we can compute
\begin{align*}
    \braket{\sqrt{\pi}| \dots AB^{c_3}AB^{c_2}AB^{c_1}AB^{c_0}|z^{\star}}&=\sum_{x_1,\dots}\dots \braket{x_4|AB^{c_3}|x_3}\braket{x_3|AB^{c_2}|x_2} \braket{x_2|AB^{c_1}|x_1}\braket{x_1|AB^{c_0}|z^{\star}}\\
    &=\sum_{x_1,\dots}  b^{\sum_j {x_j}}  \dots \braket{x_2|D|x_1}f(H(x_1)/E^{\star})^{c_1}  \braket{x_1|D|z^{\star}}\\
    &=\pi(z^{\star})^{1/2}b^{\sum_{j=0}^{\infty} {c_j}} \sum_{x_1,\dots}\dots P(x_1,x_2)f(H(x_1)/E^{\star})^{c_1}P(x_1,z^{\star})\\
    &=\pi(z^{\star})^{1/2}b^{\sum_{j=0}^{\infty} {c_j}}\mathbb{E}_{z^{\star} \underset{P}{\sim} x_1}[(\tilde{f}(x_1))^{c_1}\mathbb{E}_{x_1 \underset{P}{\sim} x_2}[(\tilde{f}(x_2))^{c_2}\cdots ]]\\
    &\geq \pi(z^{\star})^{1/2}b^{\sum_{j=0}^\infty c_j}  \prod_{j=0}^\infty f((1-\alpha)^j)^{c_j}.
\end{align*}
By assumption, $b^{\sum_{j=0}^\infty c_j}$ is finite.

We only pick up $\pi(z^{\star})^{1/2}$ term instead of $2^{-n/2}$ in front of the product. Therefore Propositions~15~and~16 of \cite{dalzell2022mind} hold as long as $\ell > \frac{3}{\alpha^2}$. As observed in \cite{dalzell2022mind}, the function ``$F(x) = 1-x+x\ln(x)$'' satisfies $\frac{F(1-\eta)}{\eta} \geq \eta/2$. Hence, we obtain the stated result.
\end{proof}

In the above proof, it is the inclusion of trajectories $|z^*\rangle|x_1\rangle |x_2\rangle \dots$ that walk around the landscape of the global min, defined by the filter $g_{\eta}$, that enhance the localization ($f$ does not annihilate all paths leaving $z^*$, unlike in unstructured search). After combining Lemmas \ref{lem:long_jump_overlap_lower_bound} and \ref{lem:pi_to_z_overlap}, the proof of the following result is self-evident. This formalizes the second consequence of the previous lemma and also accounts for all $z^{\star} \in \textup{Ran}~\Pi^{\star}$.

\begin{theorem}[Long Jump Cost for General Markov Chains]
\label{thm:long_jump_cost_bound}
Let $\Mcal = (\Xcal, P, \pi)$ be a reversible, aperiodic Markov chain. Let $\nu$ be a lower bound on either the log-Sobolev or Poincar\'e constant of $P$. Given positive parameters $\eta<1$, $b<1$, $\alpha_P < (1 - b)/2$. Suppose that
$(H, g_\eta)$ has the $\alpha_P$-subdepolarizing property, and that there exists an integer $\ell$ such that
\begin{align*}
&\max\left(\frac{4\ln(V)}{\nu}, \frac{3}{\alpha_P^2}\right) < \ell \leq \frac{\langle \psi_b|\sqrt{\pi}\rangle}{b\sqrt{\mathbb{P}_{\pi}( E \leq (1-\eta)E^*)}}.
\end{align*}
If $|\psi_b\rangle$ is the unique ground state of $H_b$ and the spectral gap of $H_b$ is lower bounded by $\frac{\nu}{2}$, then
\begin{align*}
\lVert \Pi^{\star}|\psi_b\rangle\rVert_{2} =  \Omega\left(\sqrt{\pi(E^{\star})}e^{\frac{b\eta}{2\alpha}}\right).
\end{align*}
\end{theorem}

It may not be immediately obvious that the conditions on $\ell$ are not contradicting. Here, we give intuition for why this is not the case for typical applications of the algorithm. Note for an efficient algorithm, using Theorem \ref{thrm:MTG_RT}, at the very least we will need $\omega^{-1} = \mathcal{O}(\poly(n))$. The overlap satisfies $|\langle \psi_b|\sqrt{\pi}\rangle| \geq \frac{\omega}{2}$ by Lemma \ref{lem:short-path-overlap}. For hard problems, 
$\pi(E^{\star})$ will be exponentially small in $n$ and $\gamma = \Theta(1)$.  Note $\alpha_{P}^2$ is significantly larger than $\pi(E^{\star})$ for a hard problem, for showing super-Grover runtime we will  want $\alpha_{P} = \Theta(\frac{1}{\ln(1/\pi(E^\star))})$ anyways. Thus, since $b \leq \frac{2}{3}\gamma \omega \ln(1/\pi(E^{\star}))$, if  the spectral density condition holds, then  $\frac{1}{\gamma\ln(1/\pi(E^{\star}))\sqrt{\mathbb{P}_{\pi}( E \leq (1-\eta)E^*)}}$ will be exponentially larger than the lower bound on $\ell$.  So it is fine to assume we have the conditions on $\ell$ stated in Theorem \ref{thm:long_jump_cost_bound} in settings where the algorithm can successfully be applied (validating Remark \ref{rem:conditions}).

As mentioned in Corollary \ref{cor:super_grover_run}, at a constant $b$ the existence of an asymptotic speedup over Grover is determined solely by $\frac{\Delta_{P}}{\lvert E^*\rvert}$. Clearly, we must have that $\frac{\Delta_{P}}{\lvert E^*\rvert} =\Omega\left(\frac{1}{\ln(1/\pi(E^*))}\right)$, at least for a problem with at least exponential runtime. If this does not hold, then the runtime goes to zero asymptotically, an absurdity. However, it has not been \emph{directly} shown that  the derived runtime  cannot lead to this contradiction. To put one's mind at ease, we present the following result.

\begin{lemma}
Let $\Mcal = (\Xcal, P, \pi)$ be a Markov chain with log-Sobolev constant lower bounded by $\omega$. Suppose $b^*$ and $\gamma$ in Theorem \ref{thrm:MTG_RT} are $\Theta(1)$, and $\lVert H \rVert_{2} = \lvert E^*\rvert$. If 
\begin{align}
\expP[H(y)] \leq  H(x)+\tilde{\Delta}_{P},\quad \forall x \in \mathcal{X},
\end{align}
then $\frac{\tilde{\Delta}_P}{\lvert E^*\rvert} = \Omega\left(\frac{1}{\ln(1/\pi(E^*))}\right)$.
\end{lemma}
\begin{proof}
Let $n$ be a real parameter that parameterizes the space of feasible states $\mathcal{S}(n)$ with size $|\mathcal{S}(n)| = S(n)$. Let $M(n)$ denote the mixing time of a Markov chain $P$ with transition density $\pi$ on $\mathcal{S}(n)$ such that for any $t \ge M(n)$, $\TV(P^t \delta_x,  \pi) \le \frac{1}{100}$ for any $x \in \mathcal{S}(n)$.

Assume that $\tilde{\Delta}_P$ is a global upper bound, such that for all $x \in \mathcal{S}(n)$, $\mathbb{E}_{y \sim_P x}[H(y)] \le H(x) + \tilde{\Delta}_P$. It is easy to observe that for any random variable $X$ taking values in $\mathcal{S}(n)$, it holds that $\mathbb{E}_{Y \sim_P X}[H(Y)] \le \mathbb{E}_X [H(X)] + \tilde{\Delta}_P$. Now let $x_{\ast}$ be a global minimum of $H$ (with corresponding energy $E^\ast$) and consider taking $T = \lceil M(n) \rceil$ steps of $P$ starting from $x^\ast$, with the random state after $t \in [1,T]$ steps being denoted $X_t$. By induction, it is easy to see that $E_{X_T}[H(X_T)] \le E^\ast + \tilde{\Delta}_PT$. From the definition of mixing time however, it follows that $\mathbb{E}_{X_T}[H(X_T)] \ge \mathbb{E}_{\pi}[H(x)] - \frac{\lvert E^{\star}\rvert}{100}$. Denoting $\mathbb{E}_{\pi}[H(x)]$ by $\bar{E}$ it follows that $\tilde{\Delta}_P \ge \frac{\bar{E} - E^\ast}{T} - \frac{\lvert E^*\rvert}{100T} = \Omega\left(\frac{\lvert E^*\rvert}{M(n)}\right)$. So $\frac{\tilde{\Delta}_P}{\lvert E^{\star}\rvert} =  \Omega\left(\frac{1}{M(n)}\right)$.

If $b^*$ and $\gamma$ are constant, then the log-Sobolev constant must satisfy $\omega = \Omega\left(\frac{1}{\ln(1/\pi(E^*)}\right)$. It follows from standard results on Markov chains that $
\omega \leq \frac{1}{M(n)}$.
\end{proof}

Note that in \cite{dalzell2022mind}, the authors state an additional technical condition that is of course easy to satisfy, which is $\mathbb{E}_{\pi}[H(x)] = 0$.  The main reason for setting $\mathbb{E}_{\pi}[H(x)] = 0$ is to ensure that $E^{\star} < 0$ and for the ease of proving the tail bounds. In our setting, for arbitrary $\pi$, the expectation may need to be estimated if used as a shift. However, the shift is not necessary to run the algorithm if the $E^{\star} < 0 $ condition is already satisfied. The mean just appears as component of the runtime. Also, Corollaries \ref{cor:tail-herbst} and 
\ref{cor:tail-poinc} for the tail bounds do not assume this shift.

Lastly, we obtain the following self-evident corollary of Theorem \ref{thm:long_jump_cost_bound} when applied to Equation \eqref{eqn: runtime}
\begin{corollary}
\label{cor:total_runtime_bound}
Let $\Mcal = (\Xcal, P, \pi)$ be a reversible, aperiodic Markov chain. Let $\nu$ be a lower bound on either the log-Sobolev or Poincar\'e constant of $P$. Given positive parameters $\eta<1$, $b<1$, $\alpha < (1 - b)/2$. Suppose that
$(H, g_\eta)$ has the $\alpha$-subdepolarizing property, and that there exists an integer $\ell$ such that
\begin{align*}
&\max\left(\frac{4\ln(V)}{\nu}, \frac{3}{\alpha^2}\right) < \ell \leq \frac{\langle \psi_b|\sqrt{\pi}\rangle}{b\sqrt{\mathbb{P}_{\pi}( E \leq (1-\eta)E^*)}}.
\end{align*}
If $|\psi_b\rangle$ is the unique ground state of $H_b$ and the spectral gap of $H_b$ is lower bounded by $\frac{\nu}{2}$, then the runtime of the short-path algorithm is bounded by
\begin{align*}
\mathcal{O}\left(\textup{poly}(\log V)\cdot \nu^{-1}\left(\lvert \langle \psi_b | \sqrt{\pi}\rangle\rvert^{-1}+\pi(E^*)^{-\frac{1}{2}}e^{-\frac{b\eta}{2\alpha_{P}}}\right)\right).
\end{align*}
\end{corollary}
Hence if the short-jump can be performed in polynomial time, i.e. \begin{align*}
    [\min(\textup{Gap}(-D(P)), \textup{Gap}(H_b))]^{-1}\lvert \langle \psi_b | \sqrt{\pi}\rangle\rvert^{-1} = \mathcal{O}\left(\text{poly}(n)\right),
\end{align*}
then the runtime of the overall generalized short path algorithm is $\mathcal{O}^{\star}\left([\pi(E^*)^{-1}]^{\left(\frac{1}{2}-\frac{b\eta}{2\alpha_{P} \ln(1/\pi(E^*))}\right)}\right)$.

\section{Applications of Generalized Short-Path Framework}
\label{sec:applications-complete}

\subsection{Optimization with Fixed Hamming Weight: Transposition Walk}
\label{sec:fixed-hamming-weight}

The $k$-particle Bernoulli-Laplace (BL) diffusion or Transposition Walk on $n$ sites is a random walk on the space of Hamming weight $k$ bitstrings. For our case we will work with $\pm 1$ strings or ``spin configurations'' $x$ and define the Hamming weight $\lvert x\rvert$ as the number of $+1$'s. Formally, $\mathcal{X} = \{ x \in \{-1, 1\}^n : \lvert x\rvert = k\}$. A single step of BL consists of choosing, uniformly at random, a transposition that swaps some $x_j = 1$ with another $x_i = -1$.

There is a very natural quantum Hamiltonian on $n$-qubits that encodes the discriminant of the transposition walk:
$$
    D(P) = P =  \frac{1}{k(n-k)}\sum_{i<j
} \frac{X_{i} X_{j} + Y_{i} Y_{j}}{2},
$$
where $X_j, Y_j$ denote the Pauli operators applied to qubit $j$. This is commonly called the complete-graph $XY$ mixer. Note that there is equality between $D(P)$ and $P$ because the walk is symmetric. The ground state of $-D(P)$ is the uniform superposition over Hamming weight $k$ computational basis states, and thus encodes the stationary distribution of the transposition walk. Note that $|\sqrt{\pi}\rangle$ is just the Hamming weight $k$ Dicke state, which can be prepared efficiently \cite{B_rtschi_2019}.

We have the following log-Sobolev inequality for the transposition walk. %
\begin{theorem}[{\cite[Theorem 5]{lee1998logarithmic, salez2020sharp}}, Discrete-time]
    \label{thm:log-sobolev-swap-standard}
    Let $P$ be the transition matrix for $k$-particle Bernoulli-Laplace diffusion on $n$ sites and $\pi$ the stationary distribution. It then holds for any real valued function $\psi$ that %
    $$
       \mathcal{D}(\psi, \psi) \ge  \frac{n}{k(n-k)\tau_0\log\left(\frac{n}{\min(k, n-k)}\right)}\mathrm{Ent}(\psi^2),
    $$
    where $\tau_0$ is a universal constant.
\end{theorem}

This leads to the following bound on $b^*$ using the formula in Theorem \ref{thrm:MTG_RT}.
\begin{lemma}
\label{lem:b_bound_transpo} For all $k$, we have the following bound on $b$ for the transposition mixer:
\begin{align}
\label{eqn:b_transposition}
     b^*  =  \frac{2C\gamma}{3},
\end{align}
for some constant $C$.
\end{lemma}
\begin{proof}
We have
\begin{align}
\label{eqn:log_sob_transpo}
\mathcal{D}(\psi, \psi) \geq \frac{n}{k(n-k)\tau_0\log\left(\frac{n}{\min(k, n-k)}\right)}\mathrm{Ent}(\psi^2).
\end{align}
Thus 
$$
    b < \frac{2n\log_2\binom{n}{k}\gamma}{3k(n-k)\tau_0\log_2\left(\frac{n}{\min(k, n-k)}\right)}.
$$
Using that for $k = o(n)$, $\log_2\binom{n}{k} = \Theta(k \log_2(n/k))$, and $\frac{n}{\min(k, n-k)} = \Theta(\log_2(n/k))$, we get
\begin{align}
     b  <  \frac{C2\gamma}{3},
\end{align}
for some constant $C$ to be determined, so $b$ is constant for $k = o(n)$.

For $k = \Theta(n)$, we have that $\log_2 \binom{n}{k} = \Theta(n)$, $\frac{n}{\min(k, n-k)} = \Theta(1)$, so $b^*$ is also constant.
\end{proof}

This leads to the following simple result that follows from applying Theorem \ref{thrm:MTG_RT}.
\begin{theorem}
\label{thm:bl_runtime_thm}
Let $\Mcal = (\Xcal, P, \pi)$ be the $k$-particle transposition walk on $n$ sites. Let $H: \pmone^n \rightarrow \R$ be a diagonal Hamiltonian with ground state energy $E^{\star}$. If  $\lVert H \rVert_{P} = \Ocal(1)$, and 
\begin{align*}
    \lvert E^{\star}\rvert = \Theta\left(\ln\binom{n}{k}\right),
\end{align*}
 then there exists a short path algorithm with runtime
$$
\Ostar\left(\binom{n}{k}^{\left(\frac{1}{2}-\frac{\eta(1-\eta)\lvert E^{\star}\rvert b}{2\ln\binom{n}{k}\Delta_P}\right)}\right).
$$
\end{theorem}
\begin{proof}
For all $k$ we have $\omega = \Omega\left(\frac{1}{k\ln(n/k)}\right)$, and $k(\ln(n/k))\ln(1/\pi(E^*)) \asymp (\ln\binom{n}{k})^2$. Thus if $\lVert H \rVert_{P} = \Ocal(1)$ and  $\lvert E^{\star}\rvert = \Theta\left(\ln\binom{n}{k}\right)$, then $\gamma$ is constant so then $b^*$ is. We also have that $\frac{\lvert E^*\rvert}{\Delta_P} = \Ocal\left(\ln\binom{n}{k}\right)$, leading to the runtime presented.
\end{proof}

We apply the above result to a Hamming-weight constrained version of MaxCut over Erd\H{o}s-R\'enyi graphs, which we call \ref{e:MaxCut}. One well-known special case is Hamming weight $\frac{n}{2}$ called \ref{e:MaxBisection}. 

\subsubsection{Hamming weight Constrained MaxCut}

Consider a graph $G(\Ncal, \Ecal)$ with vertex set $\Ncal := [n]$ and edge set $\Ecal$. We assume $G$ is drawn from the Erd\H{o}s-R\'enyi ensemble $\mathcal{G}\left(n, \frac{p}{n-1}\right)$ for a constant $p$, i.e., each edge $e_{ij} \in \Ecal$ for $(i,j) \in [n] \times [n]$ is created with probability $\frac{p}{n-1}$ such that $G$ has an average degree of $p$. We are interested in solving the  \textit{Maximum b\\isection} problem:
\begin{equation}\label{e:MaxBisection} \tag{MaxBisection}
\mathcal{C}^{*}_{\frac{n}{2}} := \min_{ x \in \{-1, 1\}^{n}}\left\{ -\frac{1}{2}\sum_{ i < j} e_{ij} (1-x_ix_j) : \lvert x\rvert = \frac{n}{2} \right\},
\end{equation}
where $e_{ij}$ is a $\frac{p}{n-1}$ Bernoulli indicating whether the $(i,j)$ edge is present. 

For generality, we strive to present the results for an arbitrary Hamming weight constraint of size $k$ and specify $k=\frac{n}{2}$ where necessary. We call the case where $k$ can be arbitrary the \textit{MaxCut Hamming} problem:
\begin{equation}\label{e:MaxCut} \tag{MaxCut-Hamming}
\mathcal{C}^{*}_{k} := \min_{ x \in \{-1, 1\}^{n}} \left\{ -\frac{1}{2}\sum_{i < j} e_{ij} (1-x_ix_j) : \lvert x \rvert = k \right\}.
\end{equation}

The following result shows that $\lVert H \rVert_{P} = \Ocal(1)$. 
\begin{lemma}
\label{lem:maxcut-ham_pseudo-lip-bound}
For the \ref{e:MaxCut} Hamiltonian $H$, the pseudo Lipschitz constant $\lVert H\rVert_{P}$ under the transposiiton walk is $\mathcal{O}(1)$ with high probability over the graph.
\end{lemma}
The proof of the above lemma is deffered to the appendix. Next we show the existance of a tail bound for \ref{e:MaxBisection}.
\begin{lemma}
Let $\Xcal = \{ x \in \{-1,1\}^n : | x | = \frac{n}{2}\}$ and $\Mcal = (\Xcal, P, \pi)$ be a reversible Markov chain. For \eqref{e:MaxBisection} on a graph $G \sim \mathcal{G}\left(n, \frac{p}{n-1}\right)$ and $D(P)$ being the transposition mixer, we have that
$$
\frac{4((1-\eta)((1-\eta) +\frac{p}{2})^2}{\tau_0} \lesssim \gamma.
$$
\end{lemma}
\begin{proof}
Recall the expression for $\gamma$ in terms of the Herbst argument provided in Theorem~\ref{Herbst}:
$$
\gamma = \frac{\omega((1-\eta)E^{\star} - \mathbb{E}_{\pi}[H])^2}{\lVert \psi \rVert_{P}\ln(1/\pi(E^{\star}))}.
$$
From Lemma \ref{lem:maxcut-ham_pseudo-lip-bound} we have that $\lVert H \rVert_{P} = \Ocal(1)$. Applying Equation \eqref{eqn:log_sob_transpo} for $k=\Theta(n)$, the log-Sobolev constant $\omega$ satisfies:
\begin{align*}
\omega \geq  \frac{n}{k(n-k)\tau_0\log\left(\frac{n}{\min(k, n-k)}\right)} = \frac{4}{\tau_0 n}.
\end{align*}
From Lemma \ref{lem:loose_maxcut_ham_bound} and Lemma \ref{lem:maxcut_mean_ener_bound} we have
\begin{align*}
((1-\eta)E^{\star}-\mathbb{E}_{\pi}(H_c))^2 \asymp\left[(1-\eta)n +\frac{np}{2}\right]^2.
\end{align*}
The definition of $\pi$ combined with the asymptotics of the binomial coefficient for $k=\Theta(n)$ gives
\begin{align*}
\ln(1/\pi(E^{\star})) \asymp n.
\end{align*}
Putting everything together:
\begin{align*}
\frac{4((1-\eta)((1-\eta) +\frac{p}{2})^2}{\tau_0} &\lesssim \gamma.
\end{align*}
\end{proof}

Lastly, $\lvert E^{\star}\rvert = \Theta(n)$ with high probability from Lemma \ref{lem:loose_maxcut_ham_bound}. Since this is $\Theta\left(\ln\binom{n}{k}\right)$ for $k =\Theta(n)$ 
all of the conditions of Theorem \ref{thm:bl_runtime_thm} are met.

Unfortunately, the current analysis is insufficient to show this for $k=o(n)$. For example, using Lemma \ref{lem:constr_max_delta} we can take an upper bound of 
\begin{align}
    \Delta_{P} = \frac{\mathcal{C}^{*}_{k}(n-2)}{k(n-k)},
\end{align}
which gives that $\frac{\lvert E^*\rvert}{\Delta_P} = \Ocal(k)$. However, $\ln\binom{n}{k} = \Theta(k\ln(n/k))$ for $k = o(n)$. Assuming a tail bound, this leads to a speedup that is falling with $n$. Specifically, the speedup is falling with $\frac{1}{\ln(n)}$, which is reminiscent of the running time achieved by \cite{hastings2018shortPath}. We summarize the two cases in the following theorem.
\begin{theorem}
\label{thm:maxcut_runtime}
Let $\Mcal = (\Xcal, P, \pi)$ be the $k$-particle transposition walk on $n$ sites. Let $H: \pmone^n \rightarrow \R$ be a diagonal Hamiltonian encoding the \ref{e:MaxCut} cost function for a graph $G \sim \mathcal{G}\left(n, \frac{p}{n-1}\right)$. Then either of the following runtimes hold depending on $k \leq n$. 
\begin{itemize}
    \item If $k=\Theta(n)$, then there exists a short path algorithm with runtime
$$
\Ostar\left(\binom{n}{k}^{\frac{1}{2}-c}\right),
$$
for some constant $c$, and
    \item if $k = o(n)$, then there exists a short path algorithm with runtime
$$
\Ostar\left(\binom{n}{k}^{\frac{1}{2}-\frac{c}{\ln(n)}}\right),
$$
for some  constant $c$.
\end{itemize}
\end{theorem}

\subsection{Glauber Dynamics}
\label{sec:glauber_dynamics}

\textit{Glauber dynamics} is a Markov chain algorithm designed to sample from the Gibbs distribution of a system, particularly in spin models like the Ising or hardcore model \cite{glauber1963time}. The Gibbs measure $\pi_{\beta}$ for a system with configuration space $\mathcal{X}$ and Hamiltonian $H$ is defined as 
\begin{equation*}
\label{eq:Gibbs}
   \pi_{\beta}(x)=  \frac{\exp(-\beta H(x))}{Z(\beta)}
\end{equation*}
where $x \in \mathcal{X}$ is a configuration, $\beta = \frac{1}{T}$ is the inverse temperature, and $Z(\beta)$ is the partition function
\begin{equation*}
    Z(\beta) = \sum_{x \in \mathcal{X}} \exp(-\beta H(x)).
\end{equation*}
Glauber dynamics generates a Markov chain with this measure as its stationary distribution by sequentially updating a single site (spin or vertex) according to marginal distribution. Specifically, the update proceeds as follows: (i) a site (vertex) $v$ is chosen uniformly at random; (ii) the states of all other sites $u\neq v$ remain unchanged; (iii) the new state of $v$ is sampled from the marginal distribution of $v$ conditioned on its neighbors. 

The efficiency of Gibbs sampling in this case is related to how fast the Glauber dynamics mixes to its stationary distribution. In fact, approximate Gibbs sampling is tightly connected to partition function estimation in terms of computational complexity \cite{vstefankovivc2009adaptive}, both of which are closely related to statistical phase transitions. For bounded-degree graphs, these transitions have been shown correspond to the \textit{uniqueness/non-uniqueness} threshold on an infinite $d$-regular tree which captures whether the root of the tree is affected by the leaves. Hence, this critical point is also called the \emph{tree-uniqueness threshold}. In the uniqueness regime, correlations decay rapidly, allowing efficient approximation of the partition function. However, beyond the non-uniqueness threshold, long range correlations emerge and no polynomial-time algorithm can approximate the partition function. However, the tree threshold only appears to indicate the hardness transition for worst-case instances. The phase transition point for random graphs appears to be beyond the tree-uniqueness threshold.

In particular, for the hardcore model with fugacity parameter $\lambda = e^{\beta}$, \cite{10.1145/1132516.1132538} presented a fully polynomial time approximation scheme ($\mathsf{FPTAS}$) for computing the partition function for graphs with maximum degree $d$ when $\lambda\leq (1-\xi)\lambda_c$ where $\lambda_c = \frac{(d-1)^{(d-1)}}{(d-2)^d}$ is the corresponding uniqueness threshold. Additionally, Glauber dynamics is known to mix in polynomial-time for any constant $\xi$ and $d$. On the other hand, \cite{sly2010computationaltransitionuniquenessthreshold} proved that there is no fully polynomial randomized approximation scheme ($\mathsf{FPRAS}$) to approximate the partition function for bounded-degree graphs when $\lambda>\lambda_c$ unless $\mathsf{NP}=\mathsf{RP}$ confirming the main conjecture of \cite{mossel2007hardnesssamplingindependentsets}.

For $\lambda$ at criticality, $\lambda = \lambda_c$, \cite{chen2025rapidmixinguniquenessthreshold} showed that Glauber dynamics mixes in polynomial time and hence partition function approximation is also polynomial time. However, approaching criticality $\lambda = (1-\xi)\lambda_c$, the modified log-Sobolev and standard log-Sobolev constants for Glauber dynamics appear to fall as $\mathcal{O}(\exp(-1/\xi))$ \cite{chen2023optimal}. Still, the Poincar\'e constant (hence the spectral gap) remains $\mathcal{O}(1/\text{poly}(n))$ at criticality.

 For random regular graphs, the situation is quite different. \cite{chen2025rapidmixingrandomregular} showed that Glauber dynamics mixes in polynomial time for $\lambda > \lambda_{c}$, going all the way up to $\lambda = \mathcal{O}(1/\sqrt{d})$. For the hardcore model, this was done by providing a Poincar\'e inequality rather than log-Sobolev inequality.

There is a similar situation for the antiferromagnetic Ising model over bounded-degree graphs. The critical inverse-temperature $\beta_c$, corresponding to the uniqueness/non-uniqueness phase transition on the infinite $d$-regular tree, is defined by $(d-1)\tanh (\beta_c) = 1$ \footnote{Note in some literature the criticality condition is expressed in terms of $e^{-2\beta_C}$ where the condition becomes $e^{-2\beta_C} \geq  \frac{d-2}{d}$.} \cite{chen2025rapidmixinguniquenessthreshold}. \cite{galanis2016} show that unless $\mathsf{NP}=\mathsf{RP}$  there is no $\mathsf{FPRAS}$ for approximating the partition function of the antiferromagnetic Ising model when $\beta > \beta_c$. Glauber dynamics mixes in polynomial-time for $(d-1)\tanh(\beta_c) \leq 1 - \delta$ \cite{chen2023optimal} for constant $\delta \in (0, 1)$. More recently, \cite{chen2025rapidmixinguniquenessthreshold} showed that Glauber dynamics also mixes efficiently for $(d-1)\tanh(\beta_c) \leq 1 + \frac{\alpha}{\sqrt{n}}$ for constant $\alpha > 0$. Hence for the antiferromagnetic Ising model, we can mix slightly beyond the threshold  with an edge that vanishes as $n \rightarrow \infty$. 

For the antiferromagnetic Ising model over \emph{random $d$-regular graphs} \cite{chen2025rapidmixingrandomregular}, Glauber dynamics mixes in total variation distance for \emph{constant $d$} when $(d-1)\tanh(\beta) \leq \frac{d-1}{8\sqrt{d-1}-1} \footnote{This is converted from the condition $e^{-2\beta} \geq 1- \frac{1}{4\sqrt{d-1}}$ in \cite[Corollary 5.5]{chen2025rapidmixingrandomregular}.}$ with high probability. This is close to a computational threshold, as it is known that with high probability no (stable \cite{alaoui2024samplingsherringtonkirkpatrickgibbsmeasure}) classical algorithm can efficiently mix in \emph{Wasserstein-2 distance} beyond $\beta_{\textup{DC}}$, where $(d-1)\tanh(\beta_{\textup{DC}}) = \frac{d-1}{\sqrt{d}}$ \cite{huang2024hardnesssamplingantiferromagneticising}, when $d$ is sufficiently large. This is due to a phenomenon known as \emph{disorder chaos} \cite{chatterjee2009disorderchaosmultiplevalleys}.

There is a well-known correspondence between the free-energy densities of the antiferromagnetic Ising model over random $d$-regular graphs, as $d \rightarrow \infty$, and the SK model \cite{dembo2017extremal}. Specifically, the SK model also suffers from disorder chaos and no stable classical algorithm can efficiently mix in Wasserstein-2  when the inverse-temperature exceeds $\beta > 1$ \cite[Theorem 2.5]{alaoui2024samplingsherringtonkirkpatrickgibbsmeasure}, with high probability. One can show that a properly normalized $\beta_{\textup{DC}}$ for the Ising case matches the SK point of disorder chaos, i.e.  $\lim_{d \rightarrow \infty}\sqrt{d} \cdot \beta_{\textup{DC}} = 1$. Classical stable algorithms have been shown to sample in Wasserstein-$2$ from the SK model for $\beta < 1$ \cite{alaoui2024samplingsherringtonkirkpatrickgibbsmeasure, celentano2022sudakovferniquepostampnewproof}. Our quantum algorithm provides a speedup for the regime of $\beta < \frac{1}{4}$ \cite{chen2022localizationschemesframeworkproving}, where efficient sampling in TV can be shown.

For convenience, we prove the following lemma that is useful when establishing $\|H\|_P = \Ocal(1)$.

\begin{lemma}
\label{lem:coordinate-Lipshitz}
Let $P$ be a Glauber dynamics chain on $\mathcal{X}$ and let $H: \Xcal \rightarrow \R$ be a diagonal Hamiltonian with ground state energy 
$$ E^{\star} \coloneqq \min_{x \in \Xcal} H(x).$$ 
If the following holds for all $x, x'$ such that $P(x, x')>0$: 
\begin{align}
   \lvert H(x') - H(x)\rvert \leq q,    
\end{align}
then $\|H\|_P \leq q^2$.
\end{lemma}
\begin{proof}
The proof simply follows from the definition of $\|H\|_P$ as
    \begin{align*}
        \|H\|_P &= \max_x \sum_{x'}P(x, x')(H(x')-H(x))^2
        \leq \max_{x, x'}  \lvert H(x') - H(x)\rvert^2
        \leq q^2.
    \end{align*}
\end{proof}

\subsubsection{Bounding the Mixing time of Glauber Chains}

\label{subsec:glauber_mixing_tech}
In this section we review some of the known techniques for bounding the log-Sobolev and Poincar\'e constants of Glauber chains. These will be used to derive all of the results in the following subsections.

Consider the following definitions regarding an arbitrary measure $\pi$.
\begin{definition}[Approximate Tensorization of Entropy]
 We say that a distribution $\pi$ satisfies approximate tensorization of entropy with constant $C_1 > 0$ if for all $\pi$-measurable functions $f : \Xcal \rightarrow \mathbb{R}$:
 \begin{align*}
\Ent_{\pi}[f]&\leq C_1\sum_{v\in \Xcal} \Ent_{\pi_v}[f].
 \end{align*}
\end{definition}

\begin{definition}[$u$-Marginally Bounded Measure (Definition 1.9~\cite{chen2023optimal})]
We say that the stationary distribution $\pi$ on $\Xcal \subset \{-1, 1\}^{n}$ is $u$-marginally bounded: if  for every $\Lambda \subset \Xcal$ and partial assigment $\tau : \Lambda \rightarrow \{-1, 1\}$ that are not $\pi$-measure zero, it holds that for every $x \in \Xcal \setminus \Lambda$:
\begin{align*}
    \pi(\sigma_{v} = i | \sigma_{\Lambda} = \tau) \geq u,
\end{align*}
whenever the above conditional measure is nonzero.
\end{definition}

The above is used to provide alternative sufficient conditions to Theorem \ref{thm:modified-LSI}  from \cite{chen2023optimal} under which a modified log-Sobolev inequality can be promoted to a log-Sobolev inequality. We record this as the following fact.  
\begin{fact}
\label{sobolev-marginally-bounded}
Let $\Xcal$ be a set of size $n$ and µ be a distribution over $[q]^\Xcal$. If $\pi$ satisfies
the approximate tensorization of entropy with constant $C_1$ and $\pi$ is $u$-marginally bounded, then the Glauber dynamics
for $\pi$ satisfies the standard log-Sobolev inequality with constant $\omega =\frac{1-2u}{\log(1/u-1)}\frac{1}{C_1 n }$ when $u<\frac{1}{2}$, or $\omega = \frac{1}{2C_1 n}$ when $u=\frac{1}{2}$. 

\begin{proof}
Fix a configuration $x$ and consider a Markov chain $P_v$ that updates vertex $v$ according to marginal probability distribution $\pi_v =\pi(v|x_{\Xcal- \{v\}})$. Then, LS constant of this Markov chain $\rho_v$ is lower bounded by $\rho_v\geq \frac{1-2\pi_v^{\star}}{\log(1/\pi_v^{\star}-1)}$ when $\pi_v^{\star}<\frac{1}{2}$ or $\rho_v=\frac{1}{2}$ when $\pi_v^{\star}=\frac{1}{2}$ due to \cite[Theorem A.1]{10.1214/aoap/1034968224}. By the definition of the log-Sobolev constant, we have
\begin{align*}
\rho_v \Ent_{\pi_v}[f] &\leq \mathcal{E}_{P_v}(\sqrt{f}, \sqrt{f}) = \Var_{\pi_v}[\sqrt{f}].
\end{align*}
The transition matrix of Glauber dynamics can be written as $\frac{1}{n}\sum_{v\in \Xcal} \pi(v|\cdot)$. Therefore, the Dirichlet form for Glauber dynamics is
\begin{align*}
    \mathcal{D}(\sqrt{f}, \sqrt{f}) = \frac{1}{n} \sum_{v\in \Xcal} \Var_{\pi_v}[\sqrt{f}]. 
\end{align*}
Using the tensorization of entropy,
\begin{align*}
    \Ent_{\pi}[f]&\leq C_1\sum_{v\in \Xcal} \Ent_{\pi_v}[f] \leq C_1 \sum_{v\in \Xcal}\frac{1}{\rho_v} \Var_{\pi_v}[\sqrt{f}] \leq  C_1 n \max_{v \in \Xcal} \left({\frac{1}{\rho_v}}\right) \mathcal{D}(\sqrt{f}, \sqrt{f}).
\end{align*}
Combining the marginal boundedness property of $\pi$ (i.e., $\pi_v \geq u $ for all $v\in \Xcal$), with the monotonicity of the function $\frac{1-2y}{\log(1/y-1)}$ for $y\in[0, 1/2]$, it follows that $\omega\geq  \frac{1-2u}{\log(1/u-1)}\frac{1}{C_1 n}$ when $u<\frac{1}{2}$, and $\omega = \frac{1}{2 C_1 n} $ when $u=\frac{1}{2}$. %

\end{proof}
\end{fact}

The approximate tensorization of entropy constant can be bounded if $\pi$ is marginally bounded and spectrally independent.
\begin{definition}[$\rho$-Spectral Independence]
    Given $\Lambda \subset V$ and partial assignment $\tau : \Lambda \rightarrow \{1, -1\}$ that is not $\pi$-measure zero. For every $(u, i), (v, j) \in V \times V$, $u \neq v$, if $\pi(\sigma_w = i | \sigma_{\Lambda} = \tau) > 0, w \in \{u, v\}$, then we define the corresponding element of the influence matrix to be 
    \begin{align*}
        \Psi_{\pi}^{\tau}((u,i), (v,j)) =\pi(\sigma_v = j | \sigma_v = i, \sigma_{\Lambda} = \tau) - \pi(\sigma_v = j | \sigma_{\Lambda} = \tau),
    \end{align*}
    and zero otherwise. We say that $\pi$ is $\eta$-spectrally independent if  for every $\Lambda, \tau$, the largest eigenvalue of $\Psi_{\pi}^{\tau}$ is bounded by $\rho$.
\end{definition}

\begin{theorem}[Entropy Factorization, Theorem 2.9 in \cite{chen2023optimal}]
\label{thm:entropy_fact}
    Let $d \geq 3$ be an integer and $b, \rho > 0$ be reals. Suppose that
$G = (\Ncal, \Ecal)$ is an n-vertex graph of maximum degree at most $d$ and $\mu$ is a totally connected Gibbs distribution of some spin system on $G$. If $\mu$ is both $u$-marginally
bounded and $\rho$-spectrally independent and $n\geq \frac{24 d }{u^2}(\frac{4\rho}{u^2}+1)$, then $\mu$ satisfies the approximate tensorization of entropy with constant
$$
    C_1 = \frac{18 \log(1/u)}{u^4}\left(\frac{24 d}{u^2} \right)^{\frac{4\rho}{u^2}+1}.
$$
\end{theorem}

 Analogously to entropy, we have a similar notion of approximation tensorization of variance.

\begin{definition}[Approximate Tensorization of Variance]
 We say that a distribution $\pi$ satisfies approximate tensorization of variance with constant $C_1 > 0$ if for all $\pi$-measurable functions $f : \Xcal \rightarrow \mathbb{R}$:
 \begin{align*}
\Var_{\pi}[f]&\leq C_2\sum_{v\in \Xcal} \Var_{\pi_v}[f].
 \end{align*}
\end{definition}

Approximate tensorization of variance leads to a lower bound on the spectral gap of Glauber dynamics.
\begin{fact}[Lemma 2.1 (1) from \cite {chen2025rapidmixinguniquenessthreshold}]If $\pi$ satisfies $C_2$ approximate tensorization of variance, then the spectral gap of Glauber dynamics is lower bounded by $\frac{1}{C_2n}$.
\end{fact}

\subsubsection{Hardcore Model -- Maximum Independent Set Problem}
\label{sec:MIS-glauber}

Given a graph $ G=G(\Ncal, \Ecal)$, an \emph{independent set} is a subset of vertices in $\Ncal$ such that no two vertices are connected by an edge in $\Ecal$. A \emph{maximal independent set} (MIS) is the largest independent set of $G$. Namely, we solve the optimization problem
\begin{equation}\label{e:MIS} \tag{Maximum Independent Set}
\mathcal{C}^{*}_{G=(\Ncal,\Ecal)} := \min_{ x \in \{0, 1\}^{|\Ncal|}}\left\{ - \sum_{i \in \Ncal} x_i  : x_i + x_j \le 1~\forall (i,j) \in \Ecal  \right\}
\end{equation}
Equivalently, for $|\Ncal| = n$, we can denote an independent set by a configuration $x \in \mathcal{X}\subseteq \{0, 1\}^n$ such that if an edge $e =(i,j)\in \Ecal$ then, $x_i x_j =0$. We denote the size of the maximum independent set by $|x^{\star}|$. Then, finding the maximum independent set is equivalent to finding a ground state of Hamiltonian $H: \mathcal{X}\to [1, n]$,
\begin{equation*}
    H(x) = -\sum_{i=1}^n x_i.
\end{equation*}
The Hamiltonian $H$ is defined on the constrained space $\mathcal{X}$ and therefore we need a constrained walk to explore the state space. For this purpose, we use the Glauber dynamics defined as
\begin{equation}
    P_{\lambda}(x,x') = \begin{cases}
        0 &  \text{if }|x - x'|>1  \\
        \frac{1}{n}\frac{\lambda}{\lambda+1} & \text{if } |x - x'|=1 \text{ and } x\subseteq x'\\
         \frac{1}{n}\frac{1}{\lambda+1} & \text{if } |x- x'|=1 \text{ and } x' \subseteq x\\
          1- \sum_{x''\neq x }P_{\lambda}(x,x'') & \text{if } x = x'.
    \end{cases}
\end{equation}
This model is also referred as \textit{hardcore model} in statistical physics and we'll use the same terminology. Note that Glauber dynamics initialized at an independent set can only move between independent sets. Furthermore, it converges to its stationary distribution 
\[
\pi_{\lambda}(x) = \frac{\lambda^{|x|}}{Z}.
\]
The parameter $\lambda$ is also called \textit{fugacity} and one can recover the Gibbs form in \ref{eq:Gibbs} by setting $e^{\beta} = \lambda$. In principle, one can find the maximum independent set by setting $\lambda$ sufficiently high so that distribution $\pi_{\lambda}$ concentrates around the global minimum of $H$. Unfortunately, this approach results in exponential mixing time in the worst case due to uniqueness/non-uniqueness phase transitions. Alternatively, one can draw exponentially many samples from $\pi_{\lambda}$ at $\lambda \leq \lambda_c$ as Glauber dynamics mixes in polynomial time up to and including the critical threshold \cite{chen2023optimal, chen2025rapidmixinguniquenessthreshold}. As the current quantum techniques can only quadratically improve the run time of the first approach, we use the second apprach. More specifically, we consider Glauber chain $P_{\lambda}$ at $\lambda \leq \lambda_c$ so that we can prepare $\ket{\sqrt{\pi_{\lambda}}}$ efficiently. Next, we consider the short path Hamiltonian $H_b: \mathcal{X}\rightarrow \mathbb{R}$,
\[
H_b = -D(P_{\lambda}) + b g_{\eta}\left(\frac{H}{|E^{\star}|} \right),
\]
where $D$ is the discriminant matrix as usual and $E^{\star} = \min_{x} H(x) = -|x^{\star}|$. In accordance with the generalized short path framework, the algorithm starts from $\ket{\sqrt{\pi_\lambda}}$ and jumps to the ground state of $H_b$ for $b>0$. We describe how we can prepare the block-encoding of $D$ for Glauber dynamics and also prepare $\pi_{\lambda}$ in Appendix \ref{sec:block-encoding}. 

The following lemma establishes a bound on  $\|H\|_P$ for Glauber dynamics.
\begin{lemma}
\label{lem:mis-pseudo-lip-glauber}
    Let $P$ be a Glauber dynamics chain on the Hardcore model with fugacity parameter $\lambda$ on a graph $G(\mathcal{N}, \Ecal)$. Then, $\|H\|_P = \Ocal(1)$.
\end{lemma}
\begin{proof}
As Glauber dynamics flips one spin at a time, for all $x, x' \in \mathcal{X}$ such that $P(x, x')>0$ we have
 \[
    |H(x)-H(x')|\leq 1.
\]
Then, by Lemma \ref{lem:coordinate-Lipshitz}, $\|H\|_P = \Ocal(1)$.
\end{proof}

We have the following bound on the LS constant using the techniques from Section \ref{subsec:glauber_mixing_tech}.
\begin{corollary}[Log-Sobolev Constant For Constant-Fraction From Criticality]
\label{cor:log_sob_constant_frac_critical}
Let $G = (\Ncal, \Ecal)$ be any graph with constant-bounded degree $d$, $n$ vertices, and $\lambda_c$ the corresponding tree-uniqueness threshold for the hardcore model. Under the conditions of Theorem \ref{thm:entropy_fact}, for any constant $\xi \in (0, 1)$ and $\lambda \leq (1-\xi)\lambda_c$, the log-Sobolev constant of Glauber dynamics satisfies $\omega = \Omega(1/n)$.
\end{corollary}
\begin{proof}
In  \cite{chen2021rapidmixingglauberdynamics}, as discussed in \cite[Proof of Theorem 1.1]{chen2023optimal}, it was shown that for any $\xi \in (0, 1)$, $\lambda \leq (1-\xi)\lambda_c$ and that the spectral independence is $\rho = \mathcal{O}(1/\xi)$. Also, \cite[Proof of Theorem 1.1]{chen2023optimal} showed that $u$ is lower bounded by a constant in $\xi$ and $\lambda$.  Hence for graphs of constant maximum degree $d$, Theorem \ref{thm:entropy_fact} and Fact \ref{sobolev-marginally-bounded} imply that the Glauber dynamics at fugacity $\lambda$ has a log-Sobolev constant that is $\Omega(1/n)$ for bounded-degree graphs. 
\end{proof}

Having showed that Glauber dynamics chain satisfies desired log-Sobolev constant and bounded $\|H\|_P$, we present the final run time.

\begin{theorem}[Constant Fraction From Criticality]
    \label{thm:mis_constant_frac_crit}
    Let $G = (\mathcal{N},\mathcal{E})$ be any graph with a constant maximum degree $d$ and $\lambda_c$ the corresponding tree-unqiueness threshold for $d$. Let $P$ be a Glauber dynamics chain on a graph $G$ at fugacity parameter $\lambda \leq (1-\xi)\lambda_c$ for any constant $\xi \in (0, 1)$ and stationary distribution $\pi_{\lambda}$. Then, there exists a short path algorithm that finds the maximum independent set in $G$ with running time 
    \[
    \Ostar\left([\pi_{\lambda}(E^{\star})^{-1}]^{\frac{1}{2}-c} \right),
    \]
    where $c>0$ is a constant.
\end{theorem}
\begin{proof}
    It is evident that for sparse graphs ($d = \Ocal(1)$), the size of the maximum independent set is $\Theta(n)$. Since $d$ is constant in $n$, the problem still can not be solved in $\poly(n)$ time and requires $2^{\Ocal(n)}$ samples for hard instances. Hence, $\log(1/\pi_{\lambda}(E^{\star})) = \Theta(n)$ .  The log-Sobolev constant from Corollary \ref{cor:log_sob_constant_frac_critical} satisfies $\omega^{-1} = \Ocal(n)$ for the Glauber dynamics on bounded degree graphs for any constant-fraction below the critical fugacity $\lambda_c$. Also Lemma \ref{lem:mis-pseudo-lip-glauber} implies a constant $P$-psuedo Lipschitz norm. Therefore, the conditions of Corollary \ref{cor:super_grover_run} hold. Hence, we obtain a super-quadratic speedup over sampling from $\pi_{\lambda}$.
\end{proof}

The above result only applies for a constant fraction from the critical threshold $\lambda_c$, i.e. $\xi$ is constant. As mentioned earlier, Glauber dynamics is known to mix in polynomial time when $\xi = 0$. However, it appears that the MLS and LS vanish at this threshold, while  the spectral gap remains polynomially small. This comes from showing approximate tensorization of variance (Section \ref{subsec:glauber_mixing_tech}).

\begin{lemma}[Inverse-polynomial Spectral Gap at Criticality]
\label{lem:spec_gap_hard_core_at_crit}
    Let $G = (\Ncal, \Ecal)$ be a graph with constant maximum degree, $n$ vertices, and corresponding tree-unqiuness threshold for the $d$ and the hardcore model $\lambda_c$. The spectral gap of Glauber dynamics at fugacity $\lambda_c$ satisfies $\delta = \Omega(n^{-1-{4e\left(1 + \frac{1}{d-2}\right)}})$.
\end{lemma}
\begin{proof}
Follows from comments near the end of the \cite[Proof of Theorem 1.1]{chen2025rapidmixinguniquenessthreshold}. 
\end{proof}

In this case, we can apply Theorem \ref{thrm:MTG_RT_poincare} and obtain an asymptotically falling super quadratic speedup at criticality.
\begin{theorem}[Bounded Degree Graphs At Criticality]
    \label{thm:mis_at_criticality_runtime}
    Let $P$ be a Glauber dynamics chain on a graph $G(\mathcal{N},\mathcal{E})$ with a bounded maximum degree $d = \mathcal{O}(1)$ at fugacity parameter $\lambda_c$ and stationary distribution $\pi_{\lambda_c}$.  Also suppose that $\gamma$-spectral density is satisfied for $\gamma = \Theta(1)$. Then, there exists a short path algorithm that finds the maximum independent set in $G$ with running time 
    \[
    \Ostar\left([\pi_{\lambda_c}(E^{\star})^{-1}]^{\frac{1}{2}-c(n)} \right),
    \]
    where $c(n) = \Omega(1/n^{1+4e\left(1 + \frac{1}{d-2}\right)})$.
\end{theorem}
\begin{proof}
    Follows from  Theorem \ref{thrm:MTG_RT_poincare}, Lemma \ref{lem:spec_gap_hard_core_at_crit} and Lemma \ref{lem:mis-pseudo-lip-glauber}. 
\end{proof}

As mentioned earlier, for random regular graphs, the Glauber dynamics mixes beyond $\lambda_c$ with high probability. This appears to hold all the way up to  $\lambda= \Ocal(1/\sqrt{d})$ \cite{chen2025rapidmixingrandomregular}. Specifically, we have the following lemma extracted from \cite{chen2025rapidmixingrandomregular}.

\begin{lemma}[Spectral Gap For Random Regular Beyond Criticality -- Theorem 5.2 (2)\cite{chen2025rapidmixingrandomregular}]
\label{lem:spec_gap_beyond_crit_reg}
    Let $G\sim \mathcal{G}_{\textup{RG}}(n, d)$ be a random $d$-regular graph for constant $d$. Consider the fugacity $\lambda = \frac{1-\xi}{2\sqrt{d-1}-1}$. Then with high probability $\pi_{\lambda}$ has a spectral gap lower bounded by $\delta = \Omega(\xi/n)$.
\end{lemma}

However, it appears that a lower bound was not provided fro the MLS or LS constants, i.e. \cite[Theorem 5.2]{chen2025rapidmixingrandomregular} becomes vacuous for the hardcore model. We again  apply Theorem \ref{thrm:MTG_RT_poincare} and obtain a result like Theorem \ref{thm:mis_at_criticality_runtime}, providing an asymptotically falling super-quadratic speedup beyond the uniqueness threshold for random regular graphs.

\begin{theorem}[Random Regular Graphs Beyond Criticality]
\label{thm:mis_beyond-criticality_runtime}
    Let $P$ be a Glauber dynamics chain on a random $d$ regular-graph $G\sim \mathcal{G}_{\textup{RG}}(n, d)$, with constant $d$ and fugacity parameter $\lambda \leq \frac{1-\xi}{2\sqrt{d-1}-1}$, $\xi \in (0, 1)$, and stationary distribution $\pi_{\lambda}$. Also suppose that $\gamma$-spectral density is satisfied for $\gamma = \Theta(1)$. Then with high probability over the choice of graph, there exists a short path algorithm that finds the maximum independent set in $G$ with running time 
    \[
    \Ostar\left([\pi_{\lambda}(E^{\star})^{-1}]^{\frac{1}{2}-c(n)} \right),
    \]
    where $c(n) = \Omega(\xi n^{-1})$.
\end{theorem}
\begin{proof}
    Follows from  Theorem \ref{thrm:MTG_RT_poincare}, Lemma \ref{lem:spec_gap_beyond_crit_reg} and Lemma \ref{lem:mis-pseudo-lip-glauber}. 
\end{proof}

Even with the above results, it is still not clear how Markov chain search compares to the best classical algorithms for MIS problem. As we consider a Markov chain that mixes to Gibbs distribution with fugacity $\lambda$, sampling from the Gibbs distribution may favor the smaller independent sets when $\lambda <1$. However, even if this is the case, the stationary distributions\ only has support on the constrained space (i.e. on the valid independent sets) rather than on the entire $2^n$ possible configurations. Therefore, we might still expect to prove that the success probability is larger than brute force search for certain graph models. To this end, the following proposition lower bounds the probability of finding the optimum by sampling from Glauber dynamics when restricting to \emph{random regular graphs}.

\begin{proposition}     
\label{prop:comparison-to-brute-force}
Let $\pi_{\lambda}$ be the stationary distribution of Glauber dynamics on a random $d$-regular graph $G\sim \mathcal{G}_{\textup{RG}}(n, d)$. Choose  $x^{\star}$ to be a particular maximum independent set in $G$. %
Then we have
$$
    \pi_\lambda(x^{\star}) \geq  2^{-\kappa n},
$$
 with $\kappa =  -\frac{2\log(\lambda)\log d  }{d} + \frac{1}{2d}+ \frac{\log(1+\lambda)}{2}$.
\end{proposition}
\begin{proof}
The probability of $x^{\star}$ in $\pi_{\lambda}$ is given by
$$
  \pi_{\lambda}(x^{\star}) = \frac{\lambda_c ^{|x^{\star}|}}{\sum_{x\in \mathcal{X}} \lambda ^{|x|} }.
$$
We first bound the denominator. To do that, we invoke \cite[Theorem 2]{ZHAO_2009}, which asserts that 
$$
     Z(\lambda)=\sum_{x\in\mathcal{X}  } \lambda ^{|x|} \leq (2 (1+\lambda )^d -1)^{\frac{n}{2d}},
$$
for any $d$-regular graph.

For the numerator, we need to bound $|x^{\star}|$. A known upper bound for the size of the maximum independent set is given by $2\log(d)n/d$ \cite{73f24c1e-86b4-384b-98bc-a8e418257ab8}. Combining these, we have
$$
    \pi(x^{\star}) \geq \frac{\lambda ^{(2 \frac{n}{d}\log d)}}{(2 (1+\lambda )^d -1)^{\frac{n}{2d}}}
    \geq  \frac{\lambda ^{(2 \frac{n}{d}\log d)}}{(2 (1+\lambda )^d )^{\frac{n}{2d}}} 
    = 2^{n(\frac{2\log(\lambda)\log d  }{d} - \frac{1}{2d}- \frac{\log(1+\lambda)}{2} )}.
$$
\end{proof}

Note that for any constant $\xi \in (0, 1)$, $\lambda = (1-\xi)\lambda_c$ and $\lambda_c = \frac{(d-1)^{(d-1)}}{(d-2)^d}$, we have that $\kappa < 1$ for sufficiently large degree $d > d_0$ for some $d_0$.
Hence the runtime is smaller than $2^n$, which is the runtime of brute force search over unconstrained space. Therefore, Markov chain search for this problem might be significantly faster than brute force search mainly due to fact that Glauber dynamics stays in the constrained space.

\begin{remark}
    The parameter $\kappa$ from Proposition \ref{prop:comparison-to-brute-force} is a decreasing function of the degree. For sufficiently high degree, the runtime $2^{\kappa n}$ the is in fact faster than the best known generic algorithm for Maximum Independent Set (runtime of $\tilde{\mathcal{O}}(1.1996^n)$, due to Xiao and Nagamuchi~\cite{xiao2017exact}). 
\end{remark}

\subsubsection{Antiferromagnetic Ising Model}
\label{sec:ising-glauber}

Consider the 2-spin Ising model on a graph $G=(\Ncal, \mathcal{E})$ defined via the Hamiltonian  
\begin{align*}
    H(x) = \sum_{(i,j) \in \Ecal} J_{ij}x_ix_j+\sum_{j}h_jx_{j},
\end{align*}
where the entries of $J$ are interaction coefficients and $h$ defines an external field.

We will consider the antiferromagnetic case $J_{ij} < 0$ with no external-field. Specifically, let $J_{ij} = -1$ for all $(i,j) \in \Ecal$ and $h = \mathbf{0}$. This model corresponds to Gibbs sampling with weights $\pi(x) \propto \exp(-\beta H(x))$, and a Gibbs sample can be prepared by using Glauber dynamics similar to the hardcore model. The only change is the transition probabilities, which can be computed by the marginal distribution $\pi(x_i^{t+1}|x^t_{\Ncal \setminus \{i\}})$. %

We consider the following short path Hamiltonian
\[
H_b(x) = -D(P_{\beta}(x)) +b g_{\eta}\left(\frac{H}{|E^{\star}|}\right),
\]
where $P_{\beta}$ is the Glauber dynamics transition matrix for inverse-temperature $\beta$. Similar to the setting of the MIS problem, a block-encoding of $D(P_{\beta})$ can be prepared efficiently.

First we have the following simple result bounding the psuedo-Lipschitz constant.
\begin{lemma}
\label{lem:ising-pseudo-lip-glauber}
    Let $P$ be a Glauber dynamics chain on the antiferromagnetic Ising model with inverse-temperature  $\beta$ on a graph $G(\mathcal{N}, \Ecal)$ with degree at most $d$. Then, $\|H\|_P = \Ocal(d)$.
\end{lemma}
\begin{proof}
As Glauber dynamics flips one spin at a time, for all $x, x' \in \mathcal{X}$ such that $P(x, x')>0$ we have
 \[
    |H(x)-H(x')|\leq 2d.
\]
Then, by Lemma \ref{lem:coordinate-Lipshitz}, $\|H\|_P = \Ocal(d)$.
\end{proof}

\begin{corollary}[Log-Sobolev Constant at Constant-Fraction away From Criticality]
\label{cor:ising_log_sob_constant_frac_critical}
Let $G = (\Ncal, \Ecal)$ be any graph with constant-bounded degree $d$, $n$ vertices, and $\beta_c$, $\beta$ satisfies $(d-1)\tanh(\beta) \leq 1$ the corresponding tree-uniquness threshold for the antiferromagnetic Ising Model. If $\beta$ satisfies $(d-1)\tanh(\beta) \leq 1 - \xi$  for any constant $\xi \in (0, 1)$, then under the conditions of Theorem \ref{thm:entropy_fact}, the log-Sobolev constant of Glauber dynamics satisfies $\omega = \Omega(1/n)$.
\end{corollary}
\begin{proof}
From \cite{chen2023optimal} and \cite[Lemma 3.24]{chen2025rapidmixinguniquenessthreshold} for the specified $\beta$ range $\eta = \mathcal{O}(1)$, when $d$ and $\xi$ are constants. Hence Theorem \ref{thm:entropy_fact} implies approximate entropy factorization with $C_1 = \mathcal{O}(1)$. \cite[Proof of Theorem 1.1]{chen2023optimal} gives that $\pi_{\beta}$ satisfies the bounded marginals condition when $d = \mathcal{O}(1)$. Hence, MLSI can be promoted to LSI when using approximate tensorization of entropy. Hence the LS constant is lower bounded by $\Omega(n^{-1})$.
\end{proof}

\begin{theorem}[Constant Fraction away From Criticality]
\label{thm:ising_constant_frac_crit}
    Let $P$ be a Glauber dynamics chain for the antiferromagnetic Ising model over an arbitrary graph $G=(\mathcal{N},\mathcal{E})$  with constant maximum degree $d$.  Suppose that for constant $\xi \in (0, 1)$, the inverse temperature parameter $\beta$ satisfies $(d-1)\tanh(\beta) \leq 1 - \xi$ and that the stationary distribution is $\pi_{\beta}$. If $\lvert E^{\star} \rvert = \Theta(n)$, then there exists a short path algorithm that finds the ground state of the Ising model Hamiltonian running time bounded by
    \[
    \Ostar \left([\pi_{\beta}(E^{\star})^{-1}]^{ \frac{1}{2}-c} \right),
    \]
    where $c > 0$ is a constant.
\end{theorem}
\begin{proof}
We have that $\omega = \Omega(n^{-1})$ for specified $\beta$ range from Corollary \ref{cor:ising_log_sob_constant_frac_critical}. Furthermore, $\log(1/\pi_{\beta}(E^{\star})) = \Theta(n)$ due to hardness of the problem. From Lemma \ref{lem:ising-pseudo-lip-glauber} $\lVert H \rVert_{P}= \Ocal(d) = \Ocal(1)$. The rest follows by plugging in the parameters to Theorem \ref{thrm:MTG_RT}.
\end{proof}

We also have the following result for at and slightly beyond criticality.

\begin{corollary}[Log-Sobolev Constant at Criticality]
\label{cor:ising_log_sob_at_critical}
Let $G = (\Ncal, \Ecal)$ be any graph with constant-bounded degree $d$, $n$ vertices, and $\beta$ satisfies $(d-1)\tanh(\beta) \leq 1 + o_n(1)$. Then under the conditions of Theorem \ref{thm:entropy_fact}, the log-Sobolev constant of Glauber dynamics satisfies $\omega = \Omega(1/n^{2+1/d})$.
\end{corollary}
\begin{proof}
\cite[Proof of Theorem 1.1]{chen2023optimal} gives that $\pi_{\beta}$ satisfies the bounded marginals condition when $d = \mathcal{O}(1)$. Hence, MLSI can be promoted to LSI when using approximate tensorization of entropy. From \cite[Proof of Theorem 3.23]{chen2025rapidmixinguniquenessthreshold} we have
that Glauber dynamics for the above specified range of $\beta$ satisfies approximate tensorization of entropy with $C_1 = \mathcal{O}(n^{1 + \frac{2}{d-2}})$. Hence the LS constant is lower bounded by $\Omega(n^{-2-1/d})$.
\end{proof}

\begin{theorem}[Bounded Degree Graphs at Criticality]
\label{thm:ising_at_critical_runtime}
    Let $P$ be a Glauber dynamics chain for the antiferromagnetic Ising model over an arbitrary graph $G(\mathcal{N},\mathcal{E})$  with maximum constant degree $d$.  Suppose the inverse temperature parameter $\beta$ satisfies $(d-1)\tanh(\beta) \leq 1 + o_{n}(1)$ and that the stationary distribution is $\pi_{\beta}$. If $\lvert E^{\star}\rvert = \Theta(n)$ and $\gamma$-spectral density is satisfied for $\gamma = \Theta(1)$, then there exists a short path algorithm that finds the ground state of the Ising model Hamiltonian running time bounded by
    \[
    \Ostar \left([\pi_{\beta}(E^{\star})^{-1}]^{ \frac{1}{2}-c(n)} \right),
    \]
    where $c(n) = \Omega(n^{-1-1/d})$.
\end{theorem}
\begin{proof}
From Corollary \ref{cor:ising_log_sob_at_critical}, we have a LS constant lower bounded by $\Omega(n^{-2 - \frac{2}{d-2}})$. The rest of the proof follows the same as Theorem \ref{thm:ising_constant_frac_crit}.
\end{proof}

We also have a super-quadratic speedup for the Ising model over random-regular graphs. However, unlike for the hardcore model, the speedup is an asymptotic super-quadratic speedup and does not fall with $n$. We start with a proposiiton bounding the size of $\lvert E^{\star}\rvert$ over random regular graphs.

\begin{proposition}
    \label{prop:Ising-ground-state-energy}
    The optimum energy of Ising Model Hamiltonian $H$ on a random regular graph $G\sim \mathcal{G} (n, d)$ satisfies $|E^{\star}| = \Theta(n)$ with high probability.  
\end{proposition}
\begin{proof}
Let $s$ denote the number of edges in the graph. The ground state of $H$ can be related to minimum bisection width \cite{Zdeborov__2010} denoted by $|\textrm{BW}|$ as follows
$$
    |\textrm{BW}| = \frac{s+E_{\textrm{gs}}}{2}.
$$
Using this equality, $E_{\textrm{gs}}  = s - 2|BW|$. Next, we consider random regular graphs. For sparse random regular graphs $s = \Theta(n)$ and $|\textrm{BW}| = \Theta(n)$ (See \cite{DIAZ2007120, cojaoghlan2020isingantiferromagnetmaxcut}).
\end{proof}

\begin{lemma}[Log-Sobolev Beyond Criticality For Random Regular]
\label{lem:ising_log_sob_regular}
     Let $G\sim \mathcal{G} (n, d)$ be a random $d$-regular graph for constant $d$.  Suppose the inverse temperature parameter $\beta$ satisfies $(d-1)\tanh(\beta) \leq \frac{d-1}{8\sqrt{d-1}-1}$. If $d = \mathcal{O}(1)$, then with high probability over the choice of graph, the log-Sobolev constant of Glauber dynamics for the antiferromagnetic Ising model satisfies $\omega = \Omega(1/n)$.
\end{lemma}
\begin{proof}
\cite[Proof of Theorem 1.1]{chen2023optimal} gives that $\pi_{\beta}$ satisfies the bounded marginals condition when $d = \mathcal{O}(1)$. Hence, MLSI can be promoted to LSI when using approximate tensorization of entropy.  From \cite[Theorem 5.2]{chen2025rapidmixingrandomregular} we have MLS for $(d-1)\tanh(\beta_c) \leq \frac{d-1}{8\sqrt{d-1}-1}$ scaling as $\Omega(\frac{e^{-\sqrt{d}}}{n})$ for constant $d$.
\end{proof}

\begin{theorem}[Beyond Criticality for Random Regular Graphs]
\label{thm:ising_beyond_crit_runtime}
    Let $P$ be a Glauber dynamics chain for the antiferromagnetic Ising model over a randomly chosen $d$-regular graph $G\sim \mathcal{G}_{\textup{RG}}(n, d)$.  Suppose the inverse temperature parameter $\beta$ satisfies $(d-1)\tanh(\beta) \leq \frac{d-1}{8\sqrt{d-1}-1}$ and that the stationary distribution is $\pi_{\beta}$. If $d = \mathcal{O}(1)$, then with high probability over the choice of graph there exists a short path algorithm that finds the ground state of the Ising model Hamiltonian running time bounded by
    \[
    \Ostar \left([\pi_{\beta}(E^{\star})^{-1}]^{ \frac{1}{2}-c} \right),
    \]
    where $c > 0$ is a constant.
\end{theorem}
\begin{proof}
From Lemma \ref{lem:ising_log_sob_regular}, we have with high probability over the graph $\omega = \Omega(n^{-1})$. Furthermore, $\log(1/\pi_{\beta}(E^{\star})) = \Theta(n)$ due to hardness of the problem. Again, $\lVert H\rVert_{P} = \Ocal(d) =\Ocal(1)$. Then the result follows from Corollary \ref{cor:super_grover_run} applied to the case of Theorem \ref{thrm:MTG_RT}.
\end{proof}

As mentioned earlier the inverse temperature bound, call it $\beta_T$, in Theorem \ref{thm:ising_beyond_crit_runtime} is very close to a known computational threshold, called the disorder chaos threshold. Specifically, it is known that with high probability, there exists a problem size $n_0$ such that no stable classical algorithm can mix in Wasserstein-2 distance when the inverse temperature exceeds $\frac{1}{\sqrt{d}}$ and $n \geq n_0$ \cite{huang2024hardnesssamplingantiferromagneticising}. The quantum algorithm is improving upon a Gibbs sampler that mixes in TV to a point near this threshold. For a bounded domain, TV mixing implies Wasserstein-2 mixing. Unfortunately, the current provable edge that quantum has over the Gibbs sampler at the current known threshold $\beta_T < \frac{1}{\sqrt{d}}$ falls faster with $d$ than $\frac{1}{\sqrt{d}} - \beta_T$,  implying it is currently unclear if quantum's edge is enough to breach the Gibbs sampling threshold.

The following result rigorously shows that with high probability Markov chain search outperforms brute-force search for random regular graphs.
\begin{proposition}
\label{prop:ising_improve_over_brute_force}
    For the antiferromagnetic Ising model over random $d$-regular graphs, for any $\beta > 0$, the probability of observing the ground state under the Gibbs measure is at least $2^{-\kappa n}$. With high probability in $n$, $\kappa > -1$.
\end{proposition}
\begin{proof}
    The replica symmetric bound for antiferromagnetic Ising over random regular graphs \cite[Theorem 1.1]{cojaoghlan2020isingantiferromagnetmaxcut} gives that
\begin{align*}
        \frac{\ln(Z(\beta))}{n} \leq \ln(2) + \frac{d}{2}\ln\left(\frac{1+e^{-\beta}}{2}\right) + \epsilon(n),
    \end{align*}
with high probability in $n$. Hence
\begin{align*}
    \frac{\log(\pi(E^{\star}))}{n} \geq \frac{\beta \lvert E^{\star}\rvert}{n} - \frac{d}{2
    \ln(2)}\ln\left(\frac{1+e^{-\beta}}{2}\right) - 1 + \epsilon(n).
\end{align*}
If we apply Proposition \ref{prop:Ising-ground-state-energy}, then the above is  $> -1$ for $\beta > 0, d \geq 0$ and sufficiently large $n$.
\end{proof}

\subsubsection{Sherrington-Kirkpatrick Model}
Consider the Sherrington-Kirkpatrick Hamiltonian
\begin{align}
    H(x) = \frac{1}{\sqrt{n}}\sum_{i < j \leq n}^n g_{ij} x_i x_j,
\end{align}
where the interaction coefficients $g_{ij}$ are i.i.d. standard Gaussian random variables. We can Gibbs sample from the following distribution,
\begin{align*}
    \pi(x) \propto \exp(-\beta H(x))
\end{align*}
using Glauber dynamics. The following provides a bound on $\Delta_{P}$ for Glauber over SK.
\begin{lemma}
\label{lem:delta-sk}
Let $P$ be a Glauber dynamics chain for the SK model on a graph $G(\mathcal{N},\Ecal)$. Then, $$\Delta_P = \Ocal(1).$$ 
\end{lemma}
\begin{proof}
    We first consider hypercube walk. If we flip a spin at random, the sign of each term in $H$ will flip with probability $2/n$. Therefore the energy of each term increases at most by a factor of  $1-\frac{4}{n}$ in expectation. Since the ground state energy $|E^{\star}| = \Theta(n)$, the energy increases at most by constant in expectation. From the definition of Glauber dynamics a bit flip is proposed uniformly, and accepted with probability larger than $\frac{1}{2}$ if $\beta (H(x') - H(x))<0$. Therefore, if we are running the Glauber dynamics at some positive finite temperature, moves that increase energy are made with strictly lower probability than the hypercube walk. Thus an uppper bound on stability with respect to the hypercube walk is also a valid upper bound for Glauber dynamics at positive $\beta$.
\end{proof}
\begin{lemma}
    \label{lem:lsi-sk}
    Let $P$ be a Glauber dynamics chain for SK model on a graph $G(\mathcal{N},\Ecal)$ at inverse temperature  $\beta < \frac{1}{4}$. Then, the log-Sobolev constant satisfies $\omega \geq \Omega(1/(n\log n))$.
\end{lemma}

\begin{proof}
    By \cite[Corollary 51]{chen2022localizationschemesframeworkproving}%
    , the modified log-Sobolev constant for Glauber dynamics is $\Omega(1/n)$ when $\beta < \frac{1}{4}$. For possible transitions, the transition probability of Glauber dynamics is $\Omega(n^{-1})$. Hence, by Theorem \ref{thm:modified-LSI}, the log-Sobolev constant scales as $\Omega(1/(n\log n))$.
\end{proof}

\begin{theorem}
\label{thm:sk_runtime}
    Let $P$ be a Glauber dynamics chain for the Sherrington-Kirkpatrick model at inverse temperature parameter $\beta < \frac{1}{4}$ and stationary distribution $\pi_{\beta}$. Then, there exists a short path algorithm that finds the optimal solution of the Sherrington Kirkpatrick Hamiltonian with running time 
    \[
    \Ocal \left(\poly(n)[\pi_{\beta}(E^{\star})^{-1}]^{\left(\frac{1}{2}-\frac{c}{\log(n)}\right)} \right),
    \]
    where $c>0$ is a constant.
\end{theorem}
\begin{proof}
    We first show that the tail bound holds for SK model. By using proposition 4 in \cite{dalzell2022mind}, we know that the number of low energy states with energy smaller than $E^{\star}(1-\eta)$ is smaller than $2^{\gamma n}$ where $\gamma$ is a constant. By assuming that $\log(1/\pi(E^{\star})) = \Theta(n)$, we can conclude that the generalized tail bound holds as well. Note that if this assumption fails, then it means that there exists a sub-exponential solver for SK model. Finally, since $\lvert E^{\star}\rvert= \Theta(n)$ and 
 $\log(1/\pi(E^{\star})) = \Theta(n)$. Thus, by  Lemma \ref{lem:lsi-sk} and Theorem \ref{thrm:b_log_sob}, we have $b = \Ocal(1/\log(n))$. Since $\Delta$ is constant by Lemma \ref{lem:delta-sk}, the total runtime scales as $(\pi_{\beta}(E^{\star}))^{-1\left(\frac{1}{2}-\frac{c}{\log(n)}\right)}$ due to Theorem \ref{thrm:MTG_RT}.
\end{proof}

We conclude this section by demonstrating that Markov chain search using Glauber dynamics at a positive inverse-temperature is faster than unstructured search.

\begin{proposition}
\label{prop:sk_positive-beta-overlap}
Let $\pi_{\beta}$ be the Gibbs distribution corresponding to the SK cost function $H \colon \{0,1\}^n \rightarrow \mathbb{R}$ at some positive inverse temperature $\beta > 0$. Then the probability of observing the ground state under the Gibbs measure is at least $2^{-\kappa n}$. Where, $\kappa > -1$ for sufficiently large $n$.
\end{proposition}
\begin{proof}
    Let $Z(\beta)$ be the partition function associated with $\pi_{\beta}$. We know that $\frac{\log Z(\beta)}{n}$ has subGaussian tails around its mean for every fixed $\beta$ \cite[Equation 1.54]{talagrand2010mean}. A loose version of Guerra's replica symmetric bound \cite[Theorem 1.3.7]{talagrand2010mean} gives that
    \begin{align*}
        \frac{\mathbb{E}[\log Z(\beta)]}{n} \leq 1 + \frac{\mathbb{E}[\ln \cosh(\beta z)]}{\ln(2)}.
    \end{align*}
    Hence with high probability if $\frac{\lvert E^{\star}\rvert}{n} =: \lvert P^{\star}\rvert  \approx .763\dots + o_{n}(1)$
    \begin{align*}
        \frac{\log(\pi(E^{\star}))}{n} \geq \frac{\beta}{\ln(2)} \left(\lvert P^{\star} \rvert - \frac{\mathbb{E}[\ln \cosh(\beta z)]}{\beta}\right) - 1 + \epsilon(n),
    \end{align*}
    which is $> -1$ for sufficiently large $n$.
\end{proof}

\section{Numerical Results}

\label{sec:numerical_results}

We perform numerical evaluations to empirically verify our findings. We focus on the constrained problems studied in this work, including \eqref{e:MaxBisection}, MaxCut with a Hamming weight constraint $k=o(n)$ \eqref{e:MaxCut}, and MIS with a penalized objective. For \eqref{e:MaxBisection}, we take $n$ to be even since $k=\frac{n}{2}$. 
For \eqref{e:MaxCut}, we take $k=\lfloor\sqrt{n}\rfloor$. For MIS, we take the penalty factor to be $n$, such that no energy reduction from constraint violation can justify the penalty.
For all three problems, we generate 100 random unweighted graphs for each $n$ from the Erdős–Rényi model with the probability of each edge existing to be $\frac{2\ln n}{n}$. The constant factor 2 is chosen to ensure a reasonable graph density at the scale we cover. We set $\eta=0.5$ in all experiments.

To improve the scalability of our numerical experiments, we construct $H_b$ as a sparse matrix in the compressed sparse row format and employ a GPU-accelerated iterative eigensolver to compute only the two smallest eigenvalues and the corresponding eigenvectors. For \eqref{e:MaxBisection} and \eqref{e:MaxCut}, which are explicitly constrained, we can scale up further by projecting $H_b$ onto the space spanned by all feasible states. The dimension of the computational space then drops from $2^n$ to $\binom{n}{k}$. With these efforts, we obtained results with up to 30 qubits for \eqref{e:MaxCut}. %

\begin{figure}[t]
    \centering
    \includegraphics[width=\linewidth]{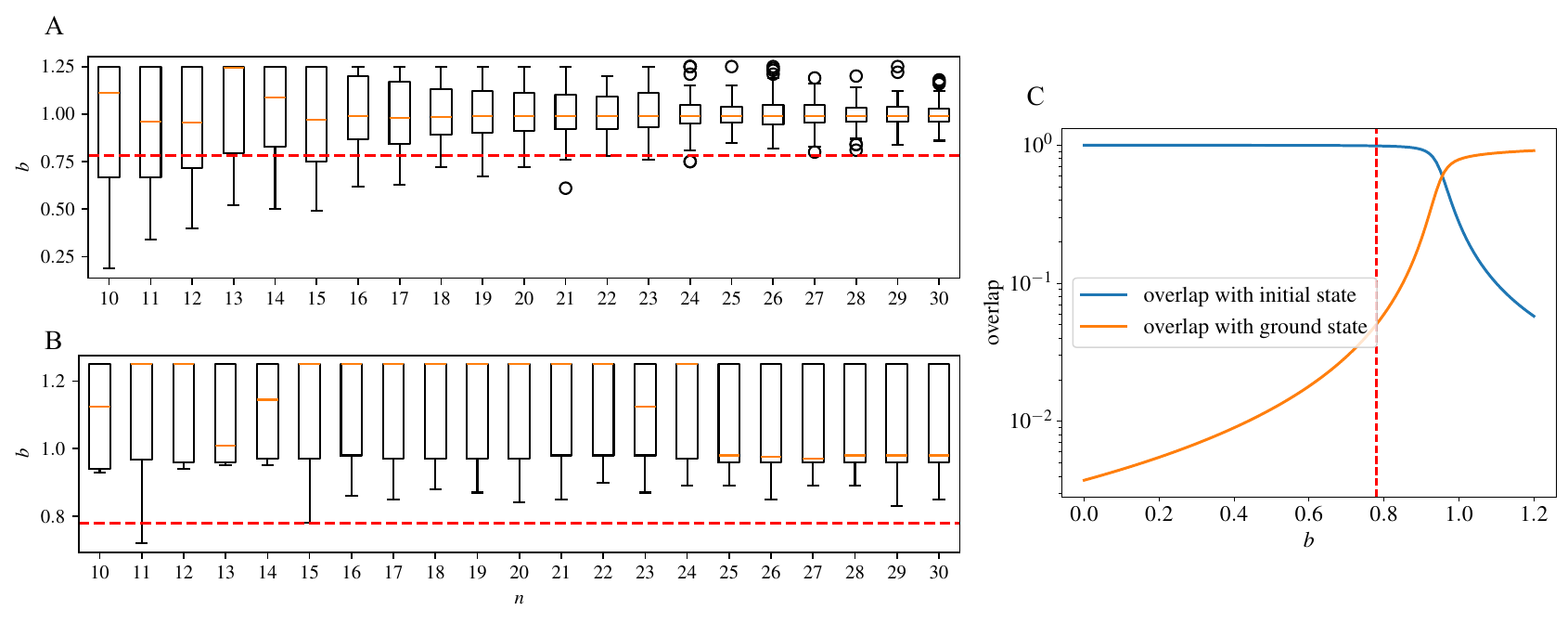}
    \caption{Empirical selection of $b$. 
    \textbf{A} Quartiles of $b$ values that minimize the effective runtime of the algorithm for \eqref{e:MaxCut}. As $n$ increases, the runtime-optimal $b$ converges to a range approximately between 0.8 and 1.2. The red dot line shows the converged value of $b \approx 0.78$ of phase transition where the overlap with the initial state crosses $0.99$. 
    \textbf{B} Quartiles of $b$ values that minimize the spectral gap for \eqref{e:MaxCut}. 
    For most instances tested, the spectral gap is minimized when $b$ is larger than the phase transition value, rendering the phase transition $b$ a safe choice. 
    \textbf{C} The overlap values with the initial state and the ground state (optimal solution) for one $n=30$ \eqref{e:MaxCut} instance with varying $b$. The dotted verticle line denotes the phase transition $b$.}
    \label{fig: mkc-chooseb}
\end{figure}

\begin{figure}[t]
    \centering
    \includegraphics[width=0.65\linewidth]{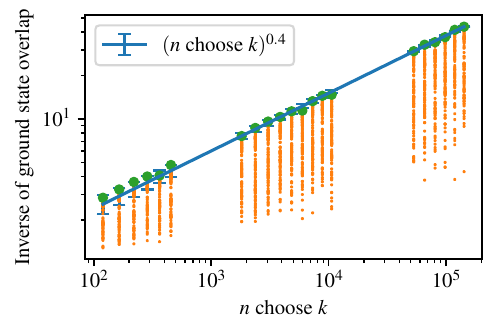}
    \caption{
    The inverse of the ground state overlap versus the feasible space size $\binom{n}{k}$ for \eqref{e:MaxCut} with $n$ varying from 10 to 30 and $b=$ 0.78. The worst-case instances are fitted using an exponential function with base $\binom{n}{k}$ with an error bar denoting one standard deviation of the fitted exponent. The 95\% confidence interval on the fitted exponent is $[0.391, 0.408]$.
    }
    \label{fig: mkc-runtime}
\end{figure}

First, we want to identify what values of $b$ are practically appropriate. In~\cite{dalzell2022mind}, the authors numerically show that the original short path algorithm works well for the 3-spin problem for $b$ up to around 0.8, which is much larger than the theoretical bound of $b\leq 1.02\times 10^{-4}$. Here, we show that a similar observation can be made for the constrained and penalized cases. For \eqref{e:MaxCut}, Figure~\ref{fig: mkc-chooseb} A shows the quartiles of $b$ values that minimize the effective runtime (Equation~\ref{eqn: runtime}) of the algorithm. Note that the $b$ values are hard-capped at 1.25 and may be higher. We see that as $n$ increases, the optimal $b$ converges to a range approximately between 0.8 and 1.2. However, when choosing the value of $b$ that needs to work for all instances, we want a conservative value that avoids encountering the possibly superexponentially small spectral gap. For this purpose, we identify the value of $b$ at which the phase transition occurs. We numerically characterize the phase transition point by the overlap of $\ket{\psi_b}$ with the initial state dropping below 0.99. An example of the overlap with varying $b$ is shown in Figure~\ref{fig: mkc-chooseb} C. Empirically, we observe that the phase transition $b$ converges to around $0.78$ as $n$ increases.  
In Figure~\ref{fig: mkc-chooseb} B, we plot the quartiles of $b$ values that minimize the spectral gap. For most instances, the spectral gap is minimized when $b$ is greater than the phase transition value $\approx 0.78$. 
Therefore, we expect a ubiquitous value of phase transition to work for a \eqref{e:MaxCut} instance with high probability.

We then fit the worst-case runtime to empirically demonstrate the super-Grover speedup. 
Although the inverse of the spectral gap term in the runtime (Equation~\ref{eqn: runtime}) has a $\Ocal(\poly(n))$ complexity, it may still affect the exponential fitting at the scale of numerical experiments. Thus, we use the inverse of ground state overlap $\lvert \langle \psi_b|z\rangle\rvert^{-1}$, the only exponential growth term in the runtime, to fit the asymptotic speedup.
In Figure~\ref{fig: mkc-runtime}, we set $b$ to be 0.78 and plot the inverse of ground state overlap $\lvert \langle \psi_b|z\rangle\rvert^{-1}$ of all \eqref{e:MaxCut} instances with respect to $\binom{n}{k}$, the size of the feasible space. We fit the worst-case instances using an exponential function with base $\binom{n}{k}$, the exponent of which is equivalent to the factor $a$ in $2^{an}$ for the unconstrained case. The error bar of the fitted line denotes one standard deviation of the fitted exponent. We see the empirical $b$ values give a super-Grover speedup, which is much better than the theoretically guaranteed bounds.%

In Figure~\ref{fig:multiple_problem_runtime}, we show the empirical worst-case scaling for all three problems with different choices of $b$. A proper selection of $b$ yields a super-Grover speedup across all examined problems. Conversely, when $b$ is excessively high, the algorithm may encounter a small spectral gap in the worst case. To demonstrate this, we use the runtime (Equation~\ref{eqn: runtime}, which includes the inverse of the spectral gap term) 
as the metric and observe that the quality of the fitting degrades. Our numerics lead to two interesting conceptual observations: firstly, as observed also by~\cite{dalzell2022mind} the optimal choices of $b$ are well beyond what is predicted by the theoretical analysis. Secondly, in the case of \eqref{e:MaxCut} we numerically observe an advantage over quadratic speedup that does not decay with $n$ which is beyond the current theoretical analysis and indicates that the runtime of the long jump can possibly be characterized through weaker conditions than $\Delta_{P}$ stability. Finally, we confirm in our setting that is indeed reasonable to make the choice of $b$ by choosing the largest such value that allows for large overlap with the ground state. If the value of this critical $b$ asymptotes quickly as a function of $n$ this suggests a numerical mechanism for the development of efficient short path algorithms.

\begin{figure}[t]
    \centering
    \includegraphics[width=\linewidth]{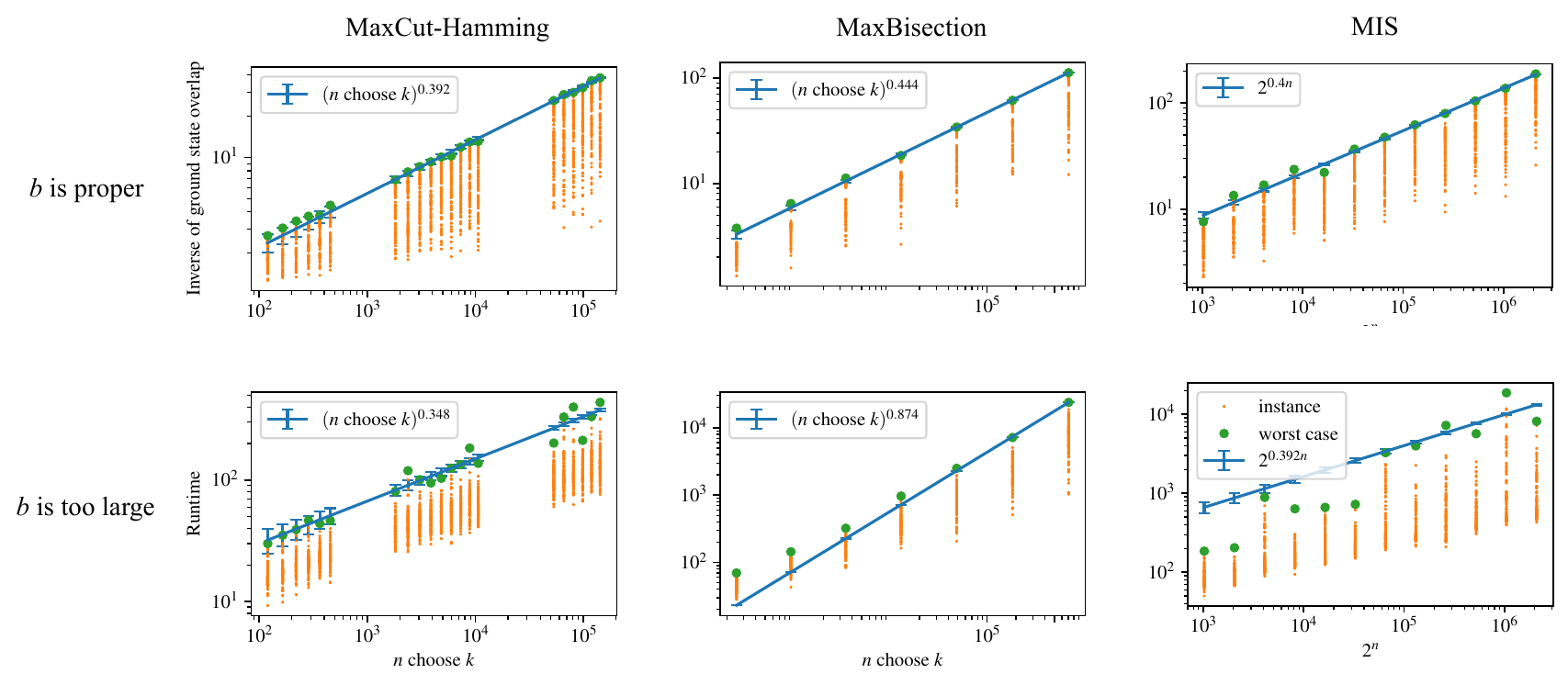}
    \caption{
    Empirical fitting of the inverse overlap of the ground state $\lvert \langle \psi_b|z\rangle\rvert^{-1}$ and the runtime~\ref{eqn: runtime} of \eqref{e:MaxCut}, \eqref{e:MaxBisection}, and MIS with difference choices of $b$. 
    For the left column \eqref{e:MaxCut}, we use data with $n$ ranging from 10 to 30. Top: $b=0.8$ the fitted exponent is 0.392 with 95\% confidence interval $[0.383, 0.402]$, indicating there could exist $b$ that is better than the phase transition one in Figure~\ref{fig: mkc-runtime}; Bottom: $b=1$ the fitted exponent is 0.348 with 95\% confidence interval $[0.271, 0.426]$. 
    For the middle column \eqref{e:MaxBisection}, we use data with $n$ ranging from 16 to 22. Top: $b=0.7$ the fitted exponent is 0.444 with 95\% confidence interval $[0.436, 0.452]$; Bottom: $b=1$ the fitted exponent is 0.873 with 95\% confidence interval $[0.842, 0.905]$. 
    For the right column (MIS), we use data with $n$ ranging from 10 to 21. Top: $b=0.6$ the fitted exponent is 0.400 with 95\% confidence interval $[0.386, 0.415]$; Bottom: $b=0.8$ the fitted exponent is 0.392 with 95\% confidence interval $[0.122, 0.663]$. 
    }
    \label{fig:multiple_problem_runtime}
\end{figure}

\section*{Acknowledgements}
The authors thank  Aram Harrow, Shree Hari Sureshbabu, and Jiayu Shen
 for insightful discussions and feedback, and their colleagues at the Global Technology Applied Research center of JPMorganChase for their support and helpful discussions. 

\appendix

   \section{Technical Details for Generalized Short Path Framework}

\shortPathToGapLem*
\begin{proof}
 Recall by construction, the ground state energy of $H_b$ is at most $-1$.  Note that $\theta \leq 1$. Suppose, in order to arrive at a contradiction, that there are at least two orthogonal eigenstates $|\psi_1\rangle$ and $|\psi_2\rangle$, with energy strictly below $-1 + \theta$. Since the short path condition is satisfied, at least one of $|\psi_1\rangle$ or $|\psi_2\rangle$ has nonzero overlap with $|\pi\rangle$. Without loss of generality, assume this is the case for $|\psi_2\rangle$. 
 
Now consider the state
$$
    |\psi'\rangle = \left(1 +  \frac{\lvert\langle \sqrt{\pi}| \psi_1\rangle\rvert^2}{\lvert\langle \pi | \psi_2 \rangle\rvert^2}\right)^{-1/2}\left(|\psi_1\rangle - \frac{\langle \sqrt{\pi}| \psi_1\rangle}{\langle \sqrt{\pi} | \psi_2 \rangle}|\psi_2\rangle\right),
$$
which is orthogonal to $|\sqrt{\pi}\rangle$. Since $\GSE\left(\Pi_{\perp}H_b\Pi_{\perp}\right) \leq \langle \psi'|H_b|\psi'\rangle$, and 
\begin{align*}
    \langle \psi'|\Pi_{\perp}H_b\Pi_{\perp}|\psi'\rangle &= \langle \psi'|H_b|\psi'\rangle\\
    &<  \left(1 +  \frac{\lvert\langle \sqrt{\pi}| \psi_1\rangle\rvert^2}{\lvert\langle \sqrt{\pi} | \psi_2 \rangle\rvert^2}\right)^{-1}\left[(-1+\theta) + \frac{\lvert\langle \sqrt{\pi}| \psi_1\rangle\rvert^2}{\lvert\langle \sqrt{\pi} | \psi_2 \rangle\rvert^2}(-1+\theta) \right] \leq -1 + \theta,
\end{align*}
we get a contradiction.
\end{proof}

\shortPathImpliesOverlap*
\begin{proof}
Let $|\psi_b^{\perp}\rangle$ be the component of $\ket{\psi_b}$ orthogonal to $|\sqrt{\pi}\rangle$:
\begin{align*}
    |\sqrt{\pi}\rangle = \langle \sqrt{\pi} | {\psi_b}\rangle | \psi_b\rangle + \sqrt{1 - \langle \sqrt{\pi} | {\psi_b}\rangle^2}|\psi_b^{\perp}\rangle,
\end{align*}
where by stoquasticity of $H_b$, $\langle \sqrt{\pi} | \psi_b \rangle \geq 0$.

The short path condition implies that 
\begin{align}
\langle \psi_b^{\perp}|  H_b  |\psi_b^{\perp}\rangle \geq -1 + \theta,
\end{align}
and since $\langle \psi_b| H_b|\psi_b\rangle \leq -1$, we have 
\begin{align}
&\lvert \langle \psi_b^{\perp}| H_b|\psi_b^{\perp}\rangle - \langle \psi_b| H_b| \psi_b\rangle \rvert  \geq \theta.
\end{align}

Combining the above with $\lVert H_b\rVert_{2} \leq 2$ gives
\begin{align*}
\theta^2 \leq \lvert \langle \psi_b| H_b| \psi_b\rangle - \langle \psi_b^{\perp}| H_b|\psi_b^{\perp}\rangle\rvert^2 &\leq 4\lVert |\psi_b\rangle - |\psi_b^{\perp}\rangle\rVert^2\\
&=8(1- \langle \sqrt{\pi} | \psi_b^{\perp}\rangle)\\
&\leq 8(\langle \sqrt{\pi} | \psi_b\rangle)^2,
\end{align*}
so the result follows.
\end{proof}

\longJumpBound*
\begin{proof}

Define 
\begin{align*}
    \mathcal{P}_{\ell} := \left(\frac{H_b}{ E_b}\right)^{\ell}.
\end{align*}

Since $P$ is aperiodic, we have $-I\preceq -D(P) \prec I$. The negative definiteness of $g_{\eta}(H)$ implies the maximum eigenvalue of  $H_b$ cannot exceed that of $-D(P)$. Hence, since $\nu$ is a lower bound on either the log-Sobolev or Poincar\'e constant of $P$ (lower bounding its spectral gap  as defined in \eqref{eqn:spec_gap}), we have that $H_b$ does not have any eigenstates with eigenvalue $> 1- \frac{\nu}{2}$.
Additionally, the only eigenstate with energy below $-1 + \frac{\nu}{2}$ is the ground state $|\psi_b\rangle$, by assumption.

Thus,
\begin{align*}
\langle \sqrt{\pi}|\mathcal{P}_{\ell}|z\rangle &= \langle \sqrt{\pi}|\psi_b\rangle\langle\psi_b|z\rangle + \sum_{|E_b'| \leq 1 - \frac{\nu}{2}} \left(\frac{E'_b}{E_b} \right)^{\ell}\langle \sqrt{\pi}|\psi'_b\rangle\langle\psi'_b|z\rangle\\
&\leq \langle \sqrt{\pi}|\psi_b\rangle\langle\psi_b|z\rangle + (V -1)(1-\frac{\nu}{2})^{\ell},
\end{align*}
as all $E_b' < 0$.

The above implies
\begin{align*}
    \lvert E_b \rvert^{-\ell}\cdot \langle \sqrt{\pi}|(-H_b)^{\ell}|z\rangle - V(1-\frac{\nu}{2})^{\ell}< \langle \sqrt{\pi}|\psi_b\rangle\langle\psi_b|z\rangle < \lvert \langle\psi_b|z\rangle\rvert.
\end{align*}

We also have that
\begin{align*}
\lvert E_b\rvert \langle \psi_b|\sqrt{\pi}\rangle &= \langle \psi_b|D(P) - bG_{\eta} | \sqrt{\pi}\rangle \\
&= \langle \psi_b | \sqrt{\pi}\rangle - b\langle \psi_b | G_{\eta} |\sqrt{\pi}\rangle \\
&\leq \langle \psi_b | \sqrt{\pi}\rangle + b \sqrt{\mathbb{P}_{\pi}( E \leq (1-\eta)E^*)}.
\end{align*}
Since $\langle \psi_b | \sqrt{\pi}\rangle \geq 0$ by stoquasticity, we get
\begin{align*}
    \lvert E_b \rvert^{-\ell} \geq \exp\left(-\ell\frac{b\sqrt{\mathbb{P}_{\pi}( E \leq (1-\eta)E^*)}}{\langle \psi_b|\sqrt{\pi}\rangle}\right).
\end{align*}
If $\ell$ satisfies the conditions in the Lemma statement, then there exists an integer giving the stated result.
\end{proof}

\equivDepol*
\begin{proof}

The $\implies$ direction follows mostly from the proof of Proposition 3 in \cite{dalzell2022mind}
By definition of $g_{\eta}$, $f$ is monotonically non-decreasing. From   \cite[Proposition 12]{dalzell2022mind}, $\prod_{t=1}^{T}f(c_tx)$ is also a convex function. Thus
\begin{align*}
\sum_{y}P(x,y)\prod_{t=1}^{T}f\left(\frac{c_tH(y)}{E^{\star}}\right) &\geq \prod_{t=1}^{T}f\left(\sum_{y}P(x,y)\frac{c_tH(y)}{E^{\star}}\right)\\
&\geq \prod_{t=1}^{T}f\left(\frac{c_tH(x)}{E^{\star}}\left(1 + \frac{\Delta}{H(x)}\right)\right).
\end{align*}

Recall that $g_{\eta}(z)$ is zero for all $z \geq (1-\eta)E^{*}$. If for some $t$, $\frac{c_tH(x)}{E^{*}} > (1-\eta)E^{*}$, then the whole product is zero, as in this case
$$\frac{c_tH(x)}{E^{\star}}\left(1 + \frac{\Delta}{H(x)}\right) > (1-\eta)E^{*}.$$ 
Thus, the stated hypothesis is satisfied trivially. If 
$$\frac{H(x)}{E^{*}} < \frac{(1-\eta)E^{*}}{c_t} \leq (1-\eta)E^{*}$$ for all $t$, then 
$$\frac{1}{H(x)} > -\frac{1}{(1-\eta)\lvert E^{*}\rvert}.$$ 
By monotonicity of $f$
\begin{align*}
\prod_{t=1}^{T}f\left(\frac{c_tH(x)}{E^{\star}}\left(1 + \frac{\Delta}{H(x)}\right)\right) \geq \prod_{t=1}^{T}f\left(\frac{c_tH(x)}{E^{\star}}\left(1 - \frac{\Delta}{(1-\eta)\lvert E^*\rvert}\right)\right).
\end{align*}

For the $\Leftarrow$ direction, consider taking $c_{t} \rightarrow 0, t \neq 0$ and $c_1 \rightarrow 1$. By continuity we have
$$
\sum_{y}P(x,y)f\left(\frac{H(y)}{E^{\star}}\right) \geq f\left(\frac{(1-\alpha_{P})H(x)}{E^{\star}}\right)
$$
We can suppose  $H(x) < -(1-\eta)\lvert E^*\rvert$, since otherwise the right-hand side is zero, and so our choice of $\Delta$ in terms of $\alpha_P$ clearly works. Thus,
$$
\sum_{y}P(x,y)f\left(\frac{H(y)}{E^{\star}}\right)\geq f\left(\frac{H(x)+\alpha_P(1-\eta)\lvert E^{*}\rvert}{E^*}\right).
$$
so by definition of $f$
$$
\sum_{y}P(x,y)g_{\eta}\left(\lvert E^*\rvert^{-1}H(y)\right) \leq  g_{\eta}\left(\lvert E^*\rvert^{-1}(H(x)+\alpha_P(1-\eta)\lvert E^{*}\rvert)\right).
$$
\end{proof}

\deltaUpper*
\begin{proof}
Note that if have the stronger condition, for some $\tilde{\Delta}_{P} > 0$,
$$
\expP[H(y)] \leq  H_c(x)+\tilde{\Delta}_{P},  
$$
which may actually be easier to show, then we also have $\Delta_{P}(\eta)$ stability with $\Delta_{P}(\eta) \leq \tilde{\Delta}_{P}$, since by concavity of $h_{\eta} = g_{\eta}(\frac{x}{\lvert E^*\rvert})$ and Jensen's inequality:
$$
   \expP[h_{\eta}\left(\lvert E^*\rvert^{-1}H(y)\right)] 
 \leq h_{\eta}\left(\sum_{y}P(x,y)H(y)\right) \leq h_{\eta}( H(x)+\tilde{\Delta}_{P}).
$$
Also from the above, it is a simple consequence of Jensen's inequality that we can take $\Delta_P(\eta)$ to be the $\sqrt{\lVert H \rVert_{P}}$:
\begin{align*}
 \sqrt{\lVert \psi \rVert_{P}} &= \sqrt{\max_{x\in \Xcal}\sum_{y}P(y,x)(H(x) - H(y))^2} \\&\geq \max_{x\in \Xcal}\sum_{y}P(y,x)\lvert H(x) - H(y)\rvert \\&\geq \tilde{\Delta}_{P} \\&\geq \Delta_{P}(\eta),\quad \forall \eta \in [0, 1).
\end{align*}
\end{proof}

\section{Technical Details for MaxCut Hamming and MaxBisection}

\begin{lemma}
\label{lem:maxcut_mean_ener_bound}
Let $G(\Ncal, \Ecal)$ be drawn from the Erd\H{o}s-R\'enyi ensemble $\mathcal{G}\left(n, \frac{p}{n-1}\right)$ for a constant $p$. Consider the objective of \eqref{e:MaxCut}:
$$
     H(x) := -\frac{1}{2}\sum_{i < j} e_{ij}(1-x_ix_j).
$$
Then, with probability at least $1-\delta$ over the graph, it follows: %
$$
 \left\lvert-\mathbb{E}_{\pi}  H(x) - \frac{pk(n-k)}{n} \right\rvert \leq C \log(\delta^{-1}) \frac{2\sqrt{p}k(n-k)}{n\sqrt{2(n-1)}},
$$
where $C > 0$ is an arbitrary constant, and $\pi$ is the uniform distribution of Hamming-weight $k$ strings.
\end{lemma}
\begin{proof}
The shift required to ensure that the mean over in-constraint strings is zero is given by
$$
    -\mathbb{E}_{\pi} H(x) = \sum_{i < j} e_{ij}\mathbb{E}_{\pi} \frac{1}{2}(1-x_ix_j)  = \sum_{i < j} e_{ij}\Pr[x_i\neq x_j].
$$
Note that there are $n \choose k$ bitstrings of Hamming weight $k$. If $x_i \neq x_j$, we have the freedom to place $k-1$ ``$-1$''s in $n-2$ spots. Adding a factor of two since the same can be done for ``$+1$''s, we get
$$
    \Pr[x_i\neq x_j] = \frac{2{n-2 \choose k-1}}{{n \choose k}} = \frac{2k(n-k)}{n(n-1)}.
$$

Suppose edge creation probability is $\frac{p}{n}$. Since  $\sum_{i<j}e_{ij}$ a Binomial random variable $B \left(\binom{n}{2}, \frac{p}{n} \right)$, applying the Chernoff bound asserts that with probability at least $1 - \delta$, we have:
$$
   \left \lvert \sum_{i<j}e_{ij} - \binom{n}{2}\frac{p}{n} \right\rvert \leq  C \log(\delta^{-1})\sqrt{\binom{n}{2}\frac{p}{n} \left(1-\frac{p}{n} \right)},
$$
where $C > 0$ is an arbitrary constant.

Accordingly, with probability $1-\delta$, it follows:
\begin{align*}
\left\lvert \mathbb{E}_{\pi}H(x) - \binom{n}{2}\frac{2k(n-k)}{n(n-1)}\frac{p}{n} \right\rvert &\leq C\log(\delta^{-1}) \frac{2k(n-k)}{n(n-1)}\sqrt{\binom{n}{2}\frac{p}{n} \left(1-\frac{p}{n} \right)}\\
\implies \left\lvert \mathbb{E}_{\pi}H(x) - \frac{pk(n-k)}{n}\right\rvert &\leq C\log(\delta^{-1}) \frac{2\sqrt{p}k(n-k)}{n\sqrt{2(n-1)}}.
\end{align*}
\end{proof}
Thus, for $k=\frac{n}{2}$ one has
$$
\left\lvert \mathbb{E}_{\pi}H(x) - \frac{pn}{4} \right\rvert \leq C\log(\delta^{-1}) \frac{\sqrt{p}n}{2\sqrt{2(n-1)}} \implies
\left\lvert \frac{\mathbb{E}_{\pi}H(x)}{n} -\frac{p}{4}\right\rvert \leq C\log(\delta^{-1}) \frac{\sqrt{p}}{2\sqrt{2(n-1)}},
$$
from which one can conclude 
\begin{equation}\label{eqn:shift_max_bisec} 
    -\frac{\mathbb{E}_{\pi}H(x)}{n} = \frac{p}{4}(1 + o(1))
\end{equation}
with high probability. Thus for $k=n/2$, $\lvert E^{\star}\rvert = \Theta(n)$ so the cost function only gets shifted by constant to make it mean  zero. More generally, for $k=o(n)$, with high probability
$$
     -\mathbb{E}_{\pi}H(x) =  pk(1 + o(1)).
$$

\begin{lemma}
\label{lem:constr_max_delta}
    Let $H$ be the cost function of \eqref{e:MaxCut}:
    $$
        H(x) := -\sum_{i < j}e_{ij}\frac{(1-x_ix_j)}{2},
    $$
    and $P$ be the transition matrix for the transposition walk on the space Hamming-weight $k$ bitstrings. Then, for an arbitrary graph $G$ with $\lvert \mathcal{E} \rvert$ edges and Hamming-weight $k$ MaxCut $\mathcal{C}^{*}_{k}$, it follows:
    $$\frac{\mathcal{C}^{*}_{k}(n-2)}{k(n-k)}  \geq \mathbb{E}_{y \sim x}[H(y)] -  H(x)  \geq  \frac{\mathcal{C}^{*}_{k}(n-2)}{k(n-k)}  - \left(\lvert \mathcal{E}\rvert -  \mathcal{C}^{*}_{k}\right)\max \left\{ \frac{2}{k}, \frac{2}{n-k} \right\},$$
    for every $x \in \{ u \in \{-1,1\}^n : | u | = k\}$.
\end{lemma}
\begin{proof}
Fix a graph $G$, and consider an arbitrary $x \in \{-1, 1\}^n$ such that $\lvert x \rvert = k$ corresponding to number of $+1$'s. Swaps only occur between $x_i$ and $x_j$ that are different. All one step transitions denoted by $\sim$ are implied to be under the transposition walk.
\begin{align*}
H(x) &= -\sum_{i < j}e_{ij}\frac{(1-x_ix_j)}{2} = -\sum_{i < j~|~x_i \neq x_j}e_{ij},\\
\mathbb{E}_{y \sim x}[H(y)] &=  -\sum_{i < j}e_{ij}\mathbb{E}_{y\sim x} \left[\frac{1-y_iy_j}{2} \right],\\
\mathbb{E}_{y\sim x} \left[\frac{1-y_iy_j}{2} \right] &= \mathbb{P}[y_i \neq y_j | x].
\end{align*}

For a given fixed $x \in \{-1, 1\}^n$ satisfying $\lvert x \rvert = k$, we know there are $k(n-k)$ pairs of indices such that $x_i \neq x_j $, $\binom{n-k}{2}$  indices where $x_i = x_j = -1$, and  $\binom{k}{2}$ indices where $x_i = x_j = 1$. Thus we can decompose the expectation as follows:
\begin{equation}\label{e:Lem2.2_exp}
 \begin{aligned}
\mathbb{E}_{y \sim x}[H] =  &-\sum_{i <j~:~x_i  \neq x_j}e_{ij}\mathbb{P}_{y\sim x}[y_i \neq y_j | x_i \neq x_j] - \sum_{i <j~:~x_i = x_j =1}e_{ij}\mathbb{P}_{y\sim x}[y_i \neq y_j |  x_i = x_j = 1]  \\&-\sum_{i <j~:~x_i = x_j = -1}e_{ij}\mathbb{P}_{y\sim x}[y_i \neq y_j | x_i = x_j = -1].
\end{aligned}
\end{equation}

Using the following facts:
\begin{align*}
&\mathbb{P}_{y\sim x}[y_i \neq y_j | x_i \neq x_j] = \frac{(k-1)(n-k-1) + 1}{k(n-k)}\\
&\mathbb{P}_{y\sim x}[y_i \neq y_j | x_i = x_j = 1] =\frac{2}{k} \\
&\mathbb{P}_{y\sim x}[y_i \neq y_j | x_i = x_j = -1] =\frac{2}{n-k}.
\end{align*}
the expression \eqref{e:Lem2.2_exp} simplifies to 
\begin{align*}
\mathbb{E}_{y \sim x}[H] =  &-\sum_{i <j~:~x_i  \neq x_j}e_{ij}\frac{(k-1)(n-k-1) + 1}{k(n-k)} - \sum_{i <j~:~x_i = x_j=1}e_{ij}\frac{2}{k} -\sum_{i <j~:~x_i = x_j = -1}e_{ij}\frac{2}{n-k}.
\end{align*}
As a consequence, for all $x$ with Hamming weight $k$,
\begin{align*}
&\mathbb{E}_{y \sim x}[H] - H(x) \\ &= -\sum_{i <j~:~x_i  \neq x_j}e_{ij}\left(\frac{(k-1)(n-k-1) +1}{k(n-k)} - 1\right) -\sum_{i <j~:~x_i = x_j =1}e_{ij}\frac{2}{k} - \sum_{i <j~:~x_i = x_j = -1}e_{ij}\frac{2}{(n-k)}\\
&= \sum_{i <j~:~x_i  \neq x_j}e_{ij}\frac{n-2}{k(n-k)} -\sum_{i <j~:~x_i = x_j=1}e_{ij}\frac{2}{k} - \sum_{i <j~:~x_i = x_j = -1}e_{ij}\frac{2}{(n-k)}.
\end{align*} 
Thus for any $x \in \{-1,1\}^n$ and any graph $G$:
$$
    \mathbb{E}_{y \sim x}[H] \leq H(x) + \frac{-H(x)(n-2)}{k(n-k)},
$$
We have:
$$
  \forall x,  \mathbb{E}_{y \sim x}[H] \leq H(x) + \frac{\mathcal{C}^{*}_{k}(n-2)}{k(n-k)}.
$$
For the lower bound, for every $x \in \{u \in \{-1,1\}^n : | u| = k \}$, we have:
$$
\mathbb{E}_{y \sim x}[H(y)] -  H(x)  \geq  \frac{\mathcal{C}^{*}_{k}(n-2)}{k(n-k)}  - \left(\lvert \mathcal{E}\rvert -  \mathcal{C}^{*}_{k}\right)\max \left\{ \frac{2}{k}, \frac{2}{n-k} \right\},
$$
where $\lvert \mathcal{E}\rvert$ is the number of edges.
\end{proof}

\begin{lemma}
For the \ref{e:MaxCut} Hamiltonian $H$, the pseudo Lipschitz constant $\lVert H\rVert_{P}$ under the transposiiton walk is $\mathcal{O}(1)$ with high probability.
\end{lemma}
\begin{proof}
Recall
$$
   \lVert H(x) \rVert_{P} := \expP[(H(y) - H(x))^2],
$$
so we can consider
$$
\expP[H(y)^2] -  H(x)(2\expP[H(y)] - H(x)).
$$
The only new term is $\mathbb{E}_{y\sim x}[H(y)^2]$, which requires us to look at  terms like:
$$
    e_{ij}e_{rs}\expP \left[\frac{1-y_iy_j}{2}\frac{1-y_{r}y_{s}}{2} \right],
$$
$y$ is $x$ after a single random transposition. We can put these terms in groups based on:

\begin{enumerate}
    \item $x_j = x_i = x_r= x_s = \pm 1$, the term is always zero
    \item 
    $x_j = x_i = x_r \neq x_s =\pm 1$ \begin{align*}&-1 :\mathbb{E}_{y\sim x} \left[\frac{1-y_iy_j}{2}\frac{1-y_{r}y_{s}}{2}\right] =\frac{2(n-k-1)}{k(n-k)}\\
    &+1:\mathbb{E}_{y\sim x}\left[\frac{1-y_iy_j}{2}\frac{1-y_{r}y_{s}}{2}\right] =\frac{2(k-1)}{k(n-k)}\end{align*}
    \item $x_j = x_i \neq x_r= x_s = \pm 1$ the term is zero unless an element of  $\{i, j\}$ is swapped with an element of  $\{r, s\}$, giving \begin{align*}\pm 1:\mathbb{E}_{y\sim x}\left[\frac{1-y_iy_j}{2}\frac{1-y_{r}y_{s}}{2}\right] = \frac{4}{(n-k)k}\end{align*}
    \item $x_j \neq x_i = x_r \neq x_s = \pm 1$,
    $$
        \mathbb{E}_{y\sim x} \left[\frac{1-y_iy_j}{2}\frac{1-y_{r}y_{s}}{2} \right] = \frac{(n-k-2)(k-2) + 2}{k(n-k)}= \frac{k(n-k) - 2(n-3)}{k(n-k)}
    $$
\end{enumerate}

Let $\hat{y}_{ij} = \frac{1- y_i y_j}{2}$, then expanding:
\begin{align*}
H(y) &=  -\sum_{i <j~:~x_i  \neq x_j}e_{ij}\hat{y}_{ij}- \sum_{i <j~:~x_i = x_j =1}e_{ij}\hat{y}_{ij} -\sum_{i <j~:~x_i = x_j = -1}e_{ij}\hat{y}_{ij}\\
\mathbb{E}_{y\sim x}[H(y)^2] &=\mathbb{E}_{y\sim x}\left\{\left(\sum_{i <j~:~x_i  \neq x_j, }e_{ij}\hat{y}_{ij}\right)^2\right\} \\&+ 2\mathbb{E}_{y\sim x}\left\{\left(\sum_{i <j~:~x_i = x_j =1}e_{ij}\hat{y}_{ij}\right)\left(\sum_{r<s~:~x_r \neq x_s}e_{ij}\hat{y}_{ij}\right)\right\} \nonumber \\
& + 2\mathbb{E}_{y\sim x}\left\{\left(\sum_{i <j~:~x_i = x_j =-1}e_{ij}\hat{y}_{ij}\right)\left(\sum_{r<s~:~x_r \neq x_s}e_{ij}\hat{y}_{ij}\right)\right\}\\&+2\mathbb{E}_{y\sim x}\left\{\left(\sum_{i <j~:~x_i = x_j =-1}e_{ij}\hat{y}_{ij}\right)\left(\sum_{r<s~:~x_r = x_s =1}e_{ij}\hat{y}_{ij}\right)\right\}, 
\end{align*}
where we have eliminated the square of the $x_i=x_j=\pm 1$ terms since they fall into group $1$. Expanding further and passing the expectations through:
\begin{align*}
\mathbb{E}_{y\sim x}[H(y)^2] &= \sum_{i<j~:~x_i \neq x_j} e_{ij}\mathbb{E}_{y\sim x}\hat{y}_{ij}\\ 
&\quad+2\sum_{i <j, r<s, (i,j)\neq (r,s)~:~x_i  \neq x_j,~x_r \neq x_s }e_{ij}e_{rs}\mathbb{E}_{y\sim x}[\hat{y}_{ij}\hat{y}_{rs}] \\
&\quad+ 2\sum_{i <j, r <s~:~x_i = x_j =1, x_r \neq x_s}e_{ij}e_{rs}\mathbb{E}_{y\sim x}[\hat{y}_{ij}\hat{y}_{rs}] \nonumber \\
&\quad+ 2\sum_{i <j, r <s~:~x_i = x_j =-1, x_r \neq x_s}e_{ij}e_{rs}\mathbb{E}_{y\sim x}[\hat{y}_{ij}\hat{y}_{rs}] \nonumber \\
&\quad+ 2\sum_{i <j, r <s~:~x_i = x_j =1, x_r = x_s = -1}e_{ij}e_{rs}\mathbb{E}_{y\sim x}[\hat{y}_{ij}\hat{y}_{rs}] \nonumber.
\end{align*}

Next we compute the expectations by identifying which group they belong to:
\begin{align*}
\mathbb{E}_{y\sim x}[H(y)^2] &= \frac{k(n-k)-(n-2)}{k(n-k)}\sum_{i<j~:~x_i \neq x_j} e_{ij}\\ 
&+\frac{2k(n-k) - 4(n-3)}{k(n-k)}\sum_{i <j, r<s, (i,j)\neq (r,s)~:~x_i  \neq x_j,~x_r \neq x_s }e_{ij}e_{rs} \\&+ \frac{4(n-k-1)}{k(n-k)}\sum_{i <j, r <s~:~x_i = x_j =1, x_r \neq x_s}e_{ij}e_{rs} \nonumber \\
& + \frac{4(k-1)}{k(n-k)}\sum_{i <j, r <s~:~x_i = x_j =-1, x_r \neq x_s}e_{ij}e_{rs} \nonumber \\
&+ \frac{8}{(n-k)k}\sum_{i <j, r <s~:~x_i = x_j =1, x_r = x_s = -1}e_{ij}e_{rs}. \nonumber
\end{align*}
The only component to clarify is the first sum, which follows from a calculation in the previous lemma
\begin{align}
\expP[\hat{y} | x_i \neq x_j] = \frac{(k-1)(n-k-1) + 1}{k(n-k)} = \frac{k(n-k)-(n-2)}{k(n-k)}.
\end{align}

For the $2\mathbb{E}_{y \sim x}[H] - H(x) $ part, we use the same computations from the previous lemma, where we computed $\mathbb{E}_{y \sim x}[H] - H(x)$:

\begin{align*}
&2\mathbb{E}_{y \sim x}[H] - H(x) \\ &= -\sum_{i <j~:~x_i  \neq x_j}e_{ij}\left(\frac{2(k-1)(n-k-1) +2 }{k(n-k)} - 1\right) -\sum_{i <j~:~x_i = x_j =1}e_{ij}\frac{4}{k} - \sum_{i <j~:~x_i = x_j = -1}e_{ij}\frac{4}{(n-k)}\\
&= \sum_{i <j~:~x_i  \neq x_j}e_{ij}\frac{-k(n-k)+2(n-2)}{k(n-k)} -\sum_{i <j~:~x_i = x_j=1}e_{ij}\frac{4}{k} - \sum_{i <j~:~x_i = x_j = -1}e_{ij}\frac{4}{(n-k)},
\end{align*}
where recall that
\begin{align*}
&\expP[\hat{y}| x_i = x_j = 1] =\frac{2}{k} \\
&\expP[\hat{y}| x_i = x_j = -1] =\frac{2}{n-k}.
\end{align*}
Then we compute $-H(x)[2\mathbb{E}_{y \sim x}[H] - H(x)]$:
\begin{align*}
&-H(x)[2\mathbb{E}_{y \sim x}[H] - H(x)] \\&=\sum_{i <j:~x_i  \neq x_j}e_{ij}\frac{-k(n-k)+2(n-2)}{k(n-k)}\\&+\sum_{i <j, r < s, (i,j)\neq(r,s):~x_i  \neq x_j,x_r \neq x_s}e_{rs}e_{ij}\frac{-2k(n-k)+4(n-2)}{k(n-k)} \nonumber\\ &-\sum_{i <j, r < s~:~x_i = x_j=1, x_r \neq x_s}e_{rs}e_{ij}\frac{4}{k} \\&- \sum_{i <j, r < s:~x_i = x_j = -1, x_r \neq x_s}e_{rs}e_{ij}\frac{4}{(n-k)}.
\end{align*}

We can now put all the expressions together to get:
\begin{align}
\label{eqn:psuedo_lip_bound}
\mathbb{E}_{y\sim x}[H(y)^2] -  H(x)(2\mathbb{E}_{y\sim x}[H(y)] - H(x)) &= \frac{(n-2)}{k(n-k)}\sum_{i<j~:~x_i \neq x_j} e_{ij}\\ 
&\quad+\frac{4}{k(n-k)}\sum_{i <j, r<s, (i,j)\neq (r,s)~:~x_i  \neq x_j,~x_r \neq x_s }e_{ij}e_{rs} \\
&\quad- \frac{4}{k(n-k)}\sum_{i <j, r <s~:~x_i = x_j =1, x_r \neq x_s}e_{ij}e_{rs} \nonumber \\
&\quad- \frac{4}{k(n-k)}\sum_{i <j, r <s~:~x_i = x_j =-1, x_r \neq x_s}e_{ij}e_{rs} \nonumber \\
&\quad+ \frac{8}{k(n-k)}\sum_{i <j, r <s~:~x_i = x_j =1, x_r = x_s = -1}e_{ij}e_{rs}. 
\end{align}

Let $X$ be a Binomial random variable $\mathcal{B}(M, q)$, then
\begin{align}
\label{eqn:sqr_cher}
\mathbb{P}[ X^2 \geq (1 +\delta)^2\mathbb{E}[X^2]] \leq e^{-\frac{\delta^2Mq}{2+\delta}}.
\end{align}
This follows simply from Jensen's inequality
\begin{align*}
\mathbb{P}[ X^2 \geq (1 +\delta)^2\mathbb{E}[X^2]] &\leq  \mathbb{P}[ X^2 \geq (1 +\delta)^2(\mathbb{E}[X])^2]\\
&= \mathbb{P}[ X \geq (1 +\delta)\mathbb{E}[X]]\\
&\leq e^{-\frac{\delta^2\mathbb{E}[X]}{2+\delta}},
\end{align*}
where the last inequality is the multiplicative Chernoff bound. Also for $X$ and $Y$ being independent Binomials we have that:
\begin{align*}
 \mathbb{P}[XY \geq (1+\delta)\mathbb{E}[XY]]  \leq 2e^{-\frac{\delta^2\min(\E[X], \E[Y])}{2+\delta}},
\end{align*}
which follows from
\begin{align*}
\mathbb{P}[XY \geq (1+\delta)\mathbb{E}[XY]] \le \mathbb{P}[X \geq (1 + \delta/2)\E[X] \vee Y \geq (1 + \delta/2)\E[Y]],
\end{align*}
which follows from 
\begin{align*}
(X \geq (1 + \delta/3)\E[X]) \wedge (Y \geq (1 + \delta/3)\E[Y]) &\implies  XY \leq (1 + \delta)\E[X]\E[Y] \\&= XY \leq (1 + \delta)\E[XY],
\end{align*}
so union bound gives the desired result.

Thus, we can assume with high probability all of the sums of Bernoullis in Equation \eqref{eqn:psuedo_lip_bound} are constant factors from their means. Their means are
\begin{enumerate}
    \item $\E[\sum_{i<j~:~x_i \neq x_j} e_{ij}] \frac{pk(n-k)}{n-1} \asymp pk$
    \item $\E[\sum_{i <j, r<s, (i,j)\neq (r,s)~:~x_i  \neq x_j,~x_r \neq x_s }e_{ij}e_{rs}] =[\frac{p}{n-1}k(n-k)]^2 \asymp p^2k^2$
    \item  $\E[\sum_{i <j, r <s~:~x_i = x_j =1, x_r \neq x_s}e_{ij}e_{rs}] =[\frac{p}{n-1}]^2\binom{k}{2}k(n-k) \asymp p^2\frac{k^3}{n}$
    \item $\E[\sum_{i <j, r <s~:~x_i = x_j =-1, x_r \neq x_s}e_{ij}e_{rs}] = [\frac{p}{n-1}]^2\binom{n-k}{2}k(n-k) \asymp p^2kn$
    \item $\E[\sum_{i <j, r <s~:~x_i = x_j =1, x_r = x_s = -1}e_{ij}e_{rs}] = [\frac{p}{n-1}]^2\binom{k}{2}\binom{n-k}{2} \asymp p^2k^2$.
\end{enumerate}

Plugging the above asymptotics in for the Binomials in Equation \eqref{eqn:psuedo_lip_bound} suffices to obtain $\lVert H \rVert_{P} = \Ocal(1)$.

\end{proof}

\begin{lemma}
\label{lem:loose_maxcut_ham_bound}
With high probability, the optimal objective value of \eqref{e:MaxCut} satisfies:
\begin{align}
    \mathcal{C}^*_k = \begin{cases}
          o(k\log(n)) &\text{if}~k = o(n),\\
           \Ocal(n) &\text{if}~k = \Theta(n),
    \end{cases}
\end{align}
where $k$ is the Hamming weight.
\end{lemma}

\begin{proof}
We can use the indicator trick to try to bound  the probability of a cut set of a given size. Let $z \in \{0, 1\}^n$, $\lvert z \rvert = k$ and $ I_{z,m}$ be a Bernoulli random variable indicating whether there are $m$ edges cut with assignment $z$ over the random choice of graph. Then 
$$X_m := \sum_{z \in \{0, 1\}^n, \lvert z\rvert =k} I_{z,m}$$
is the number of in-constraint cuts of size $m$.  Note that $I_{z,m}$ are not independent. For any graph $G$ drawn from $\mathcal{G}(n, p/n)$, we have
$$\mathbb{P}[I_{z,m} = 1] = \binom{k(n-k)}{m}(p/n)^{m}(1-p/n)^{k(n-k) - m}.$$ The first moment method gives:
$$
    \mathbb{P}[X_m > 0] \leq \mathbb{E}[X_m] \leq \binom{n}{k}\binom{k(n-k)}{m}n^{-m}.
$$

Suppose that $k = o(n) \cap \omega(1)$ and $m = \Ocal(n)$. Then $k(n-k) = \Ocal(nk)$, and we can apply the following asymptotics:
$$
    \mathbb{P}[X_m > 0] \leq \binom{n}{k}\binom{k(n-k)}{m}n^{-m} \asymp \left(\frac{n}{k}\right)^{k}\left(\frac{nk}{m}\right)^m n^{-m} = n^k k^{m-k}m^{-m}.
$$

The goal is to try to identify the phase transition point at which the probability goes to zero asymptotically. We can look at 
$$
    k\log(n) + m\log(k) - k\log(k) - m\log(m).
$$
Upon taking $ m = \Theta(k \cdot r)$, we obtain $$  k\log(n) + m\log(k) - k\log(k) - m\log(m) = k\log \left( \frac{n}{kr^r} \right).$$  Transition is at $\log(n/k) = r\log(r)$. For $r = \log(n)$, $ \mathbb{P}[X_m > 0] \rightarrow 0$, thus $C^*_k =o(k\log(n))$.

Suppose $k = \Theta(n)$,  then $\log\binom{n}{k} \asymp \mathcal{H}(\frac{k}{n})n$, where $\mathcal{H}$ is the binary entropy function. Thus
$$
    \mathbb{P}[X_m > 0] \leq \binom{n}{k}\binom{k(n-k)}{m}n^{-m} \asymp 2^{\mathcal{H}(n/k)n}\left(\frac{nk}{m}\right)^m n^{-m} = 2^{\mathcal{H}(n/k)n}k^{m}m^{-m}.
$$
We can look at 
$$
   \mathcal{H}(n/k)n - m\log(k/m),
$$
so  transition is at $m = 2^{\mathcal{H}(n/k)}k = \Theta(n)$. Thus for $k=\Theta(n)$, $\mathcal{C}^{*}_k = \Ocal(n)$.
\end{proof}

\section{Constrained Short Path via Penalized Objective}
\label{sec:penalized_obj}
Suppose we want to solve the following constrained problem
$$
    \min_{x \in \mathcal{X} \subseteq \{-1, 1\}^n} H(x).
$$
Suppose we also have a CSP 
\begin{align*}
\mathcal{C}(x) = \sum_{\ell=1}^mC_{\ell}(x),
\end{align*}
with $m$ constraints $C_{\ell}$, indicating the in-constraint solutions. Suppose $C_{\ell}$ has $k_{\ell}$  literals and has $s_{\ell}$ satisfying assignments, then
\begin{align*}
    \mathcal{C}_{\ell}(x) = 
    \begin{cases}
        -\frac{1}{s_{\ell}} \quad \mbox{if x satisfies constraint $\ell$} \\
        \frac{1}{2^{k_{\ell}}-s_{\ell}} \quad \mbox{otherwise}.
    \end{cases}
\end{align*}

Consider the task of minimizing the penalized Hamiltonian
\begin{align}
\label{eqn:penalty_ham}
    \tilde{H}(x) =\frac{H(x)}{\lVert H \rVert_{2}} + \mathcal{C}(x).
\end{align}
The Markov chain $\Mcal = (\Xcal, P, \pi)$ is the random walk on the hypercube, and hence $\pi$ is the uniform distribution over $\{-1, 1\}^n$. Thus, the short path Hamiltonian is the same as \cite{dalzell2022mind} but with the penalized cost Hamiltonian.

The global minimum of Equation \eqref{eqn:penalty_ham} is in fact the in-constraint minimum and $\mathbb{E}_{\pi}[\tilde{H}(x)] = \mathbb{E}_{\pi}[\frac{H(x)}{\lVert H\rVert_{2}}]$. Also note that $\lvert E^*\rvert = \Theta(m)$ for $\tilde{H}$. The below results will show that due to the normalization, the properties of the CSP are the only components that matter for determining the runtime. We denote $k = \max k_{\ell}$.

\begin{lemma}
Let $\Mcal = (\Xcal, P, \pi)$ be the random walk on the Hamming hypercube. The penalized Hamiltonian is $\Delta_P$ stable with
\begin{align}
    \Delta_{P}(\eta) = \Ocal\left(\frac{mk}{n}\right).
\end{align}
\end{lemma}
\begin{proof}
\begin{align}
\Delta_{P}(\eta) \leq \frac{\mathbb{E}_{y\sim}[H(y) - H(x)]}{\lVert H\rVert_{2}} + \frac{mk}{n}  = \Ocal\left(\frac{mk}{n}\right).
\end{align}
\end{proof}

\begin{lemma}[Tail bound condition for Penalized Hamiltonian]
Suppose $\mathbb{E}_{\pi}[\tilde{H_c}] = 0$,
$$
\pi((1-\eta)E^{\star})\leq 2^{-\gamma}, \quad \gamma = \Omega\left[\frac{2(1-\eta)^2}{\ell^2}\right].
$$
\end{lemma}
\begin{proof}

    Suppose flipping one variable $x_i$ changes at most $\Delta_i$ in the value of $H$ while holding other variables constant. Then 
    $$
        \Delta_i  = \mathcal{O}\left(\frac{mk}{n}\right).
    $$

    Assuming $\E_{\pi}[H(x)]=0$. Let  $C(E)$ denote the number of $x$ with energy $\leq E$ under $\tilde{H}$. For $E<0$, $C(E)/2^n$ is the probability for a random $x$, $H(x)$ deviates below $\E_{\pi}[H(x)]=0$ by amount at least $|E|$. The $P$-pseudo Lipschitz norm of $\tilde{H}$ is bounded by $$\lVert \tilde{H}\rVert_{P}=\sum_{i=1}^{m}\frac{\Delta_i^2}{n} = \mathcal{O}\left(\frac{m^2k^2}{n^2}\right).$$ Also, $\omega = \frac{2}{n}$ for the hypercube walk, and thus the Herbst argument implies
    $$
        \pi(E) \leq e^{-\Omega\left(\frac{2nE^2}{(mk)^2}\right)}.
    $$
    Taking $E=E^{\star}(1-\eta)$, we have
     $$
    \pi((1-\eta)E^{\star})\leq 2^{-n\gamma}, \quad \gamma = \Omega\left[\left(\frac{\abs{E^{\star}}}{m}\right)^2\frac{2(1-\eta)^2}{k^2}\right].
     $$
     However $\lvert E^{\star}\rvert = \Theta(m)$.
    
\end{proof}

The following result is then evident from Theorem \ref{thrm:MTG_RT}, given that from \cite{dalzell2022mind} we have that $b^*$ is constant for the random walk on the Hamming cube if $\gamma$ is constant.
\begin{theorem}
Let $\Mcal = (\Xcal, P, \pi)$ be random walk on the $n$-bit Hamming cube. Let $H: \pmone^n \rightarrow \R$ be a diagonal Hamiltonian with ground state energy 
$$ E^{\star} := \min_{x \in \Xcal} H(x).$$ 
If the number of constraints is $m = \Theta(n)$, then there exists a short path algorithm with runtime
$$
\Ostar\left(2^{n\left(\frac{1}{2}-\frac{(1-\eta)\lvert E^{\star}\rvert b}{n2\ln(2)\Delta_P}\right)}\right).
$$
\end{theorem}

The penalty method ensures that our algorithm only outputs feasible solutions, but the framework only provides a speedup over unconstrained brute-force search. Hence,  the generalized short path algorithm is significantly more effective.

\section{Details about Mixer Implementation}
\label{sec:block-encoding}
\subsection{Block-encoding the Glauber Mixer}
Since the Glauber dynamics mixer is a symmetric sparse matrix, we can implement the block-encoding in polynomial time by assuming sparse oracle access to the non-zero entries of $P$. Let $s$ be the maximum number of non-zero entries in any row of $D$. Then, we can implement the following oracle,
$$
    O_{S} \ket{x}\ket{0} \mapsto \frac{1}{\sqrt{s}}\sum_{y} \ket{x}\ket{y}
$$
where $y$ is an index of a non-zero entry in $P(x,\cdot)$. Since Glauber dynamics can only update one site at most, the transition matrix contains at most $n$ entries. Hence, we can implement this oracle by computing $P$ at most $n$ times. Define the oracle for access to the elements of $P$ in the following way,
$$
    O_{A} \ket{x}\ket{y}\ket{0} \mapsto \ket{x}\ket{y}\left(\sqrt{P(x,y)}\ket{0} +\sqrt{1-P(x,y)}\ket{1}\right).
$$
Implementing this oracle takes at most $\Ocal(n)$. Finally, we define \textsf{SWAP} operator,
$$
    \SWAP \ket{a} \ket{x}\ket{y}\ket{b} \mapsto \ket{y}\ket{b}\ket{a} \ket{x}.
$$
Then, the circuit $O_{S}^{\dagger} O_{A}^{\dagger} \SWAP O_{A} O_{S}$ implements block-encoding of $D/s$. To see this, compute
$$
  \bra{0}\bra{0}\bra{0}\bra{y} O_{S}^{\dagger} O_{A}^{\dagger} \SWAP O_{A} O_{S} \ket{0}\ket{0}\ket{0}\ket{x}.
$$
One has
\begin{align*}
  \ket{0}\ket{0}\ket{0}\ket{x}& \xrightarrow{O_S}  \frac{1}{\sqrt{s}}\sum_y \ket{0}\ket{y} \ket{0}\ket{x}\\
  &\xrightarrow{O_A}  \frac{1}{\sqrt{s}}\sum_y \ket{0}\ket{y} (\sqrt{P(x,y)}\ket{0}+\sqrt{1-P(x,y)}\ket{1})\ket{x}\\
  &\xrightarrow{\SWAP}  \frac{1}{\sqrt{s}}\sum_y (\sqrt{P(x,y)}\ket{0}+\sqrt{1-P(x,y)}\ket{1})\ket{x}\ket{0}\ket{y}.
\end{align*}
On the other hand,
\begin{align*}
    \bra{0}\bra{0}\bra{0}\bra{y} O_{S}^{\dagger} O_{A}^{\dagger}= \frac{1}{\sqrt{s}}\sum_z \bra{0}\bra{z} (\sqrt{P(y,z)}\ket{0}+\sqrt{1-P(y,z)}\ket{1})\ket{y}.
\end{align*}
Therefore,
$$
      \bra{0}\bra{0}\bra{0}\bra{y} O_{S}^{\dagger} O_{A}^{\dagger} \SWAP O_{A} O_{S} \ket{0}\ket{0}\ket{0}\ket{x} =\frac{1}{s}\sqrt{P(x,y)P(y,x)}.
$$
Hence, we can implement the block-encoding in at most polynomial time.
\subsection{Ground State Preparation for Glauber Mixer}
We discuss how to prepare the ground state of discriminant operator for Glauber dynamics for hard-core model and Ising model at $\lambda\leq \lambda_c$ and $\beta_c\leq \beta$ respectively. This is all we can afford to prepare since the spectral gap of $D$ falls exponentially fast when $\lambda>\lambda_c$ ($\beta>\beta_c$) due to statistical phase transitions. For simplicity, we consider the hard-core model to explain the idea. However, the details for both models will be given as seperate prpositions below.
Let $D_{\lambda}$ be the discriminant matrix of Glauber dynamics at fugacity $\lambda$. We use the block-encoding of $D_{\lambda}$ and amplitude amplification to prepare its ground state. Although, we can efficiently create block-encoding for $D_{\lambda}$ and apply singular value transformations to build a projector to its ground state, amplitude amplification might need exponentially many calls to this projector when we do not have a warm initial state.  
Instead, we combine classical annealing with quantum singular value transformation to prepare the ground state of $D_{\lambda}$ to create sequence of states where each state is warm with respect to the next one. Suppose that we prepare $\pi^{(1)}$, the coherent quantum state corresponding to the ground state of $D_{\lambda_1}$ where $0<\lambda_1<\varepsilon$,
$$
    \ket{\sqrt{\pi^{(1)}}} =\sum_{x\in \mathcal{X}} \sqrt{\pi_{\lambda_1}(x)}\ket{x}
$$
where the sum is over all independent sets $x$ in $G$. Next, we increase fugacity to $\lambda_2 = \lambda_1(1+\Delta)$ and prepare,
$$
    \ket{\sqrt{\pi^{(2)}}} =\sum_{x\in \mathcal{X}} \sqrt{\pi_{\lambda_2}(x)}\ket{x}.
$$
This quantum state can be prepared by applying ground state projector of $D_{\lambda_1}$ to and $D_{\lambda_2}$ to $\ket{\sqrt{\pi^{(0)}}}$  through fixed point amplitude amplification. We repeat this process until we prepare $\ket{\sqrt{\pi^{(k)}}}$,
$$
    \ket{\sqrt{\pi^{(k)}}} =\sum_{x\in \mathcal{X}} \sqrt{\pi_{\lambda_k}(x)}\ket{x}
$$
with $\lambda_k = \lambda_c$. We need to show that this process can be done in $\poly(n)$ time.
\begin{proposition}[Preparation of Gibbs State for Hard-core Model]
    Consider Glauber dynamics chain for hard-core model on graph $G(\mathcal{N},\mathcal{E})$ with transition matrix $P_{\lambda}$ with stationary distribution,
    \begin{equation}
        \pi_{\lambda}(x) = \frac{\lambda^{|x|}}{Z}.
    \end{equation}
Let $\delta(\lambda)>0$ be the spectral gap of the discriminant matrix $D(P_{\lambda})$ associated with $P_\lambda$. Then, there exists a quantum algorithm that prepares the  ground state of $-D$ up to $\varepsilon$ accuracy in $\ell_2$ norm with run time $\widetilde{\Ocal}(\delta_{\textup{min}}^{-1/2}\log(1/\varepsilon))$ where $\delta_{\textup{min}} = \inf_{0\leq\lambda'\leq \lambda}\delta(\lambda')$.
\end{proposition}
\begin{proof}
We start with the following quantum state,
$$
    \ket{\sqrt{\pi^{(0)}}} = \frac{1}{\sqrt{n}}\sum_{i=1}^n \ket{x_i}
$$
where $x_i$ is all $0$ bit string except location $i$. This quantum state is ground state of $D$ at $\lambda = 0$ as each $x_i$ is an independent set and they are the only ones with non-zero probability. This quantum state is essentially superposition over all strings with Hemming weight 1. This state can be prepared by applying $\frac{1}{\sqrt{n}}\sum_{i=1}^n X_i$ where $X_i$ is Pauli-$X$ applied to the all $0$ state. Therefore, this state can be prepared in polynomial time. However, we cannot implement the Glauber dynamics at $\lambda=0$ as the transition density is not meaningful. Instead, we start with $\ket{\sqrt{\pi^{(1)}}}$ which is the ground state of $D_{\lambda_1}$ which can be prepared from $\ket{\sqrt{\pi^{(0)}}}$ since $\ket{\sqrt{\pi^{(0)}}}$ and $\ket{\sqrt{\pi^{(1)}}}$ overlaps significantly
\begin{align*}
    |\braket{\sqrt{\pi^{(0)}}|\sqrt{\pi^{(1)}}}| = \frac{1}{\sqrt{n}}\sum_{i=1}^n\frac{\lambda_{1}^{1/2}}{\sqrt{Z_{1}}}
    \geq \frac{\sqrt{n\lambda_1}}{\sqrt{n\lambda_1 + \sum\limits_{|\mathcal{I}|>1}\lambda_1^{\mathcal{|I|}}}}
    &= \Omega(1),
\end{align*}
for $\lambda_1 \leq \frac{3n}{2^n}$. Similarly, given $\ket{\sqrt{\pi^{(k-1)}}}$, we can prepare $\ket{\sqrt{\pi^{(k)}}}$ efficiently. The number of calls to the ground state projector by the fixed point amplitude amplification from $\ket{\sqrt{\pi^{(k-1)}}}$ to $\ket{\sqrt{\pi^{(k)}}}$ is $\widetilde{\Ocal}(|\langle \sqrt{\pi^{(k-1)}}|\sqrt{\pi^{(k)}}\rangle|^{-1})$. The overlap can be calculated as 
\begin{align*}
    \left|\!\braket{\sqrt{\pi}^{(k-1)}|\pi^{(k)}}\!\right| &= \sum_{x\in \mathcal{X}}\frac{\lambda_{k-1}^{|x|/2}}{\sqrt{Z_{k-1}}}\frac{\lambda_{k}^{|x|/2}}{\sqrt{Z_{k}}}\\
    &\geq \sum_{x\in \mathcal{X}}\frac{\lambda_{k-1}^{|x|/2}\lambda_{k}^{|x|/2}}{Z_{k}}\\
    &= \sum_{x\in \mathcal{X}}\frac{\lambda_{k}^{|x|}\Delta^{-|x|/2}}{Z_{k}}\\
    &\geq  (1+\Delta)^{\frac{-n}{2}}\\
    &\geq 1-\frac{n\Delta}{2}\\
    &=\Omega(1),
\end{align*}
if we set $\Delta\leq \frac{2}{n}$. Hence, starting from $\ket{\sqrt{\pi^{(1)}}}$ we can prepare a schedule that maintains constant overlap with the subsequent state. Since we have constant overlap throughout the schedule, each amplitude amplification step only requires constant number of calls to ground state projectors. Also note that implementing the ground state projector requires $\widetilde{\mathcal{O}}(\delta^{-1})$ calls to block-encoding of $D$ \cite{gilyen2019QSVT}. Finally, we only need to do $\Ocal(\poly(n))$ rounds of amplitude amplification since $\lambda_{k} = \lambda_1 \exp(2k/n)$ and for $k\geq \frac{n}{2}\log(\lambda_c/\lambda_1)$, $\lambda_k \geq \lambda_c$.

\end{proof}

\begin{proposition}[Preparation of Gibbs State for Ising Model Model]
     Consider Glauber dynamics chain for Ising model on graph $G(\mathcal{N},\mathcal{E})$ with transition matrix $P_{\beta}$ with stationary distribution,
    \begin{equation}
        \pi_{\beta}(x) = \frac{\exp(-\beta H(x))}{Z}.
    \end{equation}
Let $\delta(\beta)>0$ be the spectral gap of the discriminant matrix $D(P_{\beta})$ associated with $P_\beta$. Then, there exists a quantum algorithm that prepares the  ground state of $-D$ up to $\varepsilon$ accuracy in $\ell_2$-norm with run time $\widetilde{\Ocal}(\delta_{\textup{min}}^{-1/2}\log(1/\varepsilon))$ where $\delta_{\textup{min}} = \inf_{0\leq\beta'\leq \beta}\delta(\beta')$.
\end{proposition}

\begin{proof}
The proof is similar to the hard-core model, and we start with the following quantum state,
$$
    \ket{\sqrt{\pi^{(0)}}} = \frac{1}{\sqrt{2^n}}\sum_{x \in \{0, 1\}^n} \ket{x}
$$
This quantum state is ground state of $D$ at $\beta = 0$ since for $\beta = 0$, Glauber dynamics is equivalent to hypercube walk. Similar to MIS case, given $\ket{\sqrt{\pi^{(k-1)}}}$, we can prepare $\ket{\sqrt{\pi^{(k)}}}$ up to $\varepsilon$ accuracy efficiently. The number of calls to the ground state projector by the fixed point amplitude amplification from $\ket{\sqrt{\pi^{(k-1)}}}$ to $\ket{\sqrt{\pi^{(k)}}}$ is $\widetilde{\Ocal}(|\langle \sqrt{\pi^{(k-1)}}|\sqrt{\pi^{(k)}}\rangle|^{-1})$. The overlap can be calculated as 
\begin{align*}
    \left|\!\braket{\sqrt{\pi^{(k-1)}}|\sqrt{\pi^{(k)}}}\!\right| &= \sum_{x\in G}\frac{\exp(-\beta_{k-1} H(x)/2)}{\sqrt{Z_{k-1}}}\frac{\exp(-\beta_{k} H(x)/2)}{\sqrt{Z_{k}}}\\
    &\geq  \sum_{x\in G}\frac{\exp(-\beta_{k-1} H(x)/2)\exp(-\beta_{k} H(x)/2)}{Z_k}\\
    &=  \sum_{x\in G}\frac{\exp(-\beta_{k} H(x))\exp((\beta_{k}-\beta_{k-1}) H(x)/2)}{Z_k}\\
    &\geq \exp((\beta_k-\beta_{k-1})\|H\|/2)\\
    &=\Omega(1),
\end{align*}
if we set $(\beta_k-\beta_{k-1})=\frac{1}{\|H\|}$. Hence, starting from $\ket{\sqrt{\pi^{(0)}}}$ we can prepare a schedule that maintains constant overlap with the subsequent state.

Since we have constant overlap throughout the schedule, each amplitude amplification step only requires constant number of calls to ground state projectors. Also note that implementing the ground state projector requires $\widetilde{\mathcal{O}}(\delta^{-1})$ calls to block-encoding of $D$ \cite{gilyen2019QSVT}. Finally, we only need to do $\poly(n)$ rounds of amplitude amplification since $\|H\| = \poly(n)$ and length of the schedule is polynomial. 
\end{proof}

\clearpage
\bibliography{bibo}

\section*{Disclaimer}
This paper was prepared for informational purposes by the Global Technology Applied Research center of JPMorgan Chase \& Co. This paper is not a product of the Research Department of JPMorgan Chase \& Co. or its affiliates. Neither JPMorgan Chase \& Co. nor any of its affiliates makes any explicit or implied representation or warranty and none of them accept any liability in connection with this paper, including, without limitation, with respect to the completeness, accuracy, or reliability of the information contained herein and the potential legal, compliance, tax, or accounting effects thereof. This document is not intended as investment research or investment advice, or as a recommendation, offer, or solicitation for the purchase or sale of any security, financial instrument, financial product or service, or to be used in any way for evaluating the merits of participating in any transaction.

\end{document}